\documentclass{article}

\pdfoutput=1

\usepackage{fullpage}
\usepackage{amsmath,amsthm}
\usepackage{amssymb,mathtools}
\usepackage{graphicx,subcaption}
\usepackage{ulem}
\usepackage{xcolor}
\usepackage{enumitem}
\usepackage{times}
\usepackage{bm}
\usepackage[round,sort]{natbib}
\usepackage[plain,noend]{algorithm2e}
\usepackage{tabularx}

\newtheorem{theorem}{Theorem}
\newtheorem{corollary}{Corollary}[theorem]
\newtheorem{lemma}{Lemma}

\def\V{\textup{\textsf{V}}}
\def\Y{\textup{\textsf{Y}}}
\def\M{\textup{\textsf{M}}}
\def\Vx{\includegraphics[height=1.6ex,keepaspectratio=true]{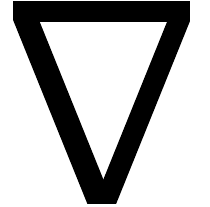}}
\def\rightM{\includegraphics[height=1.6ex,keepaspectratio=true]{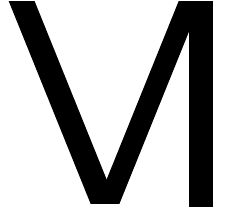}}
\def\leftM{\includegraphics[height=1.6ex,keepaspectratio=true]{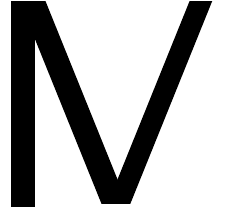}}
\def\rightlongM{\includegraphics[height=1.6ex,keepaspectratio=true]{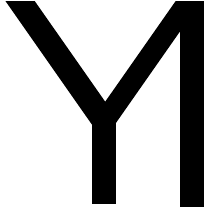}}
\def\leftlongM{\includegraphics[height=1.6ex,keepaspectratio=true]{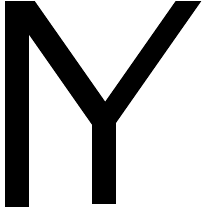}}
\def\longM{\includegraphics[height=1.6ex,keepaspectratio=true]{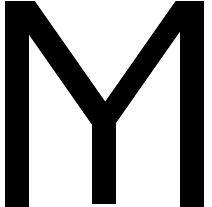}}

\def\P{\mathrm{P}}

\newcommand{\add}[1]{{\color{black}{}#1}}

\begin{document}

\title{The Magnitude and Direction of Collider Bias for Binary Variables}

\author{Trang Quynh Nguyen, Allan Dafoe, Elizabeth L. Ogburn}

\date{October 5, 2018 \\~\\ (accepted by Epidemiologic Methods, Jan 14, 2019)}

\maketitle

\begin{abstract}
Suppose we are interested in the effect of variable $X$ on variable $Y$.
If $X$ and $Y$ both influence, or are associated with variables that influence, a common outcome, called a {\it collider}, then conditioning on the collider (or on a variable influenced by the collider -- its ``child'') induces a spurious association between $X$ and $Y$, which is known as collider bias. 
Characterizing the magnitude and direction of collider bias is crucial for understanding the implications of selection bias and for adjudicating decisions about whether to control for variables that are known to be associated with both exposure and outcome but could be either confounders or colliders.  
Considering a class of situations where all variables are binary, and where $X$ and $Y$ either are, or are respectively influenced by, two marginally independent causes of a collider, we derive collider bias that results from (i) conditioning on specific levels of the collider or its child\add{ (on the covariance, risk difference, and in two cases odds ratio, scales)}, or (ii) linear regression adjustment for, the collider or its child. 
We also derive simple conditions that determine the sign of such bias.\\

{\it Key words}: Bias; Collider; Collider bias; Collider-stratification bias; Selection bias; M-bias.

\end{abstract}

\section{Introduction}

{\it Collider bias} is bias in a measure of association between two variables due to conditioning on a common outcome (a {\it collider}) of the two variables or of their causes. 
In various contexts collider bias is also known as M-bias, selection bias, endogenous selection bias, or Berkson's bias \citep{Elwert2014}. 
Here we use {\it collider bias} as a general term that includes bias due to stratifying or subsetting on, as well as bias due to statistical adjustment for, a collider or a variable influenced by a collider. 
Examples of collider bias in empirical research abound; for examples in epidemiology see papers on this topic by \citet{Greenland2003,Cole2010a,Hernan2004} and for examples in sociology, see a review by \citet{Elwert2014}.

Collider bias is at the heart of a long-standing controversy in the literature on estimating causal effects using observational data. 
If a treatment or exposure, $X$, and an outcome, $Y$, both influence (or are respectively associated with two variables that both influence) another variable, called a collider, then conditioning on the collider (or on a variable influenced by it) induces a spurious association between $X$ and $Y$; this spurious association is known as collider bias. 
When selecting covariates to control for, many statisticians and applied researchers hew to the {\it pretreatment criterion} \citep{vanderweele2011new}, which stipulates that all available baseline covariates should be controlled for.  
This is the recommendation found, for example, in \citet{rosenbaum2002observational}, \citet{rubin2009should} and \citet{rubin2009author}; it is based on the rationale that any pretreatment (or pre-exposure) covariate is a potential confounder of the $X$-$Y$ relation thus controlling for all such potential confounders maximizes the chance that no unmeasured confounding remains to bias causal effect estimates.  
In contrast to the pretreatment criterion, other approaches to controlling for confounding attempt to differentiate between pretreatment confounders and pretreatment colliders, and to control for the former but not the latter.  
This is the approach advocated by, for example, \citet{rothman2008modern, Pearl2009, hernan2010causal} and \citet{vanderweele2011new}.  
After the publication of a lengthy back-and-forth exchange debating the foundations and merits of these two approaches \citep{rubin2009author, rubin2009should, shrier2008letter, sjolander2009propensity, pearl2009letter, pearl2009myth}, researchers have attempted to mediate between the two schools of thought \citep{vanderweele2011new, Ding2015}.  
While it is widely accepted that conditioning on a pretreatment collider can introduce bias, it is still a matter for debate whether this bias is significant enough to undermine the seductively simple pretreatment criterion.

While the debate described above is about pretreatment colliders, there is also much interest in and a literature surrounding post-treatment colliders, especially in the field of epidemiology.  
Colliders that are influenced by treatment are implicated in several ``paradoxes" in the epidemiology literature \citep{porta2015current}, such as the ``obesity paradox," where selection on the basis of diabetes status creates a spurious negative correlation between obesity and mortality \citep{banack2013obesity}, and the ``birth weight paradox," where stratifying on birth weight creates a spurious negative association between a risk factor and neonatal mortality \citep{whitcomb2009}.  
\add{Collider bias is also a problem in mediation analysis if there is a post-treatment confounder of the mediator -- the natural direct and indirect effects are unidentified because an attempt to isolate them would need to condition on the mediator, but conditioning on the mediator would induce collider bias that distorts the effects \citep{TchetgenTchetgen2014e}.}

Methodologists have provided insights into bias due to conditioning on pre-treatment \citep{Pearl2013,Ding2015,Greenland2003} and post-treatment \citep{jiang2016directions,Greenland2003,Pearl2013} colliders. 
\add{\cite{vanderweele2007directed} determine the sign of the conditional covariance of the causes of a collider when the causal structure is a sufficient component cause \citep{rothman2008modern} model.}
\cite{Pearl2013} and \cite{Ding2015} derive collider bias under the assumption of linear models, while \cite{Greenland2003} relies on assumptions of no interactions on the odds ratio scale.
Without making such assumptions, \citet{jiang2016directions} determine the sign, but do not quantify the magnitude, of collider bias.
\add{\citet{Shahar2017} identify conditions in which marginally independent causes of a collider are independent or dependent conditional on the collider.}
Other work provides criteria, based on graphical models, for qualitatively assessing whether conditioning on variables associated with colliders may induce bias \citep{greenland2011adjustments} and, when all variables are jointly Gaussian, for partially ordering the bias induced by conditioning on different variables associated with a collider \citep{chaudhuri2002using}. 
In this paper we provide analytic results on the direction and magnitude of collider bias in settings with both pre- and post-treatment colliders, and with binary treatment, collider, and outcome variables.

We focus our attention on settings where the putative exposure and outcome are marginally independent, i.e., settings under the null hypothesis of no exposure-outcome relationship. 
We consider situations where the exposure and outcome either influence the collider or are influenced by causes of the collider. We derive precise formulas for bias of the exposure-outcome effect measured on the covariance, risk difference, and in some cases odds ratio scale, due to conditioning on specific levels of the collider (or a variable it influences); and bias in risk difference estimated by linear regression due to adjustment for the collider (or the variable it influences).
We discuss conditions under which collider bias is negative, positive, or zero, and point out how collider bias in each structure relates to collider bias in the simpler structure(s) embedded in it.
To the best of our knowledge, the analytic results we present here (except one result \add{in section \ref{subsec:Vx}} that we include for completeness) are novel results quantifying collider bias in this class of binary variable structures without any simplifying assumptions about the effects of the two causes on the collider.
In settings with pre-treatment colliders, \add{our results help adjudicate the debate described above}: using the associations among the variables, we determine the direction and magnitude of bias due to controlling for a potential collider, which can be weighed against the comparable risks of failing to control for a potential confounder.
In settings with post-treatment colliders, our paper is in dialogue with the papers mentioned above. It provides additional insight into the settings considered by \cite{jiang2016directions}, and relaxes the assumptions of \cite{Greenland2003}.

The remainder of this paper is organized as follows.
In Section \ref{sec:notation} we introduce notation and terminology, including definitions of measures of collider bias.
Sections \ref{sec:stratum} and \ref{sec:lm} present analytic results for collider bias, and discuss the direction of such bias.
Section \ref{sec:stratum} covers collider bias conditioning on specific levels of the collider (or its child) in all the binary variable structures in Fig. \ref{fig:1}, while Section \ref{sec:lm} covers collider bias due to linear regression adjustment for the collider (or its child) in the structures where the collider's parents are marginally independent. Section \ref{sec:discussion} concludes with a discussion.
\add{Proofs of all results presented in the paper are provided in the Appendix.}

\begin{figure}
\caption{Binary variable structures considered in this paper}
\centering
\begin{subfigure}[b]{.2\textwidth}
\includegraphics[width=\textwidth]{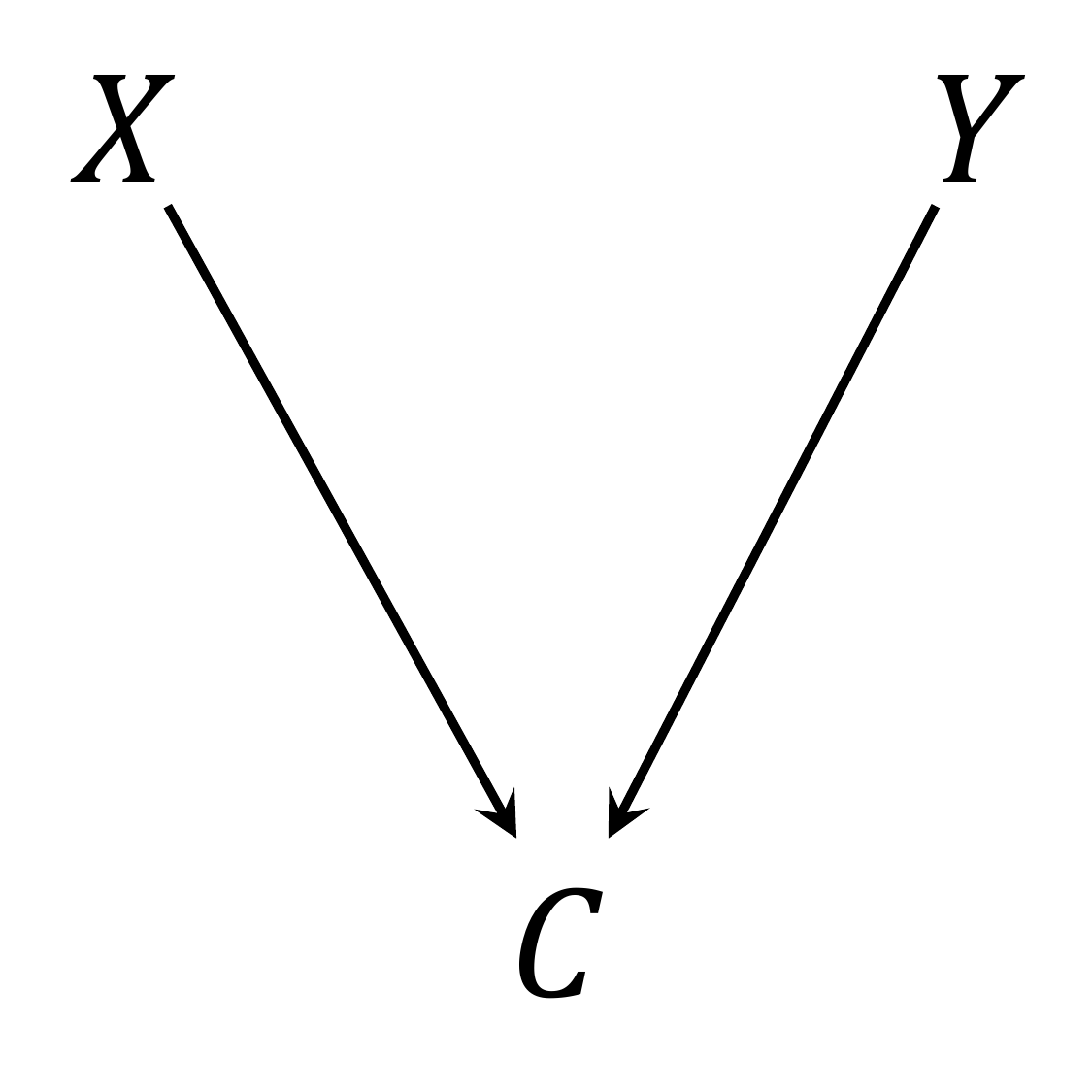}
\caption*{~~~~\V~structure}
\end{subfigure}
\begin{subfigure}[b]{.13\textwidth}
\includegraphics[width=\textwidth]{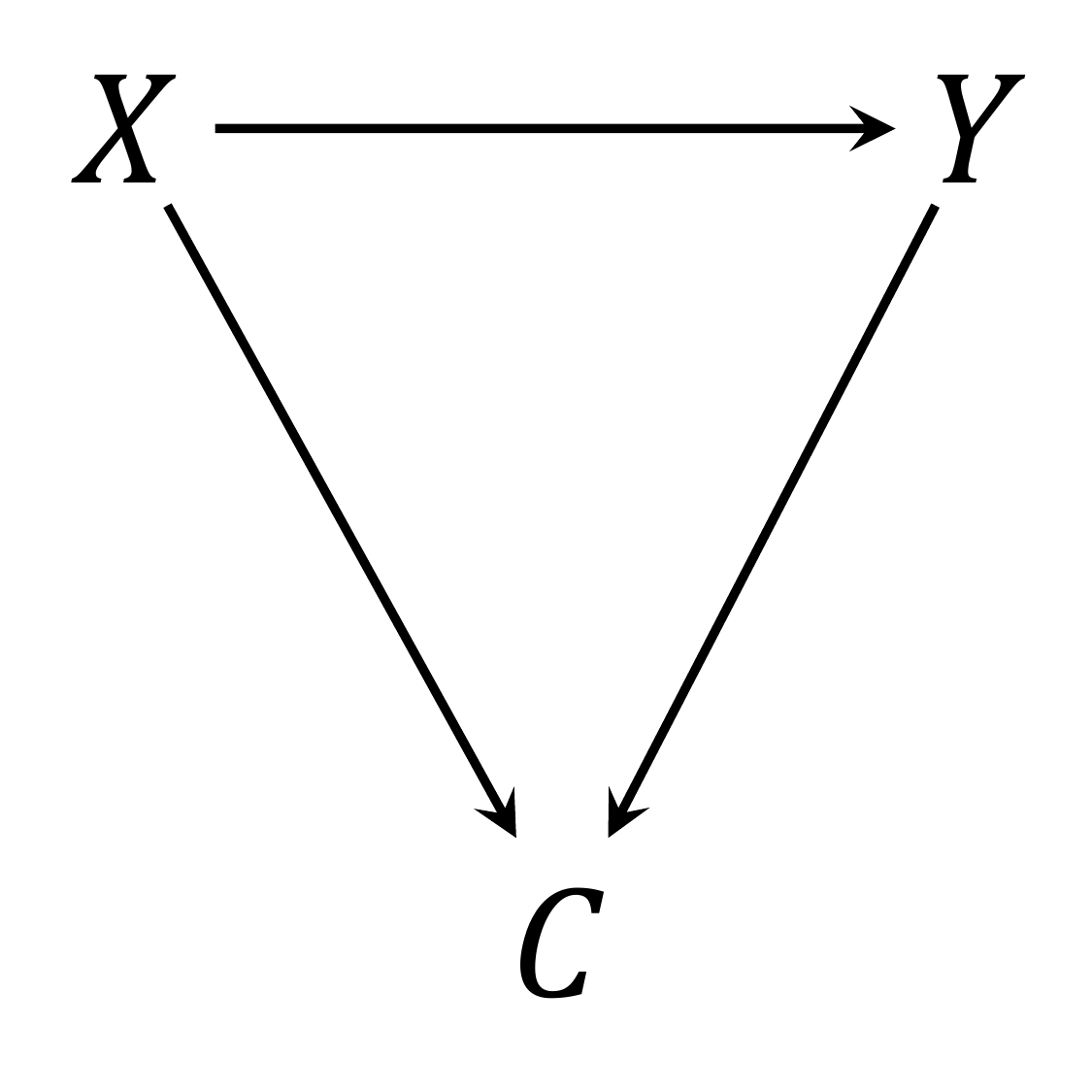}
\caption*{\footnotesize\Vx~structure}
\end{subfigure}
~~~~~~~~~~~~~~~~~~~
\begin{subfigure}[b]{.2\textwidth}
\includegraphics[width=\textwidth]{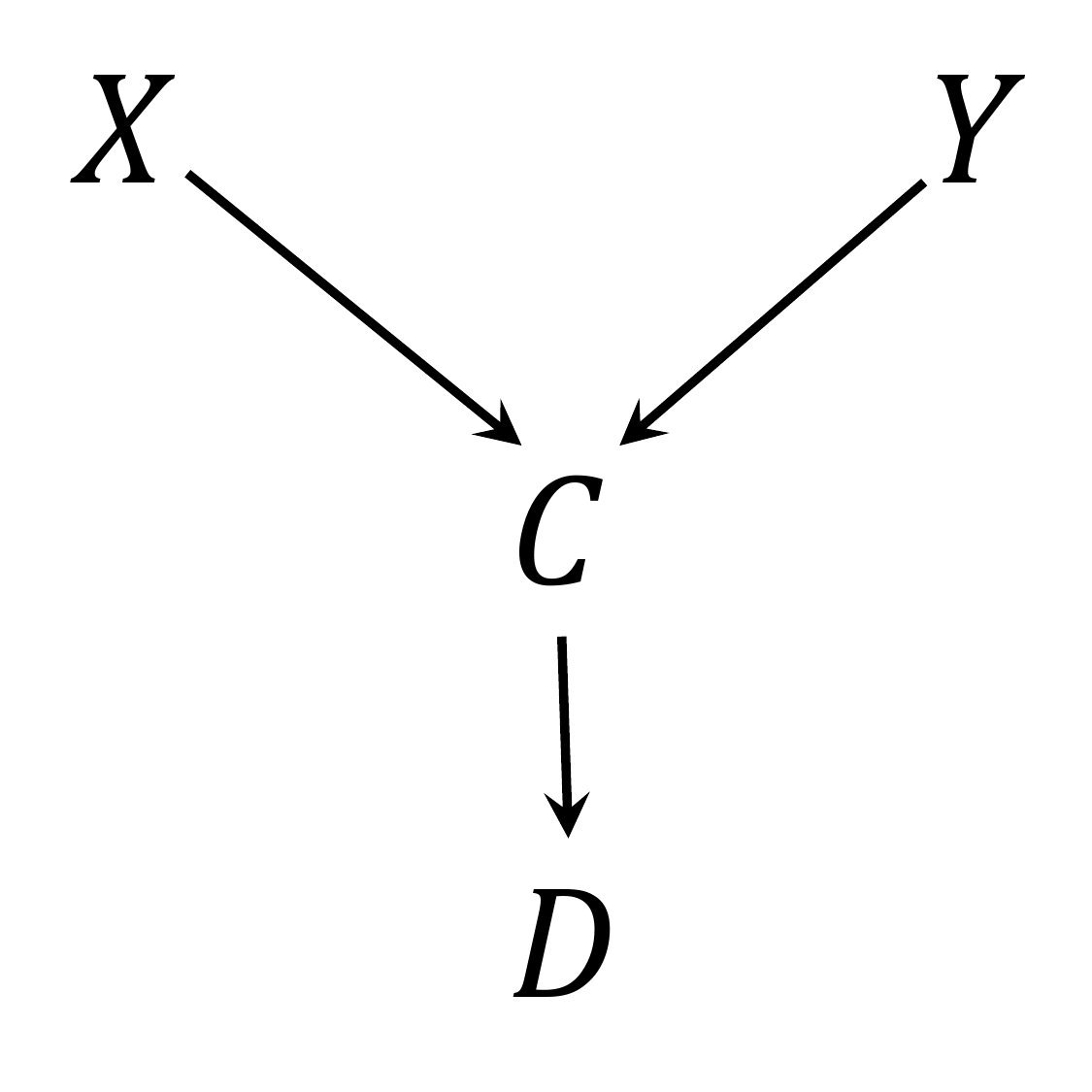}
\caption*{~~~~\Y~structure}
\end{subfigure}
\\
~\\
~\\
\begin{subfigure}[]{.13\textwidth}
\includegraphics[width=\textwidth]{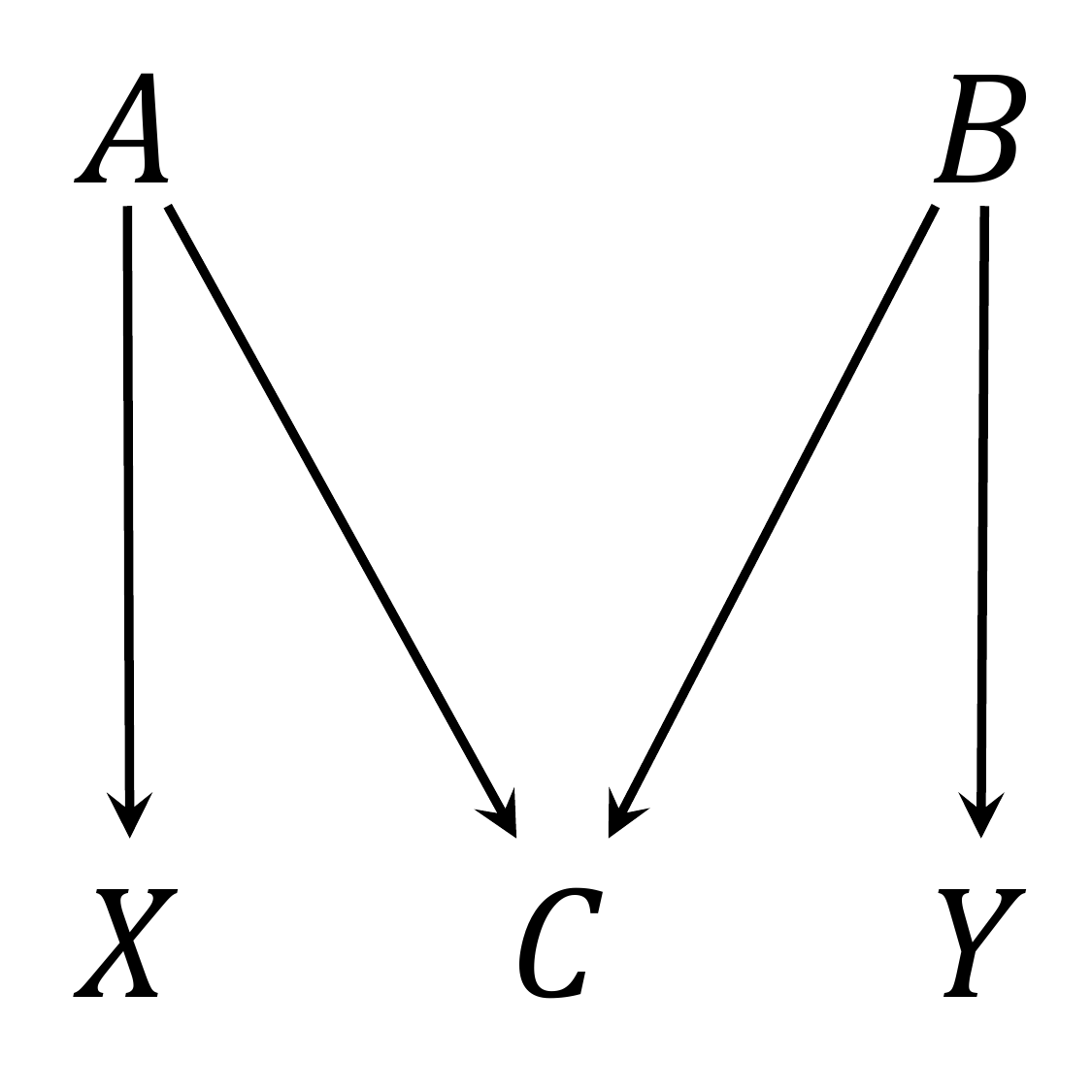}
\caption*{\footnotesize\M~structure}
\end{subfigure}
~~~
\begin{subfigure}[]{.13\textwidth}
\includegraphics[width=\textwidth]{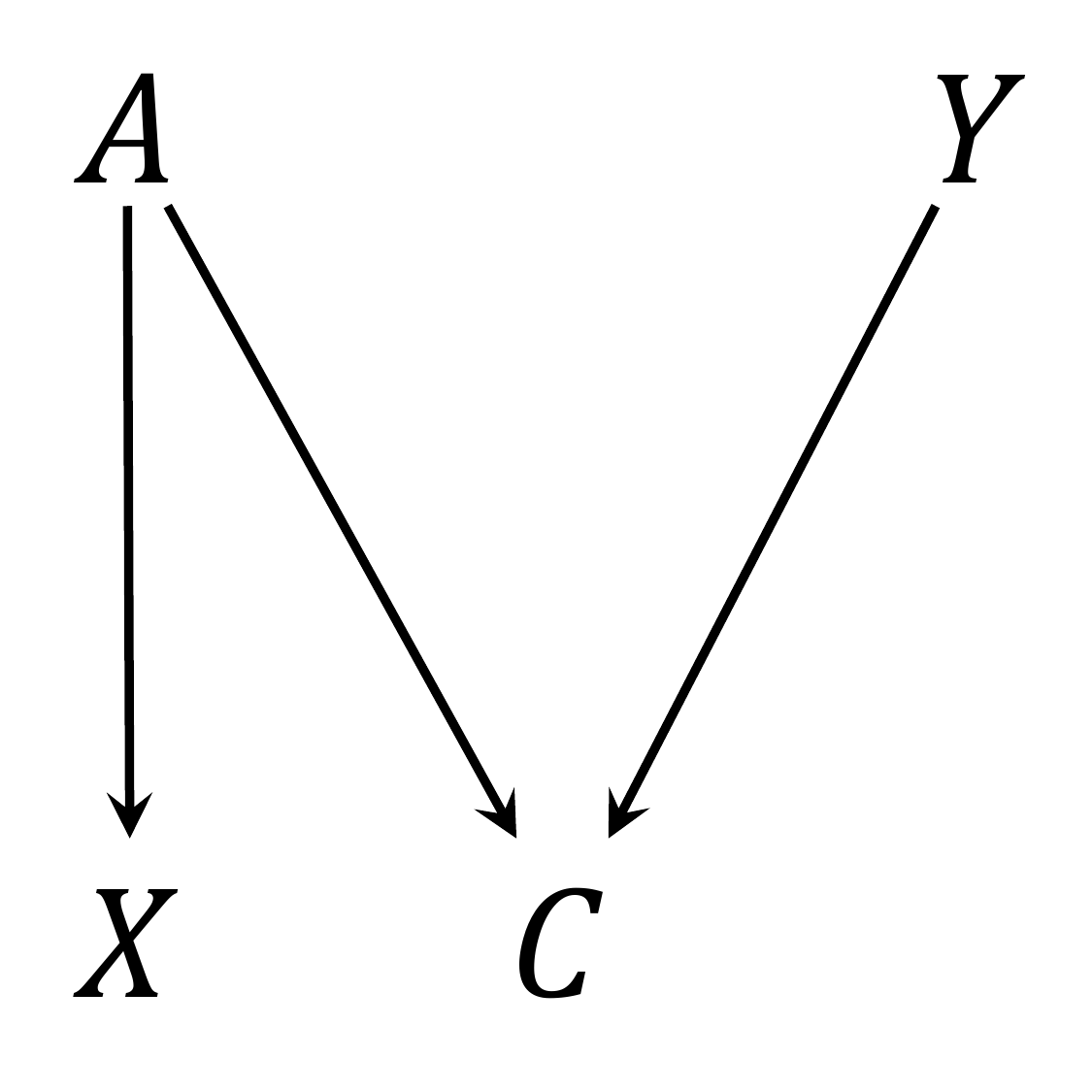}
\caption*{\footnotesize\leftM~structure}
\end{subfigure}
~~~
\begin{subfigure}[]{.13\textwidth}
\includegraphics[width=\textwidth]{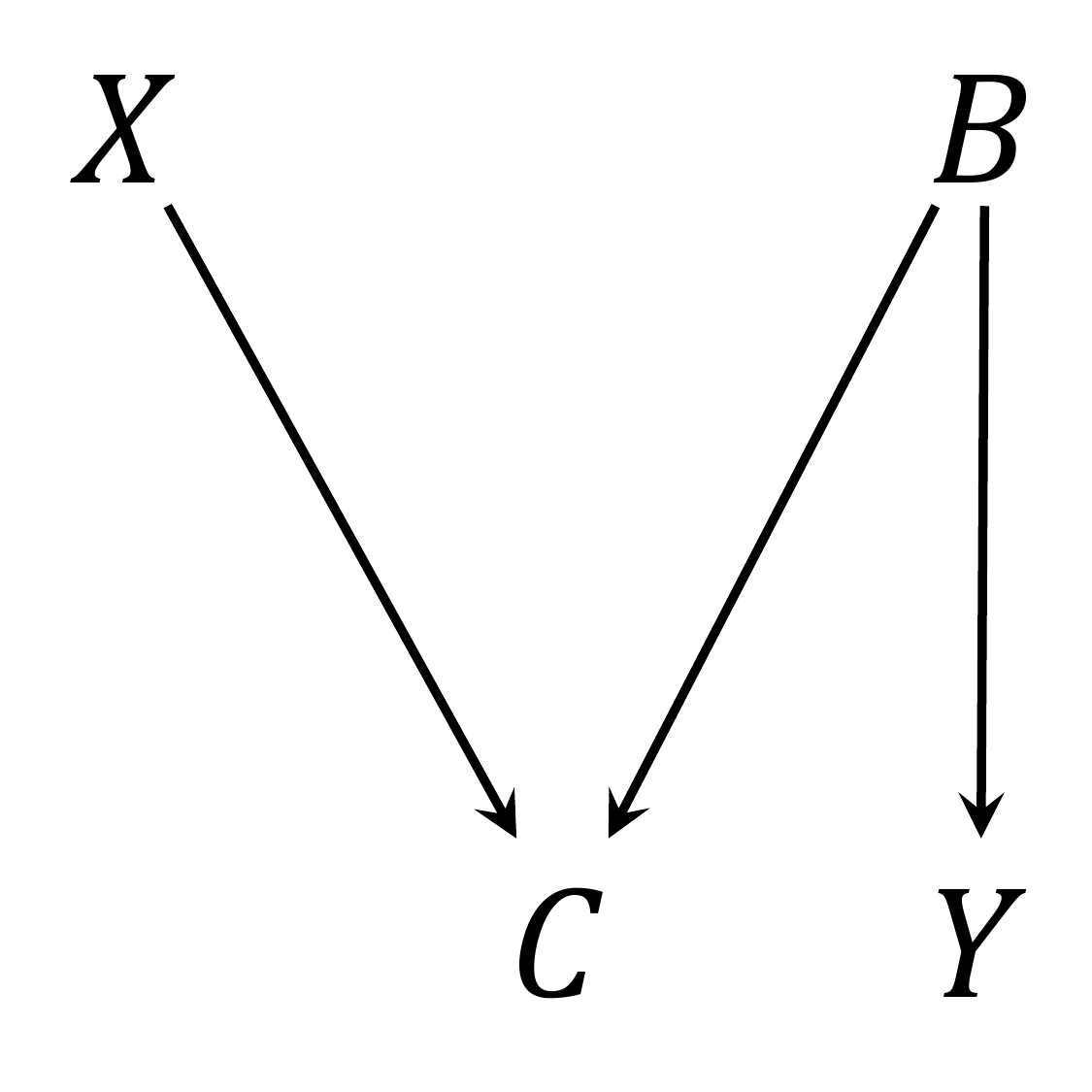}
\caption*{\footnotesize\rightM~structure}
\end{subfigure}
~~~~
\begin{subfigure}[]{.13\textwidth}
\includegraphics[width=\textwidth]{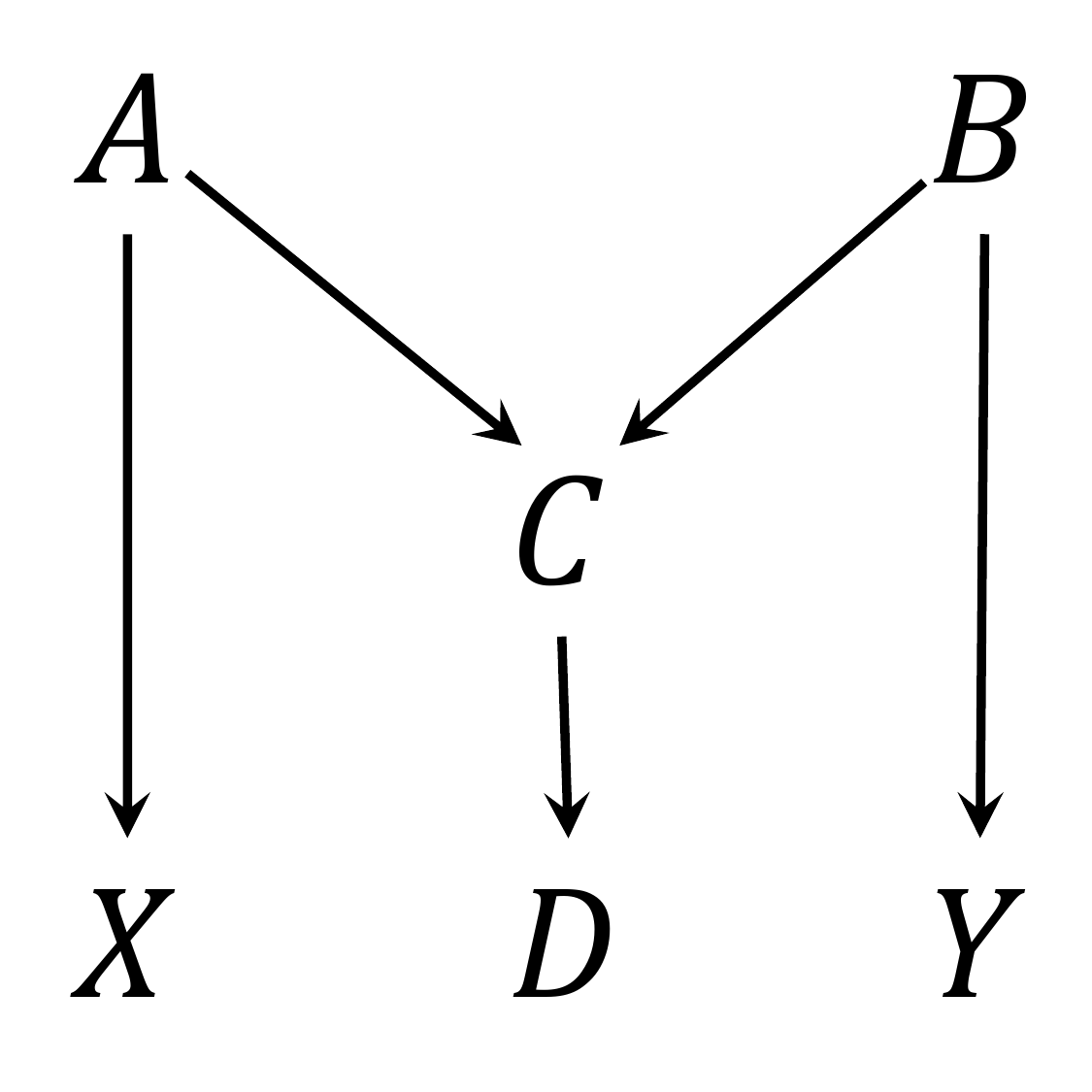}
\caption*{\footnotesize\longM~structure}
\end{subfigure}
~~~
\begin{subfigure}[]{.13\textwidth}
\includegraphics[width=\textwidth]{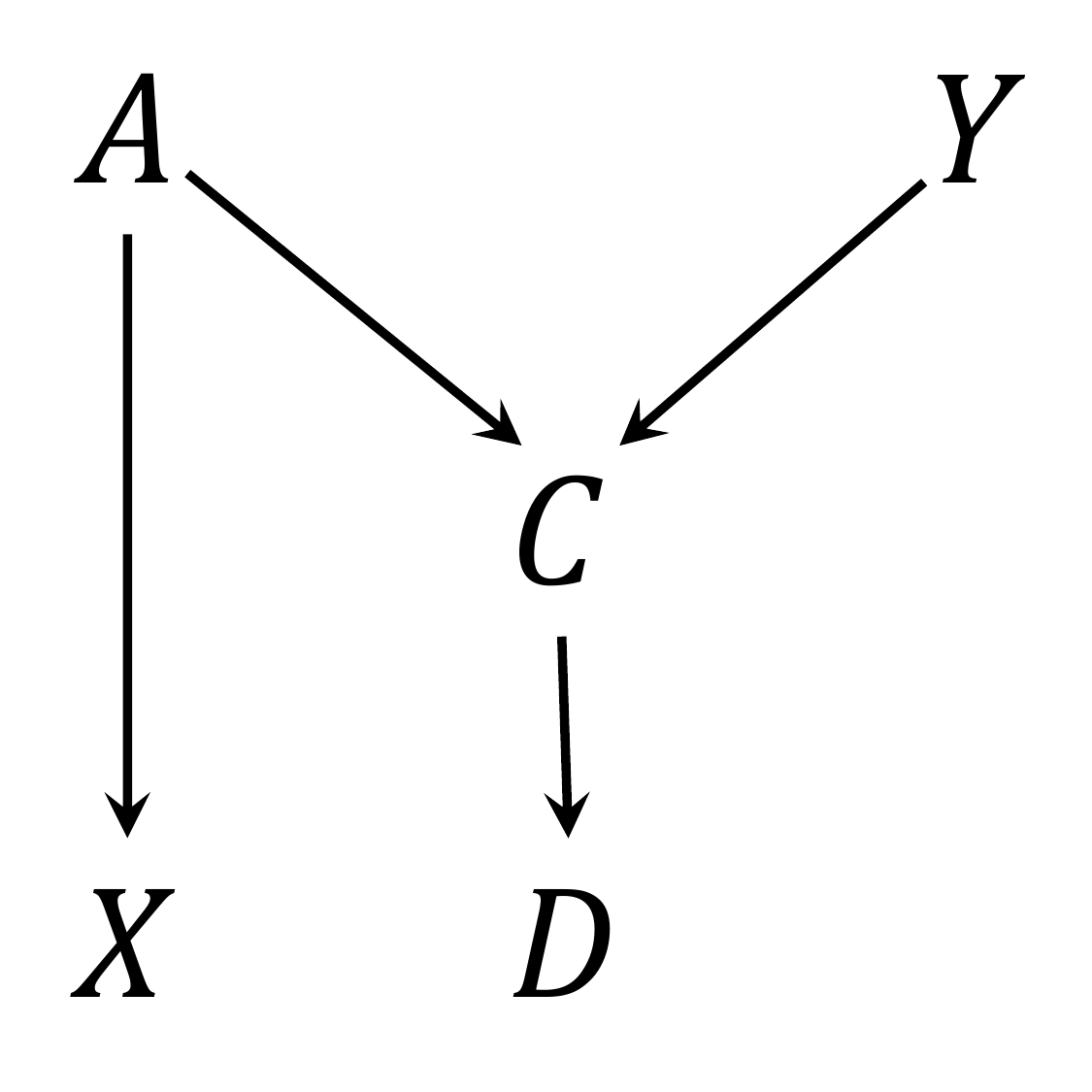}
\caption*{\footnotesize\leftlongM~structure}
\end{subfigure}
~~~
\begin{subfigure}[]{.13\textwidth}
\includegraphics[width=\textwidth]{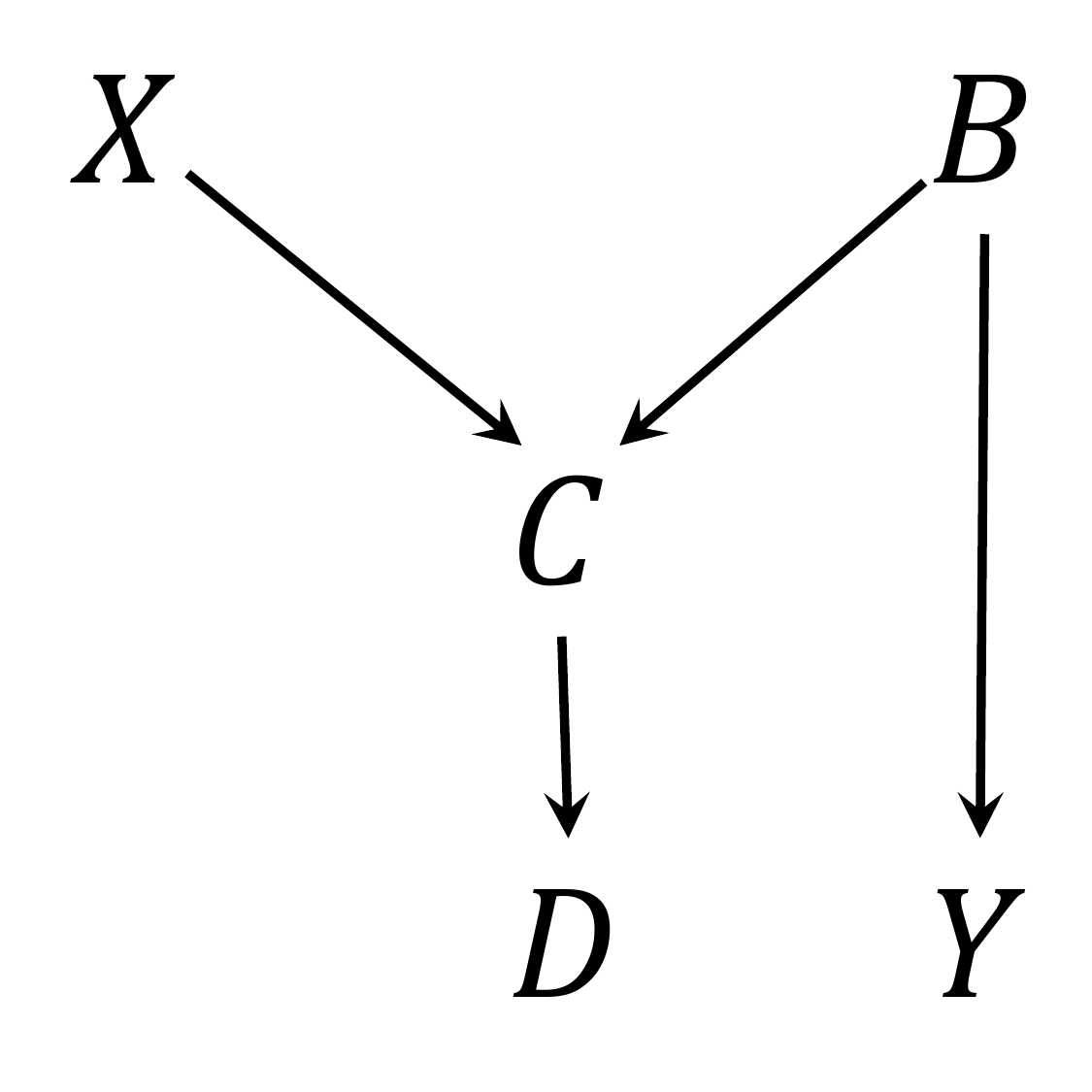}
\caption*{\footnotesize\rightlongM~structure}
\end{subfigure}
\label{fig:1}
\end{figure}

\section{Notations, definitions, terms and abbreviations} \label{sec:notation}

Figure \ref{fig:1} depicts the relationships among $X$, $Y$ and a collider $C$, and in some cases a variable $D$ influenced by $C$ that we consider in this paper.  The diagrams in Figure \ref{fig:1} are causal directed acyclic graphs (DAGs; \citealp{Pearl2009}).  A causal DAG consists of nodes that represent variables and arrows that represent causal effects; it includes all common causes of any pair of variables in the DAG.
If one variable is a cause of another (either directly, or through intermediate variables), the former is called an {\it ancestor} of the latter, and the latter a {\it descendant} of the former.
An ancestor and a descendant directly connected by a single arrow are referred to as {\it parent} and {\it child}, respectively.
A child with more than one parent is a collider between its parents.

In each of the DAGs in Fig. \ref{fig:1}, $C$ is a collider between two variables, one of which is denoted either $X$ or $A$ and the other either $Y$ or $B$, depending on the structure.
In all these settings, the relationship between $X$ and $Y$ is of interest to a scientist, who considers $X$ the exposure and $Y$ the outcome.
Depending on the setting, each of these two variables is either a parent of $C$ or another child of a parent of $C$.
In four of the DAGs in Figure \ref{fig:1}, $C$ has a child $D$ that is not influenced by any other variables in the DAG.
We use sans serif letters to refer to three structures in Figure \ref{fig:1}: \V, \Y~and \M.
We refer to the other structures using letter-like symbols that mimic how these structures are drawn: \Vx~(`upside-down triangle'), \leftM~(`left-sided M'), \rightM~(`right-sided M'), \longM~(`long M'), \leftlongM~(`left-sided long M'), and \rightlongM~(`right-sided long M').
The \V~ structure is the simplest case, in which $X$ and $Y$ affect a post-treatment collider directly.
The \M~structure, which gives rise to the term {\it M-bias}, in which $C$ may be a post-treatment or pre-treatment variable, is perhaps the most well-known of these structures, having been at the heart of the debate described in the introduction.
Several ``paradoxes'' in epidemiology involve the \rightM~structure or a structure with an embedded \rightM~substructure.


We consider collider bias conditioning on $C$ in the \V, \Vx, \M, \leftM, and \rightM~ structures, and conditioning on $D$ in the \Y, \longM, \leftlongM, and \rightlongM~structures.
We define collider bias generally as \textit{the departure of the conditional association of $X$ and $Y$ from their marginal association}.
This departure is characterized either as a difference or a ratio, depending on the measure of association (or effect scale) used to characterize the $X$-$Y$ relationship.
For example, collider bias conditioning on $C=c$ may be measured by
a difference between covariances,
$$\textup{cov}(X,Y\mid C=c)-\textup{cov}(X,Y),$$
a difference between risk differences (\textsc{rd}s),
$$\left\{\begin{matrix*}[l]
\mathrm{P}(Y=1\mid X=1,C=c)-\\
\mathrm{P}(Y=1\mid X=0,C=c)
\end{matrix*}\right\}-
\left\{\begin{matrix*}[l]
\mathrm{P}(Y=1\mid X=1)-\\
\mathrm{P}(Y=1\mid X=0)
\end{matrix*}\right\},$$
a ratio of risk ratios (\textsc{rr}s),
$$\frac{\mathrm{P}(Y=1\mid X=1,C=c)}{\mathrm{P}(Y=1\mid X=0,C=c)}
\bigg/\frac{\mathrm{P}(Y=1\mid X=1)}{\mathrm{P}(Y=1\mid X=0)},$$
or a ratio of odds ratios (\textsc{or}s),
$$\frac{\P(Y=1\mid X=1,C=c)/\P(Y=0\mid X=1,C=c)}{\P(Y=1\mid X=0,C=c)/\P(Y=0\mid X=0,C=c)}
\bigg/
\frac{\P(Y=1\mid X=1)/\P(Y=0\mid X=1)}{\P(Y=1\mid X=0)/\P(Y=0\mid X=0)}.$$
%
We also consider collider bias due to linear regression adjustment for $C$ (or $D$).
This bias is defined as the difference between the adjusted \textsc{rd} of $Y$ comparing the two levels of $X$, represented by the coefficient of $X$ in the linear model regressing $Y$ on $X$ and $C$ (or $D$), and the marginal \textsc{rd}. Note that \textit{when $X$ and $Y$ are marginally independent} (all structures in Fig. \ref{fig:1} except \Vx), \textit{these bias measures reduce to the conditional/adjusted association}.
Although all these measures of bias can be derived, our interest is in expressions that provide insights into the direction and magnitude of bias. To this end, we restrict our attention to the covariance, \textsc{rd}, linear regression, and \textsc{or} measures for the \V~structure; to the \textsc{or} scale for the \Vx~structure; and to the covariance, \textsc{rd}, and linear regression measures for the other structures.

When referring to collider bias in a specific structure, we name the bias after the structure: \V-bias, \M-bias, \leftM-bias, etc.
When referencing a specific effect scale and/or a type of conditioning, we add such information after the bias name, e.g., \Y-bias$(D=0,\text{cov})$ and \M-bias(\textsc{lm}) refer to \Y-bias conditioning on $D=0$ on the covariance effect scale and \M-bias due to linear regression adjustment for $C$, respectively.

As will be shown, collider bias measures are often complex functions of the marginal and conditional probabilities of variables in the structure.
To improve clarity, we introduce a simple shorthand system for some of these probabilities.
For an exogenous variable, a marginal probability is abbreviated using $p$ with an index, e.g., $\mathrm{P}(A=1)$ is abbreviated to $p_{A=1}$.
For an endogenous variable, a conditional probability conditioning on all its parents is similarly abbreviated, with the conditioning event added to the index and the parents implied, e.g., $\mathrm{P}(D=1\mid C=0)$ becomes $p_{D=1\mid0}$.
For the collider, the conditioning index includes two values, the first referring to the parent on the left hand side, the second the parent on the right hand side, e.g., $\mathrm{P}(C=0\mid X=1,Y=0)$ in the \V~structure and $\mathrm{P}(C=0\mid A=1,B=0)$ in the \M~structure are abbreviated to $p_{C=0\mid10}$.
In addition, when referring to an eligible but not specific value of a variable, we use lower case notation (e.g., $d$ representing a value of variable $D$), and abbreviate the index further, e.g., $p_{D=d\mid1}$ becomes $p_{d\mid1}$, and $p_{C=c\mid11}$ becomes $p_{c\mid11}$.

When discussing the direction of collider bias, we will use \textit{sign of bias} language.
By `zero' or `no' bias, we mean that a conditional covariance, \textsc{rd}, or \textsc{or}, is equal to its marginal counterpart, thus the measure of collider bias is equal to 0 on the covariance or \textsc{rd} effect scale, and equal 1 on the \textsc{or} effect scale.
By `positive' (`negative') bias, we mean that a conditional covariance, \textsc{rd}, or \textsc{or}, is greater (smaller) than its marginal counterpart, thus the measure of collider bias is greater (smaller) than 0 on the covariance or \textsc{rd} effect scale, and greater (smaller) than 1 on the \textsc{or} effect scale.

\section{Bias due to conditioning on a specific level of a collider or its child}\label{sec:stratum}

\subsection{Bias due to conditioning on a specific level of the collider in the \V~structure}\label{subsec:V}

\add{There are many examples of \V-bias conditioning on a level of $C$.
In a simple didactic example \citep{Cole2010a}, some of the people who attended a party later developed a fever ($C$) because they either ate tainted sandwiches ($X$) at the party or had contracted the flu ($Y$) prior to the party, or both. The flu and the sandwiches were truly unrelated, but among those with fever there was a negative association between having the flu and eating the tainted sandwiches. The explanation is intuitive: a person who did not eat tainted sandwiches must have had the flu to have the fever and vice versa.
A set of more complex examples was presented by
\cite{Berkson1946}, where $X$ is diabetes and $Y$ is cholecystitis (inflammation of the gallbladder). Both of these conditions may lead to hospitalization, thereby increasing the probability of being selected into a hospital-based sample ($C$). The observed association between diabetes and cholecystitis in this sample may be either positive or negative, depending on how the sample is selected, even though these two conditions are independent in the general population.
We will touch on these examples as we present and discuss analytical results regarding the direction and magnitude of \V-bias.}

\medskip

The following theorem provides formulas for \V-bias conditioning on a specific level of $C$. On the covariance or \textsc{rd} effect scale, such bias is dependent on the marginal probabilities of $X$ and $Y$ and the conditional probabilities of $C$ given $X$ and $Y$; on the \textsc{or} effect scale, it is a function of the conditional probabilities of $C$ only.

\begin{theorem}[$C$-specific \V-bias theorem]\label{thm:V}
\V-bias conditioning on $C=c$, for $c\in\{0,1\}$, is given by the following expressions:
\begin{align*}
    \V\textup{-bias}(C=c,\mathrm{cov})&=\frac{p_{X=1}p_{X=0}p_{Y=1}p_{Y=0}}{\{\mathrm{P}(C=c)\}^2}\cdot(p_{c\mid00}p_{c\mid11}-p_{c\mid10}p_{c\mid01}),\\
    \V\textup{-bias}(C=c,\textsc{rd})&=\frac{p_{Y=1}p_{Y=0}}
    {
    (p_{Y=1}p_{c\mid11}+p_{Y=0}p_{c\mid10})
    (p_{Y=1}p_{c\mid01}+p_{Y=0}p_{c\mid00})
    }\cdot(p_{c\mid00}p_{c\mid11}-p_{c\mid10}p_{c\mid01}),\\
    \V\textup{-bias}(C=c,\textsc{or})&=\frac{p_{c\mid00}p_{c\mid11}}{p_{c\mid10}p_{c\mid01}},
\end{align*}
and for completeness, $\mathrm{P}(C=c)=p_{X=1}p_{Y=1}p_{c\mid11}+p_{X=0}p_{Y=1}p_{c\mid01}+p_{X=1}p_{Y=0}p_{c\mid10}+p_{X=0}p_{Y=0}p_{c\mid00}.$
\end{theorem}

\medskip

An insight from Theorem \ref{thm:V} is that \V-bias conditioning on $C=c$ is of the same sign as the sign of function $g(c)=p_{c\mid00}p_{c\mid11}-p_{c\mid10}p_{c\mid01}$.
In the special case where either $p_{c\mid00}$ or $p_{c\mid11}$ is zero and either $p_{c\mid10}$ or $p_{c\mid01}$ is zero, $g(0)$ will be zero. Outside of this special case, generally $g(c)=0$ when
$$\frac{p_{c\mid11}}{p_{c\mid00}}=\frac{p_{c\mid10}}{p_{c\mid00}}\cdot\frac{p_{c\mid01}}{p_{c\mid00}}.$$
As the ratios above represent effects of $X$ and $Y$ on the probability of $C=c$ on the \textsc{rr} scale, this expression reads: the effect of the combination of two actions, shifting $X$ to 1 and shifting $Y$ to 1 (from a base of $(X=0,Y=0)$), equals the combination of the individual actions' effects, that is, there is no interaction. Hence, \V-bias conditioning on $C=c$ is zero when $X$ and $Y$ do not interact in their effects on the probability of $C=c$ on the \textsc{rr} scale\add{; this is consistent with \citeauthor{Shahar2017}'s (\citeyear{Shahar2017}) finding}.
In addition, \V-bias conditioning on $C=c$ is positive when $X$ and $Y$ interact in their effects on the probability of $C=c$ on the \textsc{rr} scale such that
$$\frac{p_{c\mid11}}{p_{c\mid00}}>\frac{p_{c\mid10}}{p_{c\mid00}}\cdot\frac{p_{c\mid01}}{p_{c\mid00}},$$
and is negative when the interaction is in the opposite direction, i.e.,
$$\frac{p_{c\mid11}}{p_{c\mid00}}<\frac{p_{c\mid10}}{p_{c\mid00}}\cdot\frac{p_{c\mid01}}{p_{c\mid00}}.$$
%
%
The sign of \V-bias conditioning on a specific level of $C$ is not dependent on the choice of effect scale.
(This is true generally for collider bias conditioning on a specific level of the collider or its child, provided that the collider's parents are marginally independent -- true in all structures in Fig. \ref{fig:1} except \Vx.)

\add{With the \textsc{or} scale, another way to see this is to re-express the \V-bias$(C=c,\textsc{or})$ formula in Theorem \ref{thm:V} as $\displaystyle\frac{p_{c\mid11}}{p_{c\mid01}}\big/\frac{p_{c\mid10}}{p_{c\mid00}}$, the ratio of the \textsc{rr}-scale effects of $X$ on the probability of $C=c$ conditional on the two levels of $Y$ (or equivalently as $\displaystyle\frac{p_{c\mid11}}{p_{c\mid10}}\big/\frac{p_{c\mid01}}{p_{c\mid00}}$, the ratio of the \textsc{rr}-scale effects of $Y$ on the probability of $C=c$ conditional on the two levels of $X$), which measures the interaction effect between $X$ and $Y$ on the probability of $C=c$. Clearly, the bias factor is equal to 1 in the absence of such interaction, and is not equal to 1 in the presence of such interaction.}

\medskip

Theorem \ref{thm:V} implies the two corollaries below.

\begin{corollary}[$C$-specific \V-bias corollary 1]\label{crlry:V1}
In the \V~structure, if $X$ has positive effects on $C$ at both levels of $Y$ (i.e., $p_{C=1\mid10}>p_{C=1\mid00}$, $p_{C=1\mid11}>p_{C=1\mid01}$) and $Y$ has positive effects on $C$ at both levels of $X$ (i.e., $p_{C=1\mid01}>p_{C=1\mid00}$, $p_{C=1\mid11}>p_{C=1\mid10}$), or alternatively, if $X$ has negative effects on $C$ at both levels of $Y$ and $Y$ has negative effects on $C$ at both levels of $X$, then \V-bias is always negative for at least one level of $C$.
On the other hand, if $X$ has positive effects on $C$ at both levels of $Y$ and $Y$ has negative effects on $C$ at both levels of $X$, or vice versa, then \V-bias is always positive for at least one level of $C$.
\end{corollary}

In simpler (but less precise) terms, the conditions for \V-bias to be always negative for at least one level of $C$ stated in Corollary \ref{crlry:V1} may be thought of as ``$X$ and $Y$ influence $C$ in the same direction'', and the condition for \V-bias to be always positive for at least one level of $C$, as ``$X$ and $Y$ influence $C$ in opposite directions.'' The clarity provided by this corollary is that the influences on $C$ referred to here are conditional, i.e., the influences of $X$ on $C$ given levels of $Y$ and the influences of $Y$ on $C$ given levels of $X$.

\begin{corollary}[$C$-specific \V-bias corrolary 2]\label{crlry:V2}
In the \V~structure, if $X$ and $Y$ interact qualitatively on $C$, i.e., the effects $X$ on $C$ are of opposite signs across the two levels of $Y$ (e.g., $p_{C=1\mid10}>p_{C=1\mid00}$ and $p_{C=1\mid11}<p_{C=1\mid01}$) and/or the effects of $Y$ on $C$ are of opposite signs across the two levels of $X$, then \V-bias is negative for one level of $C$ and positive for the other level of $C$.
\end{corollary}

\subsection{An extension to a situation where the causes of the collider are marginally dependent}\label{subsec:Vx}

The result for the \textsc{or} effect scale in Theorem \ref{thm:V} for the \V~structure carries over to the \Vx~structure, such that
$$\Vx\textup{-bias}(C=c,\textsc{or})=\frac{p_{c\mid00}p_{c\mid11}}{p_{c\mid10}p_{c\mid01}},$$
that is, the \textsc{or} relating $X$ and $Y$ conditional on \mbox{$C=c$} is equal to their marginal \textsc{or} (which is $\displaystyle\frac{p_{Y=1\mid1}/p_{Y=0\mid1}}{p_{Y=1\mid0}/p_{Y=0\mid0}}$) times the bias factor above.
(The correspondence with the \V-bias result is due to the fact that the \V~structure is a special case of the \Vx~structure where the marginal \textsc{or} is 1.)
This is a well-known result in epidemiology for situations where sample selection and/or retention (here $C$) is influenced by both  exposure and outcome \citep{Kleinbaum1981,Greenland1996,Lash2009,jiang2016directions}; it is stated here for completeness.

This means that for the \Vx~structure, the effect of $X$ on $Y$, measured on the \textsc{or} scale, is not biased by conditioning on $C=c$ if $X$ and $Y$ do not interact on $C=c$ on the \textsc{rr} scale (and thus the bias factor is 1), but it is biased if $X$ and $Y$ interact on $C=c$ on the \textsc{rr} scale. This is similar to the interpretation of Theorem \ref{thm:V} for the \V~structure.
The difference is that with the \V~structure, this interpretation holds for all effect scales, but with the \Vx~structure, it holds only for the \textsc{or} scale.

The results in Corollaries \ref{crlry:V1} and \ref{crlry:V2} about the sign of \V-bias also hold for \Vx-bias on the \textsc{or} (but not other) effect scale.

\subsection{Visual representation of the sign of collider bias based on results in Sections \ref{subsec:V} and \ref{subsec:Vx}}\label{subsec:VisualVVx}

\begin{figure}
\caption{Schematic representation of variation in the effects of $X$ and $Y$ on $C$}
\centering
\includegraphics[width=.4\textwidth]{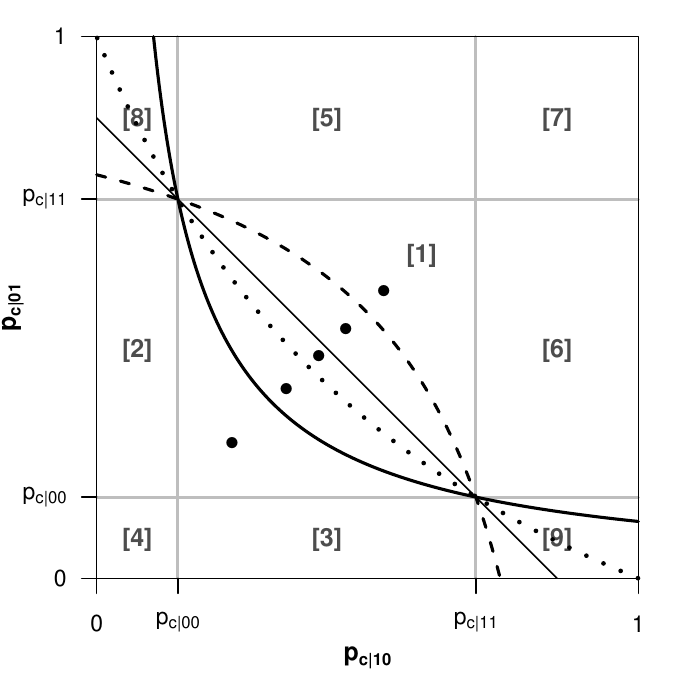}
\caption*{\footnotesize In this representation, $c$ is the level of $C$ such that $p_{c\mid11}\geq p_{c\mid00}$. Of the four black curves, the thick solid  and dashed curves are the two sets of points where $X$ and $Y$ do not interact on the \textsc{rr} scale on $C=c$ and on $C=1-c$, respectively; the diagonal line is where $X$ and $Y$ do not interact on the \textsc{rd} scale on (either value of) $C$; and the dotted curve is where $X$ and $Y$ do not interact on the \textsc{or} scale on (either value of) $C$.\\
\\}

\label{fig:2}
\caption{The sign of \V-bias on all effect scales and of $\nabla$-bias on the \textsc{or} scale, conditioning on the two levels of $C$}
\includegraphics[width=.4\textwidth]{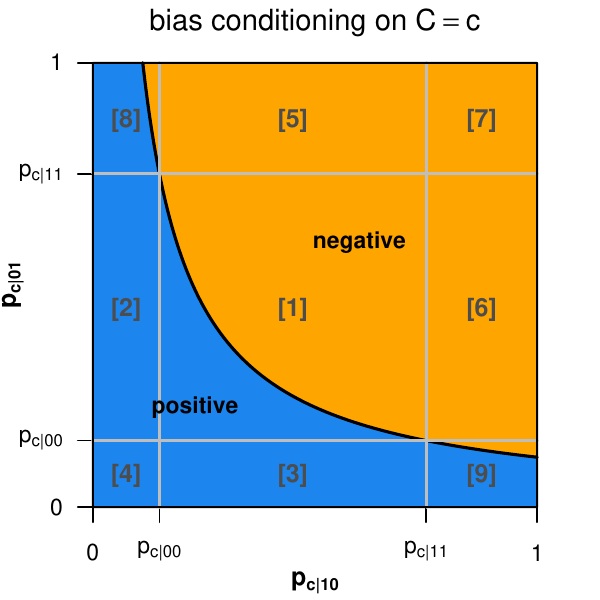}
\includegraphics[width=.4\textwidth]{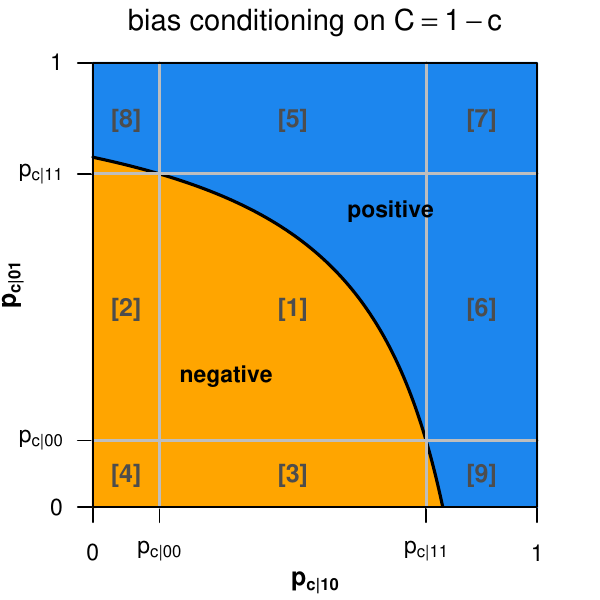}
\caption*{\footnotesize Here $c$ is the level of $C$ such that $p_{c\mid11}\geq p_{c\mid00}$. Bias is zero on the black curve, where $X$ and $Y$ do not interact on the \textsc{rr} scale on $C=c$ (in the left plot) or on $C=1-c$ (in the right plot).}
\label{fig:3}
\end{figure}

Given that much attention in the collider bias literature has been on the sign of bias, we devise a simple visual representation for the sign of collider bias as a function of the effects of $X$ and $Y$ on $C$.
This relies on depicting these effects through the four conditional probabilities of $C$ as in Fig. \ref{fig:2}.
In this representation, we pick $c$ to be the level of $C$ such that $p_{c\mid11}\geq p_{c\mid00}$.\add{ (This does not constitute any restriction, because if $p_{C=1\mid11}<p_{C=0\mid00}$, we simply pick $c=0$.)}
Fig. \ref{fig:2} is a graph with two axes representing possible values of  $p_{c\mid10}$ and $p_{c\mid01}$, ranging from 0 to 1, resulting in a square.
The specific values of $p_{c\mid00}$ and $p_{c\mid11}$ in a particular case are used to draw in the gray lines which divide the square into nine regions.
The specific values of $p_{c\mid10}$ and $p_{c\mid01}$ are represented as a point.
Take for example the first (i.e., leftmost) point of the five points marked by big dots. This point, combined with the gray lines, indicates a particular case where $p_{c\mid00}=.15,p_{c\mid11}=.75,p_{c\mid10}=p_{c\mid01}=.25$.

Which region the point falls in represents the type of effects $X$ and $Y$ have on the probability of $C=c$ (hereafter abbreviated to effects ``on $C=c$'').
\begin{itemize}
\item Region [1] is the \textit{same-sign effects}\footnote{Region [1] is named \textit{same-sign effects} instead of \textit{positive effects}, because the way $c$ is chosen ensures that some conditional effects are positive, so positive effects is not a defining characteristic of this region; the defining characteristic is that the effects are all in the same direction.} region: $X$ has positive effects on $C=c$ at both levels of $Y$ ($p_{c\mid10}>p_{c\mid00}$ and $p_{c\mid11}>p_{c\mid01}$); $Y$ has positive effects on $C=c$ at both levels of $X$ ($p_{c\mid01}>p_{c\mid00}$ and $p_{c\mid11}>p_{c\mid10}$). 
\item Regions [8] and [9] are the \textit{opposite-sign effects} regions: in region [8] $Y$ has positive effects on $C=c$ at both levels of $X$, but $X$ has negative effects on $C=c$ at both levels of $Y$ ($p_{c\mid10}<p_{c\mid00}$ and $p_{c\mid11}<p_{c\mid01}$); the opposite is true in region [9].
\item Regions [2]--[7] are the \textit{qualitative interaction} regions: in regions [2] and [5] the effect of $X$ changes sign depending on the level of $Y$; in regions [3] and [6] the effect of $Y$ changes sign across levels of $X$; and in regions [4] and [7] the effects of both causes change sign across levels of the other cause.
\end{itemize}
\add{There are no points in this unit square where the conditional effects of $X$ and $Y$ on $C=c$ are all negative, because that would imply $p_{c\mid00}>p_{c\mid11}$, and we have picked $c$ so that the opposite is true.}

\add{For a specific situation, this representation (with the regions and the point) is not unique, as it depends on the coding of the causes $X$ and $Y$ of the collider $C$. Reverse-coding both causes does not change how the regions look, but reflects the point across the 45 degree diagonal of the square -- e.g., if we are originally in region [8], we will be in region [9] in the new representation. Reverse-coding only one of the two causes, however, shifts the regions (because the probabilities labeled $p_{c\mid00}$ and $p_{c\mid11}$ are now different) in addition to changing the point; one potentially helpful result of this is switching from a representation with \textit{opposite-sign effects} (the point is in region [8] or [9]) to a new representation with \textit{same-sign effects} (the point is in region [1]) -- see an example shortly.
On the other hand, reverse-coding $C$ does not change the representation; it only changes the label of $c$.}

\medskip

The next level of detail in this representation scheme concerns only the \textit{same-sign} and \textit{opposite-sign effects} regions (regions [1], [8] and [9]). Within these regions, the effects of $X$ and $Y$ on $C$ can be further characterized in terms of \textit{quantitative interaction} (or non-interaction) on different scales. Fig. \ref{fig:2} includes four non-interaction curves:
(i) non-interaction on the \textsc{rr} scale on $C=c$ (the thick solid curve), (ii) non-interaction on the \textsc{rr} scale on $C=1-c$ (dashed curve), and non-interaction (iii) on the \textsc{or} scale (dotted curve) and (iv) on the \textsc{rd} scale (thin diagonal line) on either level of $C$.
\add{Note that on the \textsc{or} and \textsc{rd} scale, non-interaction effects on one level of $C$ implies non-interaction effects on the other level, but the same is not true for the \textsc{rr} scale -- hence the two curves for this scale.}

Consider $X$ and $Y$'s interactive effects at the five marked points (from left to right) in Fig. \ref{fig:2}.
$X$-$Y$ interaction on $C=c$ is positive on all three scales at the first point; negative on the \textsc{rr} scale but positive on the other two scales at the second point; negative on the \textsc{rr} and \textsc{or} scales and positive on the \textsc{rd} scale at the third point; and negative on all three scales at the fourth and fifth points.
$X$-$Y$ interaction on $C=1-c$ is negative on all three scales at the first and second points; positive on the \textsc{or} scale but negative on the other two scales at the third point; positive on the \textsc{or} and \textsc{rd} scales but negative on the \textsc{rr} scale at the fourth point; and positive on all three scales at the fifth point.

\medskip

Now we use this schematic representation to examine the sign of collider bias.
Fig. \ref{fig:3}, which shows the sign of \V-bias (regardless of effect scale) and \Vx-bias on the \textsc{or} effect scale conditioning on the two levels of $C$, makes clear several points:
\begin{itemize}
\item In the \textit{qualitative interaction} regions (regions [2]--[7]), such bias is positive for one level of $C$ and negative for the other (Corollary \ref{crlry:V2}). \add{For each level, the sign of bias is clear from the figure.}
\item In the \textit{same-sign effects} region (region [1]), such bias is always negative for at least one level of $C$ (Corollary \ref{crlry:V1}). \add{The sign of bias depends on the specific location within the region. For this region, there is a nice heuristic: bias conditioning on a level of $C$ is negative (positive) if $X$ and $Y$ interact negatively (positively) on the \textsc{rr} scale on that level of $C$, and is zero in the absence of such interaction.}
\item In the \textit{opposite-sign effects} regions (regions [8] and [9]), such bias is always positive for at least one level of $C$ (Corollary \ref{crlry:V1}). \add{Here the sign of bias also depends on the specific location within the region. Since the effects of $X$ and $Y$ are in opposite directions, there is no memorable heuristic, but one can always use the function $g(C)$ to determine the sign of bias.}
\end{itemize}
\add{Continuing the earlier comment about different representations, an alternative for this third case is to reverse-code either $X$ or $Y$ to obtain a new representation where the point is in the \textit{same-sign effects} region, making available the mentioned heuristic. This requires re-representation of all the relevant probabilities in the new notation.}

\medskip

\add{The examples in \cite{Berkson1946} belong in the \textit{same-sign effects} (region [1]) case, with $c=1$, as both diabetes ($X$) and cholecystitis ($Y$) increasing the probability of being selected into the hospital-based sample ($C=1$). While this paper is complex with many scenarios, we select (and adapt) two examples to illustrate the results shown in Fig. \ref{fig:3}. The two examples (see Table \ref{tab:BerksonExamples}) involve the same general population, but two different sampling cases, one with diabetes and cholecystitis interacting positively on the \textsc{rr} scale on sample selection, leading to positive bias (Example 1a), the other with negative interaction, resulting in negative bias (Example 1b). Example 1a is thus a point located in region [1] and in the blue-in-the-left-plot area; Example 1b is a point also in region [1] but in the orange-in-the-left-plot area.}

\begin{table}
\caption{Positive and negative \V-bias examples based on \cite{Berkson1946}}\label{tab:BerksonExamples}
\resizebox{\textwidth}{!}{%
\begin{tabular}{lrrrrlrrr}
\multicolumn{4}{l}{General population (from the fifth table in \citeauthor{Berkson1946})} 
\\
 & cholecystitis  & no cholecystitis & ~~~~~~~~~~~~ 
\\ \cline{2-3} 
\multicolumn{1}{l|}{diabetes} & \multicolumn{1}{r|}{3,000} & \multicolumn{1}{r|}{97,000} & 100,000 
\\ \cline{2-3} 
\multicolumn{1}{l|}{no diabetes} & \multicolumn{1}{r|}{297,000} & \multicolumn{1}{r|}{9,603,000} & 9,900,000 
\\ \cline{2-3} 
& 300,000 & 9,700,000 & 10,000,000
\\
\multicolumn{4}{l}{$\Rightarrow$ marginal \textsc{or} $=1$}
\\
\multicolumn{4}{l}{$\Rightarrow$ marginal \textsc{rd} of cholecystitis given diabetes status $=0$}
\\
\\
\multicolumn{4}{l}{\textbf{Example 1a}: Positive $X$-$Y$ interaction on $C$ on the \textsc{rr} scale} 
&& \multicolumn{4}{l}{\textbf{Example 1b}: Negative $X$-$Y$ interaction on $C$ on the \textsc{rr} scale}
\\
\multicolumn{4}{l}{$\{p_{C=1\mid00},p_{C=1\mid10},p_{C=1\mid01},p_{C=1\mid11}\}=$}
&&
\multicolumn{4}{l}{$\{p_{C=1\mid00},p_{C=1\mid10},p_{C=1\mid01},p_{C=1\mid11}\}=$}
\\
\multicolumn{4}{l}{~~~~$\{.0041,.0056,.0045,.0093\}$,~~~$g(1)=0.000013>0$.}
&& \multicolumn{4}{l}{~~~~$\{.005,.055,.030,.102\}$,~~~$g(1)=-0.0011<0$.}
\\[-.6em]
\\
\multicolumn{4}{l}{Hospital sample (from \citeauthor{Berkson1946}'s second table)} &&
\multicolumn{4}{l}{Hospital sample (adapted from \citeauthor{Berkson1946}'s seventh table)}
\\
 & cholecystitis  & no cholecystitis & ~~~~~~~~~~~~ &&& cholecystitis & no cholecystitis & ~~~~~~~~~~~~
\\ \cline{2-3} \cline{7-8}
\multicolumn{1}{l|}{diabetes} & \multicolumn{1}{r|}{28} & \multicolumn{1}{r|}{548} & 576 && 
\multicolumn{1}{l|}{diabetes} & \multicolumn{1}{r|}{306} & \multicolumn{1}{r|}{5,311} & 5,617
\\ \cline{2-3} \cline{7-8}
\multicolumn{1}{l|}{no diabetes} & \multicolumn{1}{r|}{1,326} & \multicolumn{1}{r|}{39,036} & 40,362
&&
\multicolumn{1}{l|}{no diabetes} & \multicolumn{1}{r|}{2,896} & \multicolumn{1}{r|}{48,015} & 50,911
\\ \cline{2-3} \cline{7-8}
& 1,354 & 39,584 & 40,938
&&& 3.202 & 53,326 & 56,528
\\
\multicolumn{4}{l}{$\Rightarrow$ conditional \textsc{or} $=1.504$}
&& \multicolumn{4}{l}{$\Rightarrow$ conditional \textsc{or} $=0.312$}
\\
\multicolumn{4}{l}{$\Rightarrow$ conditional \textsc{rd} of chole. given diab. status $=0.0068$}
&& \multicolumn{4}{l}{$\Rightarrow$ conditional \textsc{rd} of chole. given diab. status $=-.0024$}
\end{tabular}%
}
\end{table}
\medskip

Relating to the relevant literature, our analytical results are consistent with findings by \citet{jiang2016directions}, but sheds additional insights.
\citeauthor{jiang2016directions} examine the direction in which the treatment effect measured on the \textsc{or} scale is biased when sample selection is influenced by treatment and outcome.
This is \mbox{\Vx-bias$(C=1,\textsc{or})$} in our notation, where $C=1$ denotes sample selection.
They sign such bias in cases defined by $X$-$Y$ interaction or non-interaction on $C$ on the \textsc{rr}, \textsc{or} and \textsc{rd} scales, with focus on situations with no qualitative interaction (regions [1] and [8-9]).
Their results about how the sign of bias relates to $X$-$Y$ interaction on the \textsc{rr} scale match ours exactly.
Regarding $X$-$Y$ interaction on the \textsc{or} and \textsc{rd} scales, our results are more informative.
For an example, assume that the effects of $X$ and $Y$ on $C$ are positive, i.e., we are looking at region [1] in Fig. \ref{fig:2} and in the left plot of Fig. \ref{fig:3}, taking $c$ to be 1.
According to \citeauthor{jiang2016directions}, if there is non-positive $X$-$Y$ interaction on $C=1$ on the \textsc{or} scale (i.e., in the area to the right of the dotted curve including the dotted curve) or on the \textsc{rd} scale (i.e., in the area to the right of the diagonal line including the diagonal line), \Vx-bias conditioning on $C=1$ for the \textsc{or} effect scale is non-positive.
As can be seen in Fig. \ref{fig:2}, except for the two points where the non-interaction curves intersect, these two areas are contained in the area of negative $X$-$Y$ interaction on the \textsc{rr} scale, where this bias is strictly negative; it is only at those two special points (where only treatment or outcome, but not both, influences sample selection) that this bias is zero.

\add{As is clear from Fig. \ref{fig:3}, what matters for determining the sign of \V-bias (and of \Vx-bias on the \textsc{or} scale) is information about the interaction (or non-interaction) of $X$ and $Y$ on each conditioning level of $C$ on the \textsc{rr} scale. Information about interaction on the \textsc{or} or \textsc{rd} scale is suboptimal for this purpose.}

\medskip

\add{For the remainder of this discussion, we consider the \textit{same-sign efffects} case (region [1]) only. Within this region, the relative area of negative and positive bias depends on the positions of the corners $(p_{c\mid00},p_{c\mid00})$ and $(p_{c\mid11},p_{c\mid11})$. The closer $(p_{c\mid00},p_{c\mid00})$ is to $(0,0)$, the larger the area of negative bias conditioning on $C=c$ relative to the area of positive bias. The closer $(p_{c\mid11},p_{c\mid11})$ is to $(1,1)$, the larger the area of negative bias conditioning on $C=1-c$ relative to the area of positive bias. As $p_{c\mid00}\to 0$, the \textsc{rr}-scale non-interaction on $C=c$ curve approaches the west-and-south border of region [1]. As $p_{c\mid11}\to1$, the \textsc{rr}-scale non-interaction on $C=1-c$ curve approaches the north-and-east border of region [1].

Sufficient causes situations, where the effects of $X$ and $Y$ on $C$ are deterministic, are special cases that involve both of these limits.
Take, for example, the case where $C=c$ when $X=1,Y=1$ and $C=1-c$ otherwise (i.e., both causes are needed).
In this case, $p_{c\mid11}=1,p_{c\mid00}=0$, so region [1] is the whole unit square, and the two plots in Fig. \ref{fig:3} would be fully orange except the west-and-south border of the left plot and the north-and-east border of the right plot.
In addition, the point $(p_{c\mid10},p_{c\mid01})$ is exactly at the $(0,0)$ corner. 
Conditional on $C=c$, there is no variation in $X$ or $Y$, so the conditional covariance of $X$ and $Y$ is zero and the conditional \textsc{rd} and \textsc{or} of $Y$ comparing the two levels of $X$ are undefined (this can be read off of the formulas in Theorem \ref{thm:V}).
Yet \V-bias conditioning on $C=1-c$ is negative, with specific values $\displaystyle-\frac{p_{X=1}p_{X=0}p_{Y=1}p_{Y=0}}{(p_{X=1}p_{Y=0}+p_{X=0}p_{Y=1}+p_{X=0}p_{Y=0})^2}$ on the covariance scale, $-p_{Y=1}$ on the \textsc{rd} scale, and 0 on the \textsc{or} scale.
Mirroring this case is the case where either cause is sufficient, i.e., $C=c$ when $X=1$ and/or $Y=1$, and $C=1-c$ when $X=0,Y=0$. In this case, region [1] is also the whole unit square, but the $(p_{c\mid10},p_{c\mid01})$ point is at $(1,1)$. In addition, there are situations where the effects of $X$ and $Y$ on $C$ combine deterministic and stochastic elements; these may involve one or both of the two limits.
Table \ref{tab:examplesCole} shows several situations with varying deterministic and semi-deterministic effects, constructed based on \citeauthor{Cole2010a}'s (\citeyear{Cole2010a}) original example.}

\medskip

\begin{table}[h!]
\caption{\V-bias examples for deterministic and semi-deterministic causal structures based on \cite{Cole2010a}}\label{tab:examplesCole}
\resizebox{\textwidth}{!}{%
\begin{tabular}{lrrrrlrrr}
\multicolumn{4}{l}{All party attendants} 
\\
 & ~~~~~~~~~~~~~flu  & ~~~~~~~~~~no flu & ~~~~~~~~~~~~ 
\\ \cline{2-3} 
\multicolumn{1}{l|}{tainted food} & \multicolumn{1}{r|}{5} & \multicolumn{1}{r|}{45} & 50 
\\ \cline{2-3}
\multicolumn{1}{l|}{no tainted food} & \multicolumn{1}{r|}{5} & \multicolumn{1}{r|}{45} & 50 
\\ \cline{2-3} 
& 10 & 90 & 100
\\
\multicolumn{4}{l}{$\Rightarrow$ marginal \textsc{or} $=1$}
\\
\multicolumn{4}{l}{$\Rightarrow$ marginal \textsc{rd} of flu given food $=0$}
\\
\\
\multicolumn{4}{l}{\textbf{Example 2a}: Deterministic effects} &&
\multicolumn{4}{l}{\textbf{Example 2b}: Semi-deterministic effects}\\
\multicolumn{4}{l}{$\{p_{C=1\mid00},p_{C=1\mid10},p_{C=1\mid01},p_{C=1\mid11}\}=\{0,1,1,1\}$}
&& \multicolumn{4}{l}{$\{p_{C=1\mid00},p_{C=1\mid10},p_{C=1\mid01},p_{C=1\mid11}\}=\{0,.4,.6,1\}$}
\\[-.6em]
\\
\multicolumn{4}{l}{Those who developed fever}
&& \multicolumn{4}{l}{Those who developed fever}
\\
 & flu  & no flu & ~~~~~~~~~~~~ &&& ~~~~~~~~~~~~~flu & ~~~~~~~~~~no flu & ~~~~~~~~~~~~
\\ \cline{2-3} \cline{7-8}
\multicolumn{1}{l|}{tainted food} & \multicolumn{1}{r|}{5} & \multicolumn{1}{r|}{45} & 50 && 
\multicolumn{1}{l|}{tainted food} & \multicolumn{1}{r|}{5} & \multicolumn{1}{r|}{18} & 23
\\ \cline{2-3} \cline{7-8}
\multicolumn{1}{l|}{no tainted food} & \multicolumn{1}{r|}{5} & \multicolumn{1}{r|}{0} & 5 
&&
\multicolumn{1}{l|}{no tainted food} & \multicolumn{1}{r|}{3} & \multicolumn{1}{r|}{0} & 3
\\ \cline{2-3} \cline{7-8}
& 10 & 45 & 65
&&& 8 & 23 & 26
\\
\multicolumn{4}{l}{$\Rightarrow$ conditional \textsc{or} $=0$}
&& \multicolumn{4}{l}{$\Rightarrow$ conditional \textsc{or} $=0$}
\\
\multicolumn{4}{l}{$\Rightarrow$ conditional \textsc{rd} of flu given food $=-0.9$}
&& \multicolumn{4}{l}{$\Rightarrow$ conditional \textsc{rd} of flu given food $=-0.78$}
\\
\\
\multicolumn{4}{l}{\textbf{Example 2c}: Semi-deterministic effects} &&
\multicolumn{4}{l}{\textbf{Example 2d}: Semi-deterministic effects}\\
\multicolumn{4}{l}{$\{p_{C=1\mid00},p_{C=1\mid10},p_{C=1\mid01},p_{C=1\mid11}\}=\{0,.4,.6,.8\}$}
&& \multicolumn{4}{l}{$\{p_{C=1\mid00},p_{C=1\mid10},p_{C=1\mid01},p_{C=1\mid11}\}=\{.2,1,1,1\}$}
\\[-.6em]
\\
\multicolumn{4}{l}{Those who developed fever}
&& \multicolumn{4}{l}{Those who developed fever}
\\
 & flu  & no flu & ~~~~~~~~~~~~ &&& flu & no flu & ~~~~~~~~~~~~
\\ \cline{2-3} \cline{7-8}
\multicolumn{1}{l|}{tainted food} & \multicolumn{1}{r|}{4} & \multicolumn{1}{r|}{18} & 22 && 
\multicolumn{1}{l|}{tainted food} & \multicolumn{1}{r|}{5} & \multicolumn{1}{r|}{45} & 50
\\ \cline{2-3} \cline{7-8}
\multicolumn{1}{l|}{no tainted food} & \multicolumn{1}{r|}{3} & \multicolumn{1}{r|}{0} & 3
&&
\multicolumn{1}{l|}{no tainted food} & \multicolumn{1}{r|}{5} & \multicolumn{1}{r|}{9} & 14
\\ \cline{2-3} \cline{7-8}
& 7 & 18 & 25
&&& 10 & 54 & 64
\\
\multicolumn{4}{l}{$\Rightarrow$ conditional \textsc{or} $=0$}
&& \multicolumn{4}{l}{$\Rightarrow$ conditional \textsc{or} $=0.2$}
\\
\multicolumn{4}{l}{$\Rightarrow$ conditional \textsc{rd} of flu given food $=-0.82$}
&& \multicolumn{4}{l}{$\Rightarrow$ conditional \textsc{rd} of flu given food $=-0.26$}
\end{tabular}%
}\end{table}

\subsection{Bias due to conditioning on a specific level of the child of the collider in the \Y~structure}

\add{In the Berkson examples discussed above, we referenced the effects of diabetes ($X$) and cholecystitis ($Y$) on sample selection without considering any intermediate variables. Suppose that the selection of the sample is a random draw from the hospitalized patients with equal probabilities regardless of the reasons for which the patients are hospitalized. Then we can think of $X$ and $Y$ as affecting the probability of hospitalization (denoted $C$), which in turns is the sole cause of sample selection (denoted $D$), and we have a \Y-structure. Note that for a \Y-structure, $D$ is not influenced by $X$ and $Y$ other than through $C$, and $D$ does not share common causes with $X$ and $Y$.}

The following theorem provides formulas for \Y-bias conditioning on a level of $D$. For the covariance or \textsc{rd} effect scale, such bias is dependent on the marginal probabilities of $X$ and $Y$, and the conditional probabilities of $C$ and $D$; for the \textsc{or} effect scale, such bias is a function of the conditional probabilities of $C$ and $D$ only.

\begin{theorem}[$D$-specific \Y-bias theorem]\label{thm:Y}
\Y-bias conditioning on $D=d$, for $d\in\{0,1\}$, is given by the following expressions:
\begin{align*}
    \Y\textup{-bias}(D=d,\mathrm{cov})
    &=\frac{p_{X=1}p_{X=0}p_{Y=1}p_{Y=0}}{\{\mathrm{P}(D=d)\}^2}\cdot
    (p_{d\mid1}-p_{d\mid0})\cdot\left\{
    \begin{matrix*}[l]
        p_{d\mid1}(p_{C=1\mid00}p_{C=1\mid11}-p_{C=1\mid10}p_{C=1\mid01})-\\
        p_{d\mid0}(p_{C=0\mid00}p_{C=0\mid11}-p_{C=0\mid10}p_{C=0\mid01})
    \end{matrix*}
    \right\},\\[-.5em]
    \\
    \Y\textup{-bias}(D=d,\textsc{rd})&=
    \frac{p_{Y=1}p_{Y=0}}{
    \left\{
    \begin{matrix*}[l]
		p_{Y=1}(p_{C=1\mid11}p_{d\mid1}+p_{C=0\mid11}p_{d\mid0})+\\
		p_{Y=0}(p_{C=1\mid10}p_{d\mid1}+p_{C=0\mid10}p_{d\mid0})
    \end{matrix*}
    \right\}
    \cdot
    \left\{
    \begin{matrix*}[l]
		p_{Y=1}(p_{C=1\mid01}p_{d\mid1}+p_{C=0\mid01}p_{d\mid0})+\\
		p_{Y=0}(p_{C=1\mid00}p_{d\mid1}+p_{C=0\mid00}p_{d\mid0})
    \end{matrix*}
    \right\}
    }\times\\
    &~~~~~\times(p_{d\mid1}-p_{d\mid0})\cdot\left\{
    \begin{matrix*}[l]
        p_{d\mid1}(p_{C=1\mid00}p_{C=1\mid11}-p_{C=1\mid10}p_{C=1\mid01})-\\
        p_{d\mid0}(p_{C=0\mid00}p_{C=0\mid11}-p_{C=0\mid10}p_{C=0\mid01})
    \end{matrix*}
    \right\},\\[-.5em]
    \\
    \Y\textup{-bias}(D=d,\textsc{or})&=
    \frac{(p_{d\mid1}-p_{d\mid0})(p_{d\mid1}p_{C=1\mid00}p_{C=1\mid11}-p_{d\mid0}p_{C=0\mid00}p_{C=0\mid11})+p_{d\mid1}p_{d\mid0}}{(p_{d\mid1}-p_{d\mid0})(p_{d\mid1}p_{C=1\mid10}p_{C=1\mid01}-p_{d\mid0}p_{C=0\mid10}p_{C=0\mid01})+p_{d\mid1}p_{d\mid0}}.
\end{align*}
For completeness,
\begin{align*}
\P(D=d)&=p_{X=1}p_{Y=1}(p_{C=1\mid11}p_{d\mid1}+p_{C=0\mid11}p_{d\mid0})+p_{X=0}p_{Y=1}(p_{C=1\mid01}p_{d\mid1}+p_{C=0\mid01}p_{d\mid0})+\\
&~~~~~p_{X=1}p_{Y=0}(p_{C=1\mid10}p_{d\mid1}+p_{C=0\mid10}p_{d\mid0})+p_{X=0}p_{Y=0}(p_{C=1\mid00}p_{d\mid1}+p_{C=0\mid00}p_{d\mid0}).
\end{align*}
\end{theorem}

\medskip

Recall function $g(c)=p_{c\mid00}p_{c\mid11}-p_{c\mid10}p_{c\mid01}$ defined in section \ref{subsec:V}.
Theorem \ref{thm:Y} shows that \Y-bias conditioning on $D=d$ is of the same sign as $(p_{d\mid1}-p_{d\mid0})\cdot\{p_{d\mid1}g(1)-p_{d\mid0}g(0)\}$.
This means the sign of \Y-bias depends on the effect of $C$ on $D$, the effects of $X$ and $Y$ on $C$, and how these two types of effects relate to each other.

That $C$ is influenced by both $X$ and $Y$ means that $g(1)$ and $g(0)$ cannot be simultaneously zero, because $g(1)=g(0)=0$ implies that either $X$ or $Y$ has no effect on $C$.
The following situations are therefore mutually exclusive:
\begin{enumerate}
\item If $g(1)\geq0$ and $g(0)\leq0$, \Y-bias conditioning on $D=d$ is the same sign as the effect of $C$ on the probability of $D=d$.
\item If $g(1)\leq0$ and $g(0)\geq0$, the reverse is true.
\item If $g(1)$ and $g(0)$ are both non-positive or both non-negative, \Y-bias conditioning on $D=d$ is zero if the effect of $C$ on the probability of $D=d$ is $p_{d\mid1}/p_{d\mid0}=g(0)/g(1)$.
\begin{enumerate}
\item If $g(1)\leq0$ and $g(0)\leq0$, \Y-bias conditioning on $D=d$ is positive if $p_{d\mid1}/p_{d\mid0}$ is in between the two values $g(0)/g(1)$ and 1 (regardless of their order), and negative if $p_{d\mid1}/p_{d\mid0}$ is outside of this range.
\item If $g(1)\geq0$ and $g(0)\geq0$, the reverse is true.
\end{enumerate}
\end{enumerate}

\medskip

\add{We visualize the sign of \Y-bias in Fig. \ref{fig:4}. Building on the scheme we have for representing the effects of $X$ and $Y$ on $C$, this figure adds one design element that reflects the direction (but not magnitude) of the effect of $C$ on $D$: picking $d$ such that $\P(D=d\mid C=c)\geq\P(D=d\mid C=1-c)$.
The sign of bias information shown here is thus partial.
Yet the figure offers additional understanding that may otherwise be not obvious.}
\begin{itemize}
\item \add{In the \textit{qualitative interaction} regions (regions [2]--[7]), \Y-bias conditioning on $D=d$ is of the same sign as \V-bias conditioning on $C=c$, \Y-bias conditioning on $D=1-d$ is of the same sign as \V-bias conditioning on $C=1-c$.}
\item \add{In the \textit{same-sign effects} region (region [1]), \Y-bias is negative for at least one level of $D$. The sign of bias depends on the specific location within the region. A heuristic: the sign of \Y-bias conditioning on $D=d$ (i) is positive if $X$-$Y$ interaction on $C=c$ is non-negative (meaning positive interaction or no interaction) on the \textsc{rr} scale, (ii) is negative if $X$-$Y$ interaction on $C=c$ is non-positive on the \textsc{rd} scale, and (iii) depends additionally on the effect of $C$ on $D$ if $X$-$Y$ interaction on $C=c$ is negative on the \textsc{rr} scale but positive on the \textsc{rd} scale.
This statement also holds if we simultaneously replace $d$ with $1-d$ and $c$ with $1-c$.}
\item \add{In the \textit{opposite-sign effects} regions (regions [8] and [9]), \Y-bias is positive for at least one level of $D$. The sign of bias depends on the specific location within the region. Again, since the effects of $X$ and $Y$ on $C$ are in opposite directions, there is no convenient heuristic, but one can check the cases 1, 2, 3a and 3b listed above to determine the sign of bias. And, of course, an alternative is to reverse-code one of the two causes of $C$ to obtain a different representation where the point is in the \textit{same-sign effects} region.}
\end{itemize}

\begin{figure}
\caption{The sign of \Y-bias conditioning on the two levels of $D$}
\label{fig:4}
\centering
\includegraphics[width=.4\textwidth]{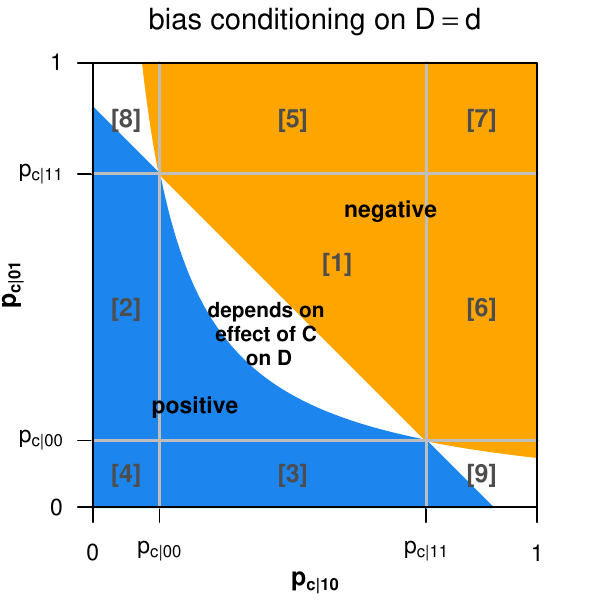}
\includegraphics[width=.4\textwidth]{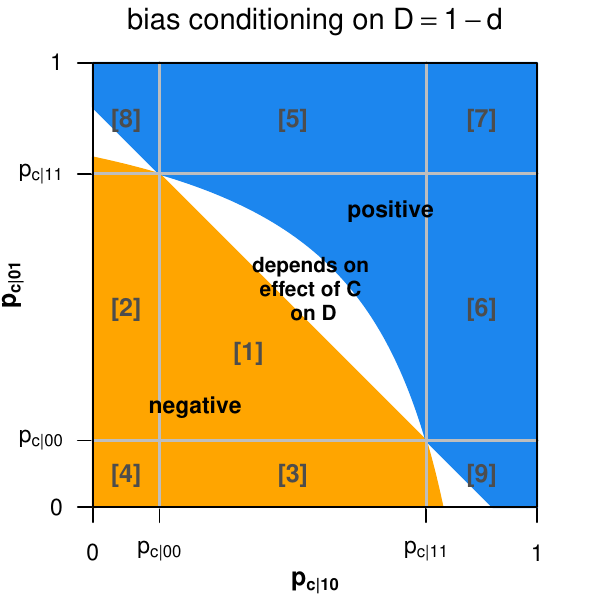}
\caption*{\footnotesize Here $c$ is the level of $C$ such that $p_{c\mid11}\geq p_{c\mid00}$, and $d$ is the level of $D$ such that $p_{d\mid C=c}>p_{d\mid C=1-c}$. Each of the two plots above are split into blue, orange and white areas by the intersection of the curve representing $X$-$Y$ non-interaction on the risk ratio scale on $C=c$ (in the left plot) and $C=1-c$ (in the right plot) and the line representing non-interaction on the risk difference scale. Except for the two intersection points, the curve/line segments separating the colored and white areas belong in the colored areas. \Y-bias conditioning on the specific level of $D$ is positive in the blue area, negative in the orange area, and depends on the effect of $C$ on $D$ (see text for more detail) in the white areas.}
\end{figure}

\add{A quick note: The pair $(c,d)$ in this schematic representation may take on any of the four pairs of values $(0,0),(0,1),(1,0),(1,1)$ depending on the effects of $X$ and $Y$ on $C$ and the effect of $C$ on $D$. If the pair in a specific situation proves inconvenient for cognitive processing or book-keeping, reverse-coding $C$ or $D$ or both is an option. Such recoding does not affect the picture, except for changing the numeric label(s) of $c$ and/or $d$.}

\smallskip

\add{This visualization shows that unlike for the \V~structure where only information on interaction on the \textsc{rr} is needed for signing \V-bias, here information on interaction on the \textsc{rd} scale is additionally helpful.

Why this is the case deserves some explanation.
To quickly see why, let's just focus on \Y-bias conditioning on $D=d$ in region [1].
To simplify the argument, we suppose $c=1$ (in this paragraph only) and refer to results for cases 1, 2 and 3a discussed earlier.
Split region [1] into two outer areas and two inner areas using the two \textsc{rr}-scale non-interaction curves and the diagonal \textsc{rd}-scale non-interaction line. Let the two outer areas include the two curves, and the inner area on the right side of the diagonal line include the line.
The bottom-left outer area belongs to case 1, where \Y-bias conditioning on $D=d$ there is the same sign as the effect of $C$ on $D=d$, which is positive, as shown in blue in the left plot of the figure.
Similarly, the top-right outer area belongs to case 2, where \Y-bias conditioning on $D=d$ is negative.
Now consider the inner area to the right of the diagonal line. With $g(0)<0,g(1)<0$, this area belongs in case 3a.
In this area $g(1)-g(0)=p_{C=1\mid00}+p_{C=1\mid11}-p_{C=1\mid10}-p_{C=1\mid01}\leq0$ because $X$ and $Y$ interact non-positively on $C$ on the \textsc{rd} scale.
It follows that $0<g(0)/g(1)\leq1$, implying that $p_{d\mid1}/p_{d\mid0}$ (which is greater than 1 because the effect of $C$ on $D=d$ is positive) is outside the range between $g(0)/g(1)$ and 1.
Based on case 3a, this means \Y-bias conditioning on $D=d$ is negative. 
Combining this and the top-right outer area, we have the negative bias part shown in orange. 
In the remaining area shown in white (which also belongs in case 3a), $g(0)/g(1)>1$, so without information on the actual magnitude of $p_{d\mid1}/p_{d\mid0}$, the sign of \Y-bias conditioning on $D=d$ cannot be determined.}

\medskip

\add{Coming back to the big picture, note that what is missing in each of the two plots in Fig. \ref{fig:4} is a curve that separates negative and positive bias, where bias is zero. This curve lies in the white space of the plot, and goes through the two points $(p_{c\mid00},p_{c\mid11})$ and $(p_{c\mid11},p_{c\mid00})$. Where it lies in the left and right plots depends respectively on the two \textsc{rr}s $\displaystyle\frac{\P(D=d\mid C=c)}{\P(D=d\mid C=1-c)}$ and $\displaystyle\frac{P(D=1-d\mid C=c)}{\P(D=1-d\mid C=1-c)}$ in the specific case.
One thing that can be said about the zero curves for \Y-bias is that they have less curvature than the zero curves for the related \V-bias.}

\medskip

\begin{corollary}[$D$-specific \Y-bias corollary]\label{crlry:Y}
We refer to collider bias in the \V~substructure embedded in the \Y~structure as `embedded \V-bias' and denote it as \V\textup{-bias-em}. For the covariance effect scale, \Y-bias relates to embedded \V-bias through the following formula:
\begin{align*}
    \Y\textup{-bias}(D=d&,\mathrm{cov})=
    \frac{(p_{d\mid1}-p_{d\mid0})}{\{\mathrm{P}(D=d)\}^2}\cdot
    \begin{bmatrix*}[l]
    p_{d\mid1}\{\mathrm{P}(C=1)\}^2\cdot\V\textup{-bias-em}(C=1,\mathrm{cov})-\\
    p_{d\mid0}\{\mathrm{P}(C=0)\}^2\cdot\V\textup{-bias-em}(C=0,\mathrm{cov})
    \end{bmatrix*}.
\end{align*}
\end{corollary}

This means that, \Y~bias conditioning on a level of $D$ for the covariance effect scale is a linear combination of embedded \V-bias (for the same effect scale) conditioning on one level of $C$ and negative embedded \V-bias conditioning on the other level of $C$.

\medskip

\add{Last but not least, note that the key quantity for signing \Y-bias, $(p_{d\mid1}-p_{d\mid0})\cdot\{p_{d\mid1}g(1)-p_{d\mid0}g(0)\}$, is equal to 
$\P(D=d\mid X=0,Y=0)\P(D=d\mid X=1,Y=1)-\P(D=d\mid X=1,Y=0)\P(D=d\mid X=0,Y=1),$
which is the same form as function $g(c)$, except replacing variable $C$ with $D$. This is no coincidence. The four conditional probabilities of $D$ in this expression represent the effects of $X$ and $Y$ on $D$. These conditional probabilities, combined with the marginal probabilities of $X$ and $Y$, are enough to inform of collider bias conditioning on $D$; how $X$ and $Y$ influence $D$ (here through $C$) provides no additional relevant information. This is no different from the \V~structure case in section \ref{subsec:V}, where the conditional probabilities of $C$ are sufficient as information on the effects of $X$ and $Y$ on $C$, and we do not need to know the specifics of how $X$ and $Y$ influence $C$. The key point is that if the effects of $X$ and $Y$ on $D$ are known, we can simply consider the triple $X\rightarrow D\leftarrow Y$, which is a causal DAG (including all common causes of any pair of included variables) and thus a \V~structure.
This provides another avenue for assessing \Y-bias conditioning on $D$ -- first deriving the reduced \V~structure with $D$ as the collider.}

\subsection{Bias due to conditioning on a specific level of the collider or its child in structures with \V~or \Y~substructures where the collider's parents are marginally independent}

\add{In this section, we consider collider bias in the more complex structures that have an embedded \V~or \Y~sub-structure. Real-world examples of these structures abound. An example of the \rightM~structure is a simplified version of the low birth weight ``paradox'' \citep{porta2015current}. Here $X$ is mother's smoking status, $Y$ is infant death, $C$ is birth weight, which is negatively affected by mother's smoking and also affected by unknown factors $B$ that both reduce birth weight and increase risk of death. It is observed that in the low birth weight stratum, mother's smoking is negatively associated with infant death. As described, this is a \rightM~structure. (Of course the real structure may be more complicated, as there may be direct effects of low birth weight and of mother's smoking on mortality).}

\add{Another example is healthy worker bias, elaborated in \citep{Hernan2004}. Here $X$ is exposure to a chemical at the workplace, $Y$ is death, $C$ is being at work, $B$ is true health status. $B$ affects both workplace participation and death. Presumably, $X$ may cause absenteeism but does not affect mortality. If all people at work are surveyed, we have a \rightM~structure. If a subsample of people are surveyed, differentiating being at work ($C$) and being sampled ($D$), we have a \rightlongM~structure. Alternatively, suppose that there is such a chemical (now denoted $A$) except that it is unknown, unmeasured, and unsuspected. However, it causes some respiratory problems that are considered the exposure of interest (our new $X$). In this case, we have an \M~or \longM~structure.}

The theorem below extends our results from the \V~and \Y~structures to these more complex structures.

\begin{theorem}[$C$- or $D$-specific collider bias extension theorem]\label{thm:stratum-extension}
Consider the six structures \rightM, \leftM, \M, \rightlongM, \leftlongM, and \longM.
We refer to \rightM-bias, \leftM-bias and \M-bias collectively as `extended \V-bias', refer to collider bias in the \V~sub-structure embedded in the \rightM, \leftM~and \M~structures as `embedded \V-bias', and denote these biases conditioning on $C=c$ by \V\textup{-bias-ext}$(C=c)$ and \mbox{\V\textup{-bias-em}$(C=c)$}, respectively. Similarly, we refer to \rightlongM-bias, \leftlongM-bias and \longM-bias collectively as `extended \Y-bias', refer to collider bias in the \Y-sub-structure embedded in the \rightlongM, \leftlongM~ and \longM~structures as `embedded \Y-bias', and use the corresponding notations \Y\textup{-bias-ext}$(D=d)$ and \Y\textup{-bias-em}$(D=d)$. Then for the covariance and \textsc{rd} effect scales, the extended biases relate to the embedded biases by the following formulas:
\begin{align*}
    &\V\textup{-bias-ext}(C=c,\mathrm{cov})=\textsc{rd}_\mathrm{left}\cdot\V\textup{-bias-em}(C=c,\mathrm{cov})\cdot\textsc{rd}_\mathrm{right},\\
    &\V\textup{-bias-ext}(C=c,\textsc{rd})=\textsc{rd}_\mathrm{left}\cdot\V\textup{-bias-em}(C=c,\textsc{rd})\cdot\textsc{rd}_\mathrm{right}\cdot\textsc{vr}(c),\\
    &\Y\textup{-bias-ext}(D=d,\mathrm{cov})=\textsc{rd}_\mathrm{left}\cdot\Y\textup{-bias-em}(D=d,\mathrm{cov})\cdot\textsc{rd}_\mathrm{right},\\
    &\Y\textup{-bias-ext}(D=d,\textsc{rd})=\textsc{rd}_\mathrm{left}\cdot\Y\textup{-bias-em}(D=d,\textsc{rd})\cdot\textsc{rd}_\mathrm{right}\cdot\textsc{vr}(d),
\end{align*}
where
\begin{itemize}
\item $\textsc{rd}_\mathrm{left}$ is 1 for the \rightM~and \rightlongM~structures, and is the \textsc{rd} representing the effect of $A$ on $X$ (i.e., $p_{X=1\mid1}-p_{X=1\mid0}$) in the \leftM, \M, \leftlongM~and \longM~structures;
\item $\textsc{rd}_\mathrm{right}$ is 1 for the \leftM~and \leftlongM~structures, and is the \textsc{rd} representing the effect of $B$ on $Y$ (i.e., $p_{Y=1\mid1}-p_{Y=1\mid0}$) in the \rightM, \M, \rightlongM~and \longM~structures;
\item $\textsc{vr}(c)$ is 1 for the \rightM~structure, and is $\mathrm{var}(A\mid C=c)/\mathrm{var}(X\mid C=c)$ in the \leftM~and \M~structures;
\item $\textsc{vr}(d)$ is 1 for the \rightlongM~structure, and is $\mathrm{var}(A\mid D=d)/\mathrm{var}(X\mid D=d)$ in the \leftlongM~and \longM~structures.
\end{itemize}
For completeness, with the \leftM~structure,
\begin{align*}
\frac{\mathrm{var}(A\mid C=c)}{\mathrm{var}(X\mid C=c)}
=
\frac{p_{A=1}(p_{Y=1}p_{c\mid11}+p_{Y=0}p_{c\mid10})\times p_{A=0}(p_{Y=1}p_{c\mid01}+p_{Y=0}p_{c\mid00})}{
    \left\{
    \begin{matrix*}[l]
        p_{X=1\mid1}p_{A=1}(p_{Y=1}p_{c\mid11}+p_{Y=0}p_{c\mid10})+\\
        p_{X=1\mid0}p_{A=0}(p_{Y=1}p_{c\mid01}+p_{Y=0}p_{c\mid00}) 
    \end{matrix*}
    \right\}\times
    \left\{
    \begin{matrix*}[l]
        p_{X=0\mid1}p_{A=1}(p_{Y=1}p_{c\mid11}+p_{Y=0}p_{c\mid10})+\\
        p_{X=0\mid0}p_{A=0}(p_{Y=1}p_{c\mid01}+p_{Y=0}p_{c\mid00})
    \end{matrix*}
    \right\}
    };
\end{align*}
the same formula applies to the \M~structure, except $Y$ is replaced with $B$. With the \leftlongM~structure,
\begin{align*}
\frac{\mathrm{var}(A\mid D=d)}{\mathrm{var}(X\mid D=d)}
=
\frac{
    p_{A=1}
    \left\{
        \begin{matrix*}[l]
            p_{d\mid1}(p_{C=1\mid11} p_{Y=1}+p_{C=1\mid10} p_{Y=0})+\\
            p_{d\mid0}(p_{C=0\mid11}p_{Y=1}+p_{C=0\mid10}p_{Y=0})
        \end{matrix*}
    \right\}
    \times
    p_{A=0}
    \left\{
        \begin{matrix*}[l]
            p_{d\mid1}(p_{C=1\mid01} p_{Y=1}+p_{C=1\mid00} p_{Y=0})+\\
            p_{d\mid0}(p_{C=0\mid01}p_{Y=1}+p_{C=0\mid00}p_{Y=0})
        \end{matrix*}
    \right\}
    }{
    \begin{bmatrix*}[l]
        p_{X=1\mid1}p_{A=1}\times\\
        \left\{
        \begin{matrix*}[l]
            p_{d\mid1}(p_{C=1\mid11} p_{Y=1}+p_{C=1\mid10} p_{Y=0})+\\
            p_{d\mid0}(p_{C=0\mid11}p_{Y=1}+p_{C=0\mid10}p_{Y=0})
        \end{matrix*}
        \right\}+\\
        p_{X=1\mid0}p_{A=0}\times\\
        \left\{
        \begin{matrix*}[l]
            p_{d\mid1}(p_{C=1\mid01} p_{Y=1}+p_{C=1\mid00} p_{Y=0})+\\
            p_{d\mid0}(p_{C=0\mid01}p_{Y=1}+p_{C=0\mid00}p_{Y=0})
        \end{matrix*}
        \right\}
    \end{bmatrix*}
    \times
    \begin{bmatrix*}[l]
        p_{X=0\mid1}p_{A=1}\times\\
        \left\{
        \begin{matrix*}[l]
            p_{d\mid1}(p_{C=1\mid11} p_{Y=1}+p_{C=1\mid10} p_{Y=0})+\\
            p_{d\mid0}(p_{C=0\mid11}p_{Y=1}+p_{C=0\mid10}p_{Y=0})
        \end{matrix*}
        \right\}+\\
        p_{X=0\mid0}p_{A=0}\times\\
        \left\{
        \begin{matrix*}[l]
            p_{d\mid1}(p_{C=1\mid01}p_{Y=1}+p_{C=1\mid00}p_{Y=0})+\\
            p_{d\mid0}(p_{C=0\mid01}p_{Y=1}+p_{C=0\mid00}p_{Y=0})
        \end{matrix*}
        \right\}
    \end{bmatrix*}
    };
\end{align*}
the same formula applies to the \longM~structure, except $Y$ is replaced with $B$.
\end{theorem}

\medskip

Theorem \ref{thm:stratum-extension} shows that extending a \V~or \Y~structure to the left (to \leftM~or \leftlongM) and to the right (to \rightM~or \rightlongM) results in symmetric changes for collider bias on the covariance scale.
The changes to collider bias for the \textsc{rd} effect scale are not symmetric: extending to the right changes collider bias only by a factor of the \textsc{rd} representing the right-extension effect, but extending to the left changes collider bias by a factor that combines the \textsc{rd} representing the left-extension effect and a ratio between two conditional variances of $A$ and of $X$. This asymmetric result is due to the fact that the \textsc{rd} is asymmetric with respect to $X$ and $Y$.

A key insight from Theorem \ref{thm:stratum-extension} is that the sign of collider bias conditioning on a specific level of a collider or its child in a structure with an embedded \V~or \Y~substructure equals the product of the sign of the embedded \V- or \Y-bias and the sign(s) of the extension path(s).
This result holds generally, regardless of the metric used to represent collider bias, provided the collider's parents are marginally independent.
For these structures, we do not present the formulas for collider bias for the \textsc{or} effect scale -- they are neither elegant nor enlightening.

\section{Bias due to linear regression adjustment for a collider or its child}\label{sec:lm}

\subsection{Bias due to linear regression adjustment for the collider in the \V~structure}

Now we turn our attention to collider bias due to linear regression.
Suppose, with the \V~structure, an analysis is conducted by fitting a linear model for $Y$ with $X$ and $C$ as predictors.
The coefficient of $X$ represents the association of $X$ and $Y$ adjusted for $C$.
With a binary $Y$, it is interpreted as a risk difference of $Y$ comparing the two levels of $X$, adjusted for $C$.
Based on a linear model result pointed out by \citet{Angrist1999} and \citet*[][page 142]{Morgan2007}, this coefficient is a weighted average of the two $C$-stratum-specific \textsc{rd}s (see Theorem \ref{thm:V}), where the weight of each is proportional to the product of the conditional variance of $X$ in the relevant $C$ stratum and the size of the stratum, $\mathrm{var}(X\mid C=c)\mathrm{P}(C=c)$.
This weighted average reduces to the formula in Theorem \ref{thm:Vlm}.

\begin{theorem}[Linear regression \V-bias theorem]\label{thm:Vlm}
\textup{\V}-bias due to linear regression adjustment for $C$ is given by
\begin{align*}
    \textup{\V-bias}(\textsc{lm})=~&-
    \left\{
    \begin{matrix*}[l]
		p_{X=1}(p_{C=1\mid11}-p_{C=1\mid10})+\\
		p_{X=0}(p_{C=1\mid01}-p_{C=1\mid00})
    \end{matrix*}
    \right\}
    \cdot
    \left\{
    \begin{matrix*}[l]
		p_{Y=1}(p_{C=1\mid11}-p_{C=1\mid01})+\\
		p_{Y=0}(p_{C=1\mid10}-p_{C=1\mid00})
    \end{matrix*}
    \right\}
    \times
    \\
    &\frac{p_{Y=1}p_{Y=0}}{
    \left\{
    \begin{matrix*}[l]
		p_{X=1}(p_{C=1\mid11}p_{Y=1}+p_{C=1\mid10}p_{Y=0})(p_{C=0\mid11}p_{Y=1}+p_{C=0\mid10}p_{Y=0})+\\
		p_{X=0}(p_{C=1\mid01}p_{Y=1}+p_{C=1\mid00}p_{Y=0})(p_{C=0\mid01}p_{Y=1}+p_{C=0\mid00}p_{Y=0})
    \end{matrix*}
    \right\}
    }.
\end{align*}
\end{theorem}

\medskip

The last term in the formula in Theorem \ref{thm:Vlm} is always positive, therefore the sign of \mbox{\V-bias(\textsc{lm})} is opposite the sign of the product of the other two terms, $\left\{
    \begin{matrix*}[l]
		p_{X=1}(p_{C=1\mid11}-p_{C=1\mid10})+\\
		p_{X=0}(p_{C=1\mid01}-p_{C=1\mid00})
    \end{matrix*}
    \right\}$
and
    $\left\{
    \begin{matrix*}[l]
		p_{Y=1}(p_{C=1\mid11}-p_{C=1\mid01})+\\
		p_{Y=0}(p_{C=1\mid10}-p_{C=1\mid00})
    \end{matrix*}
    \right\}$, 
which depend on both the conditional probabilities of $C$ and the marginal probabilities of $X$ and $Y$.
In two cases in the corollary below, the sign of \V-bias(\textsc{lm}) does not depend on the distribution of $X$ and $Y$.

\begin{corollary}[Linear regression \V-bias corollary]\label{crlry:V3}
In the \V~structure, if $X$ has positive effects on $C$ at both levels of $Y$ and $Y$ has positive effects on $C$ at both levels of $X$, or alternatively, if $X$ has negative effects on $C$ at both levels of $Y$ and $Y$ has negative effects on $C$ at both levels of $X$, then $\textup{\V-bias}(\textsc{lm})$ is negative.
On the other hand, if $X$ has positive effects on $C$ at both levels of $Y$ and $Y$ has negative effects on $C$ at both levels of $X$, or vice versa, then $\textup{\V-bias}(\textsc{lm})$ is positive.
\end{corollary}

\begin{figure}[h]
\caption{The sign of \V-bias due to linear regression adjustment for $C$}
\centering
\label{fig:5}
\begin{subfigure}[]{.45\textwidth}
\includegraphics[width=\textwidth]{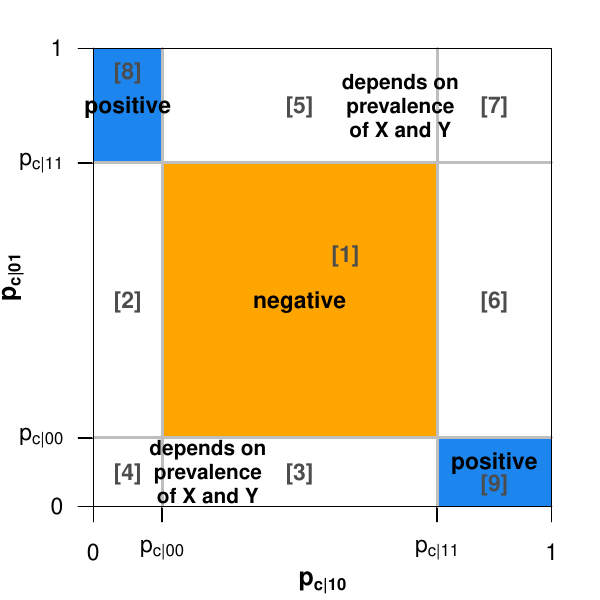}
\end{subfigure}
\begin{subfigure}[]{.3\textwidth}
\includegraphics[width=\textwidth]{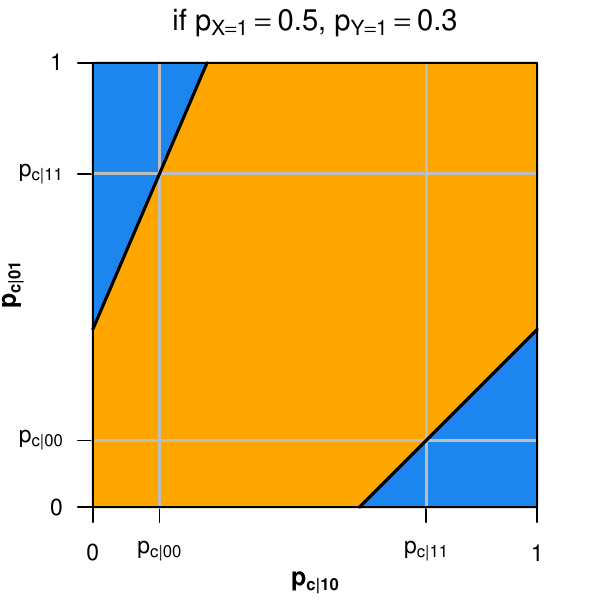}\\
\\
\includegraphics[width=\textwidth]{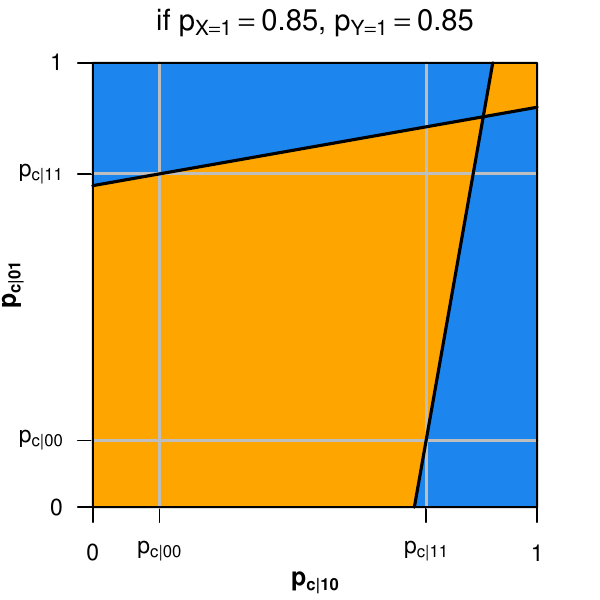}
\end{subfigure}
\caption*{\footnotesize Here $c$ is the level of $C$ such that $p_{c\mid11}>p_{c\mid00}$. There are two straight lines where \mbox{\V-bias(\textsc{lm})} is zero: one going through the $(p_{c\mid00},p_{c\mid11})$ point with slope $=(1-p_{Y=1})/p_{Y=1}$, the other going through the $(p_{c\mid11},p_{c\mid00})$ point with slope $=p_{X=1}/(1-p_{X=1})$. These two lines split the square into orange area(s) (where \V-bias(\textsc{lm}) is negative) and blue areas (where \V-bias(\textsc{lm}) is positive).}
\end{figure}

These results are shown in Fig. \ref{fig:5}.
\begin{itemize}
\item \add{The larger square on the left shows that the sign of \V-bias(\textsc{lm}) is (i) negative in the \textit{same-sign effects} region and (ii) positive in the \textit{opposite-sign effects} regions (Corollary \ref{crlry:V3}), and (iii) depends additionally on the distribution of $X$ and $Y$ in the \textit{qualitative interaction} regions.}
\item \add{For the \textit{qualitative interaction} regions, the two smaller squares on the right give two examples. In each example, there are two straight lines where \V-bias(\textsc{lm}) is zero: one going through the $(p_{c\mid00},p_{c\mid11})$ point with slope equal to the inverse odds of $Y$ (the collection of points where $X$ and $C$ are marginally independent), the other going through the $(p_{c\mid11},p_{c\mid00})$ point with slope equal to the odds of $X$ (the collection of points where $Y$ and $C$ are marginally independent). These two lines divide up the \textit{qualitative interaction} regions into areas with negative and positive \V-bias(\textsc{lm}) as shown in the figure.}
\end{itemize}

\subsection{Bias due to linear regression adjustment for the collider or its child in structures where the causes of the collider are marginally independent}

Theorem \ref{thm:general} below extends Theorem \ref{thm:Vlm} to cover all structures in Fig. \ref{fig:1} except \Vx.

\begin{theorem}[Linear regression collider bias general theorem]\label{thm:general}
Consider the eight structures \V, \M, \leftM, \rightM, \Y, \longM, \leftlongM~and \rightlongM.
Collider bias due to linear regression adjustment for $C$ in the \V, \M, \leftM~and \rightM~structures, and collider bias due to linear regression adjustment for $D$ in the \Y, \longM, \leftlongM~and \rightlongM~structures, can be expressed using one formula:
$$h(c)\cdot\textsc{rd}_\mathrm{left}\cdot\textsc{rd}_\mathrm{right}\cdot\textsc{rd}^2_\mathrm{child}\cdot\textsc{var}_\mathrm{left}\cdot\textsc{var}_\mathrm{right}\cdot\frac{1}{\phi(\mathrm{structure})},$$
where
\begin{itemize}
\item $h(c)=-
    \left\{
    \begin{matrix*}[l]
		p_{X=1}(p_{C=1\mid11}-p_{C=1\mid10})+\\
		p_{X=0}(p_{C=1\mid01}-p_{C=1\mid00})
    \end{matrix*}
    \right\}
    \cdot
    \left\{
    \begin{matrix*}[l]
		p_{Y=1}(p_{C=1\mid11}-p_{C=1\mid01})+\\
		p_{Y=0}(p_{C=1\mid10}-p_{C=1\mid00})
    \end{matrix*}
    \right\}$
for the structures in which $X$ and $Y$ are the causes of $C$ (\V~and \Y), and is the same function for the other structures except changing $X$ to $A$ if the left-side cause of $C$ is $A$ (\leftM, \leftlongM~and \longM) and changing $Y$ to $B$ if the right-side cause of $C$ is $B$ (\rightM, \rightlongM~and \longM);
\item $\textsc{rd}_\mathrm{left}$ is 1 for the structures without $A$ (\V, \rightM, \Y~and \rightlongM), and is the risk difference representing the effect of $A$ on $X$ (i.e., $p_{X=1\mid1}-p_{X=1\mid0}$) for the structures with $A$ (\M, \leftM, \longM~and \leftlongM);
\item $\textsc{rd}_\mathrm{right}$ is 1 for structures without $B$ (\V, \leftM, \Y~and \leftlongM), and is the risk difference representing the effect of $B$ on $Y$ (i.e., $p_{Y=1\mid1}-p_{Y=1\mid0}$) for structures with $B$ (\M, \rightM, \longM~and \rightlongM);
\item $\textsc{rd}_\mathrm{child}$ is 1 for structures without $D$ (\V, \M, \leftM~and \rightM), and is the risk difference representing the effect of $C$ on $D$ (i.e., $p_{D=1\mid1}-p_{D=1\mid0}$) for structures with $D$ (\Y, \longM, \leftlongM~and \rightlongM);
\item $\textsc{var}_\mathrm{left}$ is the variance of the left-side cause of $C$, $\textsc{var}_\mathrm{left}=p_{X=1}p_{X=0}$ if this cause is $X$ (\V, \rightM, \Y~and \rightlongM), $=p_{A=1}p_{A=0}$ if this cause is $A$ (\M, \leftM, \longM~and \leftlongM);
\item $\textsc{var}_\mathrm{right}$ is the variance of the right-side cause of $C$,
$\textsc{var}_\mathrm{right}=p_{Y=1}p_{Y=0}$ if this cause if $X$ (\V, \leftM, \Y~and \leftlongM), $=p_{B=1}p_{B=0}$ if this cause is $B$ (\M, \rightM, \longM~and \rightlongM);
\item $\phi(\textup{structure})=
\mathrm{P}(C=0)\mathrm{P}(C=1,X=1)\mathrm{P}(C=1,X=0)+
\mathrm{P}(C=1)\mathrm{P}(C=0,X=1)\mathrm{P}(C=0,X=0)$ for the \V, \M, \leftM~and \rightM~structures, and is the same function for the other structures except changing $C$ to $D$. The detailed expressions of this parameter for all the structures are provided in the Appendix.
\end{itemize}
\end{theorem}

Since $\phi(\mathrm{structure})$, $\textsc{rd}^2_\mathrm{child}$, $\textsc{var}_\mathrm{left}$ and $\textsc{var}_\mathrm{right}$ are all positive, the sign of collider bias due to linear regression adjustment for $C$ or $D$ in each of these structures is the product of (i) the sign of the embedded \V-bias due to linear regression for $C$ (equivalently, the sign of $h(c)$) and (ii) the sign(s) of the effect(s) of one or both of the causes of the collider on $X$ and/or $Y$, if $X$ and/or $Y$ are not the causes of the collider.
This is similar to the result about the sign of collider bias conditioning on a specific level of $C$ or $D$.
The difference is that the sign of \Y-bias due to linear regression adjustment is the same as the sign as the embedded \V-bias, and is not dependent on the effect of $C$ on $D$.

For the \V~structure, the result in Theorem \ref{thm:general} reduces to the result in Theorem \ref{thm:Vlm}.

\section{Discussion}\label{sec:discussion}

We have derived analytic results for collider bias due to conditioning on a specific level of, and due to linear regression adjustment for, a collider or a child of a collider in several structures of binary variables.
These results substantially extend the literature on collider bias.
The settings we focus on in this paper are represented here by simple causal DAGs, but encompass a broad class of causal DAGs where the variables of interest ($X$ and $Y$) are marginally independent ancestors of the collider, or are descendants of marginally independent ancestors of the collider.
For example, adding intermediate variables on any of the paths in the causal DAGs in Fig. \ref{fig:1} does not change the results, and replacing an arrow with a common cause between some pairs of variables can be treated as relabeling.
The results presented in this paper thus serve as the basis for understanding collider bias in a range of more complicated structures that may be encountered in practice.

Our paper assesses collider bias due to linear regression adjustment. Future research should evaluate collider bias due to logistic regression adjustment, which is commonly used for a binary outcome. In addition, future work could build on the current results to study collider bias in situations where $X$ and/or $Y$ are categorical variables. For example, based on the basic properties of covariance, it can be shown that with a binary $Y$ and an ordinal $X$, if collider bias (conditioning on a specific level of $C$ or $D$) is non-negative between $Y$ and all dichotomized versions of $X$, and is positive for some versions, then collider bias between $Y$ and $X$ is positive.

Our paper mostly focuses on collider bias in settings where $X$ and $Y$ are marginally independent. We relate our findings to only one situation where $X$ has a causal effect on $Y$, the \Vx~structure, to show that more insight can be gained about the sign of collider bias in that specific situation, adding to recently published results by \citet{jiang2016directions}. The current results should be extended by future work adding an effect of $X$ on $Y$ to the other structures, including \leftM, \rightM, \M, \Y, \leftlongM, \rightlongM~and \longM.

Last but not least, the structures addressed in this paper involve collider bias only, and are not affected by confounding bias.
Future work should investigate situations that involve both collider bias and confounding bias, which may be more realistic.

\section*{Acknowledgements}

TQN's work on this project was partially supported by NIDA grant T32-DA007292 (PI R. M. Johnson). The authors acknowledge helpful comments from two Referees and the Associate Editor.



\bibliographystyle{abbrvnat}
\bibliography{paper-ref}

\newpage
\section*{Appendix}\label{appx}

We go over some useful lemmas first, and then prove the theorems and corollaries presented in the paper.

\begin{lemma}\label{lm:1}
$E,F,G$ are binary variables, and $0<\mathrm{P}(G=1)<1$. For $g\in\{0,1\}$, $\mathrm{cov}(E,F\mid G=g)$ is equal to
\begin{align*}
	\frac{1}{\{\mathrm{P}(G=g)\}^2}\cdot
	\left\{
	\begin{matrix*}[l]
		\mathrm{P}(E=1,F=1,G=g)\mathrm{P}(E=0,F=0,G=g)-\\
		\mathrm{P}(E=1,F=0,G=g)\mathrm{P}(E=0,F=1,G=g)
	\end{matrix*}
	\right\}.
\end{align*}
\end{lemma}

\begin{proof}
\begin{align*}
	\mathrm{cov}(E,F\mid G=g)
	&=E(EF\mid G=g)-E(E\mid G=g)E(F\mid G=g)
	\\
	&=\mathrm{P}(E=1,F=1\mid G=g)-\mathrm{P}(E=1\mid G=g)\mathrm{P}(F=1\mid G=g)
	\\
	&=\frac{\mathrm{P}(E=1,F=1,G=g)}{\mathrm{P}(G=g)}-\\
	&~~~~~\frac{\mathrm{P}(E=1,F=1,G=g)+\mathrm{P}(E=1,F=0,G=g)}{\mathrm{P}(G=g)}\times\\
	&~~~~~\frac{\mathrm{P}(E=1,F=1,G=g)+\mathrm{P}(E=0,F=1,G=g)}{\mathrm{P}(G=g)}
	\\
	&=\frac{1}{\{\mathrm{P}(G=g)\}^2}\mathrm{P}(E=1,F=1,G=g)\mathrm{P}(G=g)-\\
	&~~~~~\frac{1}{\{\mathrm{P}(G=g)\}^2}\{\mathrm{P}(E=1,F=1,G=g)+\mathrm{P}(E=1,F=0,G=g)\}\times\\
	&~~~~~\{\mathrm{P}(E=1,F=1,G=g)+\mathrm{P}(E=0,F=1,G=g)\}.
\end{align*}
Since $\mathrm{P}(G=g)$ is equal to
\begin{align*}
\mathrm{P}(E=1,F=1,G=g)+\mathrm{P}(E=1,F=0,G=g)+
\mathrm{P}(E=0,F=1,G=g)+\mathrm{P}(E=0,F=0,G=g),
\end{align*}
\begin{align*}
\mathrm{cov}(E,F\mid G=g)
&=\frac{1}{\{\mathrm{P}(G=g)\}^2}
\left\{
\begin{matrix*}[l]
\mathrm{P}(E=1,F=1,G=g)\mathrm{P}(E=1,F=1,G=g)+\\
\mathrm{P}(E=1,F=1,G=g)\mathrm{P}(E=1,F=0,G=g)+\\
\mathrm{P}(E=1,F=1,G=g)\mathrm{P}(E=0,F=1,G=g)+\\
\mathrm{P}(E=1,F=1,G=g)\mathrm{P}(E=0,F=0,G=g)
\end{matrix*}
\right\}-\\
&~~~~~
\frac{1}{\{\mathrm{P}(G=g)\}^2}
\left\{
\begin{matrix*}[l]
\mathrm{P}(E=1,F=1,G=g)\mathrm{P}(E=1,F=1,G=g)+\\
\mathrm{P}(E=1,F=1,G=g)\mathrm{P}(E=0,F=1,G=g)+\\
\mathrm{P}(E=1,F=0,G=g)\mathrm{P}(E=1,F=1,G=g)+\\
\mathrm{P}(E=1,F=0,G=g)\mathrm{P}(E=0,F=1,G=g)
\end{matrix*}
\right\}
\\
&=\frac{1}{\{\mathrm{P}(G=g)\}^2}
\left\{
\begin{matrix*}[l]
\mathrm{P}(E=1,F=1,G=g)\mathrm{P}(E=0,F=0,G=g)-\\
\mathrm{P}(E=1,F=0,G=g)\mathrm{P}(E=0,F=1,G=g)
\end{matrix*}
\right\}
\end{align*}
~
\end{proof}

\begin{lemma}\label{lm:2}
For binary variables $E,F,G$, the following is true:
$$\mathrm{P}(E=1\mid F=1,G=g)-\mathrm{P}(E=1\mid F=0,G=g)=\frac{\mathrm{cov}(E,F\mid G=g)}{\mathrm{var}(F\mid G=g)}.$$
\end{lemma}

\begin{proof}
\begin{align*}
	\mathrm{P}&(E=1\mid F=1,G=g)-\mathrm{P}(E=1\mid F=0,G=g)=
	\\
	&=\frac{\mathrm{P}(E=1,F=1,G=g)}{\mathrm{P}(F=1,G=g)}-\frac{\mathrm{P}(E=1,F=0,G=g)}{\mathrm{P}(F=0,G=g)}
	\\
	&=\frac{
	\begin{bmatrix*}[l]
		\mathrm{P}(E=1,F=1,G=g)\{\mathrm{P}(E=1,F=0,G=g)+\mathrm{P}(E=0,F=0,G=g)\}-\\
		\mathrm{P}(E=1,F=0,G=g)\{\mathrm{P}(E=1,F=1,G=g)+\mathrm{P}(E=0,F=1,G=g)\}
	\end{bmatrix*}
	}{\mathrm{P}(F=1,G=g)\mathrm{P}(F=0,G=g)}
	\\
	&=\frac{
	\left\{
	\begin{matrix*}[l]
		\mathrm{P}(E=1,F=1,G=g)\mathrm{P}(E=0,F=0,G=g)-\\
		\mathrm{P}(E=1,F=0,G=g)\mathrm{P}(E=0,F=1,G=g)
	\end{matrix*}
	\right\}
	}{\{\mathrm{P}(G=g)\}^2\mathrm{P}(F=1\mid G=g)\mathrm{P}(F=0\mid G=g)}
	\\
	&=\mathrm{cov}(E,F\mid G=g)\times \frac{1}{\mathrm{P}(F=1\mid G=g)\mathrm{P}(F=0\mid G=g)}~~~\text{(by Lemma \ref{lm:1})}
	\\
	&=\frac{\mathrm{cov}(E,F\mid G=g)}{\mathrm{var}(F\mid G=g)}.
\end{align*}
\end{proof}

\begin{lemma}\label{lm:3}
\citet{Angrist1999} and \citet*[][page 142]{Morgan2007} pointed out that when linear regression is used to adjust an association (between predictor variable $X$ and dependent variable $Y$) for a covariate ($G$), the adjusted association is equivalent to the weighted average of the $G$-stratum-specific $X$--$Y$ associations, where the weight for stratum $G=g$ is proportion to $\mathrm{var}(X\mid G=g)\mathrm{P}(G=g)$. If $X$ and $G$ are binary, such weight can be expressed as:
$$w_{G=g}=\frac{\mathrm{P}(G=1-g)\mathrm{P}(X=1,G=g)\mathrm{P}(X=0,G=g)}{
\left\{
\begin{matrix*}[l]
	\mathrm{P}(G=0)\mathrm{P}(X=1,G=1)\mathrm{P}(X=0,G=1)+\\	
	\mathrm{P}(G=1)\mathrm{P}(X=1,G=0)\mathrm{P}(X=0,G=0)
\end{matrix*}
\right\}
}.$$
\end{lemma}

\begin{proof}
\begin{align*}
	\mathrm{var}(X\mid G=g)\mathrm{P}(G=g)
	&=\mathrm{P}(X=1\mid G=g)\mathrm{P}(X=0\mid G=g)\mathrm{P}(G=g)
	=\frac{\mathrm{P}(X=1,G=g)\mathrm{P}(X=0,G=g)}{\mathrm{P}(G=g)}
\end{align*}
\begin{align*}
	w_{G=g}
	&=\frac{\mathrm{var}(X\mid G=g)\mathrm{P}(G=g)}{\mathrm{var}(X\mid G=1)\mathrm{P}(G=1)+\mathrm{var}(X\mid G=0)\mathrm{P}(G=0)}
	\\
	&=\frac{\mathrm{P}(X=1,G=g)\mathrm{P}(X=0,G=g)/\mathrm{P}(G=g)}{
	\left\{
	\begin{matrix*}[l]
		\mathrm{P}(X=1,G=1)\mathrm{P}(X=0,G=1)/\mathrm{P}(G=1)+\\
		\mathrm{P}(X=1,G=0)\mathrm{P}(X=0,G=0)/\mathrm{P}(G=0)
	\end{matrix*}
	\right\}
	}
	\\
	&=\frac{\mathrm{P}(X=1,G=g)\mathrm{P}(X=0,G=g)/\mathrm{P}(G=g)}{
	\left\{
	\begin{matrix*}[l]
		\mathrm{P}(X=1,G=1)\mathrm{P}(X=0,G=1)/\mathrm{P}(G=1)+\\
		\mathrm{P}(X=1,G=0)\mathrm{P}(X=0,G=0)/\mathrm{P}(G=0)
	\end{matrix*}
	\right\}
	}\times
	\frac{\mathrm{P}(G=1)\mathrm{P}(G=0)}{\mathrm{P}(G=1)\mathrm{P}(G=0)}
	\\
	&=\frac{\mathrm{P}(G=1-g)\mathrm{P}(X=1,G=g)\mathrm{P}(X=0,G=g)}{
	\left\{
	\begin{matrix*}[l]
		\mathrm{P}(G=0)\mathrm{P}(X=1,G=1)\mathrm{P}(X=0,G=1)+\\
		\mathrm{P}(G=1)\mathrm{P}(X=1,G=0)\mathrm{P}(X=0,G=0)
	\end{matrix*}
	\right\}
	}.
\end{align*}
\end{proof}

\begin{lemma}\label{lm:4}
For binary variables $F$ and $G$, the following is true:
\begin{align*}
	\mathrm{P}(G=0)\mathrm{P}(G=1,F=1)\mathrm{P}(G=1,F=0)+
    \mathrm{P}(G=1)\mathrm{P}(G=0,F=1)\mathrm{P}(G=0,F=0)=
	\\
	\mathrm{P}(F=0)\mathrm{P}(F=1,G=1)\mathrm{P}(F=1,G=0)+
    \mathrm{P}(F=1)\mathrm{P}(F=0,G=1)\mathrm{P}(F=0,G=0).
\end{align*}
\end{lemma}

\begin{proof}
\begin{align*}
    &\mathrm{P}(G=0)\mathrm{P}(G=1,F=1)\mathrm{P}(G=1,F=0)+\mathrm{P}(G=1)\mathrm{P}(G=0,F=1)\mathrm{P}(G=0,F=0)
	\\
    &=\{\mathrm{P}(G=0,F=1)+\mathrm{P}(G=0,F=0)\}\mathrm{P}(G=1,F=1)\mathrm{P}(G=1,F=0)+\\
    &~~~~~~\{\mathrm{P}(G=1,F=1)+\mathrm{P}(G=1,F=0)\}\mathrm{P}(G=0,F=1)\mathrm{P}(G=0,F=0)
    \\
    &=\mathrm{P}(G=0,F=1)\mathrm{P}(G=1,F=1)\mathrm{P}(G=1,F=0)+\mathrm{P}(G=0,F=0)\mathrm{P}(G=1,F=1)\mathrm{P}(G=1,F=0)+\\
    &~~~~~~\mathrm{P}(G=1,F=1)\mathrm{P}(G=0,F=1)\mathrm{P}(G=0,F=0)+\mathrm{P}(G=1,F=0)\mathrm{P}(G=0,F=1)\mathrm{P}(G=0,F=0)
    \\
    &=\mathrm{P}(G=0,F=1)\mathrm{P}(G=1,F=1)[\mathrm{P}(G=1,F=0)+\mathrm{P}(G=0,F=0)]+\\
    &~~~~~~\mathrm{P}(G=0,F=0)\mathrm{P}(G=1,F=0)[\mathrm{P}(G=1,F=1)+\mathrm{P}(G=0,F=1)]
    \\
    &=\mathrm{P}(F=0)\mathrm{P}(F=1,G=1)\mathrm{P}(F=1,G=0)+\mathrm{P}(F=1)\mathrm{P}(F=0,G=1)\mathrm{P}(F=0,G=0).
\end{align*}
\end{proof}

\begin{lemma}\label{lm:5}

If $a,b,a'$ and $b'$ are real numbers that satisfy $0<a'<a<b<b'$ and $a'b'=ab$ then $a'+b'>a+b$.

\end{lemma}

\begin{proof}

Consider the two numbers $a,b$ with the number $c=\sqrt{ab}$. \add{Because $0<a<b$,} it is clear that $a<c<b$. Let $r$ denote $c/a$. It follows that $r>1$, $a=c/r$, $b=cr$ and $a+b=c\left(\frac{1}{r}+r\right)$. The relationship of $c$ with the pair $a',b'$ is similar: $c=\sqrt{a'b'}$ and $a'<c<b'$. Let $r'$ denote $c/a'$. It follows that $r'>r>1$, $a'=c/r'$, $b'=cr'$ and $a'+b'=c\left(\frac{1}{r'}+r'\right)$.
$$
(a'+b')-(a+b)=c\left(\frac{1}{r'}+r'\right)-c\left(\frac{1}{r}+r\right)=\frac{c(r'-r){(r'r-1)}}{r'r}>0,
$$
therefore $a'+b'>a+b$.
\end{proof}


\bigskip

\begin{proof}[\bf Proof of Theorem \ref{thm:V}]
\begin{align*}
    \V\textup{-bias}(C=c,\mathrm{cov})&=\mathrm{cov}(X,Y\mid C=c)\\
    &=\frac{1}{\{\mathrm{P}(C=c)\}^2}\cdot
    \left\{
    \begin{matrix*}[l]
		\mathrm{P}(X=1,Y=1,C=c)\mathrm{P}(X=0,Y=0,C=c)-\\
		\mathrm{P}(X=1,Y=0,C=c)\mathrm{P}(X=0,Y=1,C=c)
    \end{matrix*}
    \right\}~~~\text{(by Lemma \ref{lm:1})}\\
    &=\frac{1}{\{\mathrm{P}(C=c)\}^2}\cdot
    \left\{
    \begin{matrix*}[l]
		p_{X=1}p_{Y=1}p_{c\mid11}p_{X=0}p_{Y=0}p_{c\mid00}-\\
		p_{X=1}p_{Y=0}p_{c\mid10}p_{X=0}p_{Y=1}p_{c\mid01}
    \end{matrix*}
    \right\}
	\\
    &=\frac{p_{X=1}p_{X=0}p_{Y=1}p_{Y=0}}{\{\mathrm{P}(C=c)\}^2}\cdot
    (p_{c\mid11}p_{c\mid00}-p_{c\mid10}p_{c\mid01}).
    \\
    \V\textup{-bias}(C=c,\textsc{rd})
    &=\frac{\mathrm{cov}(X,Y\mid C=c)}{\mathrm{var}(X\mid C=c)}~~\text{(by Lemma \ref{lm:2})}
	\\
    &=\frac{p_{Y=1}p_{Y=0}p_{X=1}p_{X=0}\cdot(p_{c\mid11}p_{c\mid00}-p_{c\mid10}p_{c\mid01})}{\{\mathrm{P}(C=c)\}^2\mathrm{P}(X=1\mid C=c)\mathrm{P}(X=0\mid C=c)}
	\\
    &=\frac{p_{Y=1}p_{Y=0}}{\mathrm{P}(C=c\mid X=1)\mathrm{P}(C=c\mid X=0)}\cdot(p_{c\mid11}p_{c\mid00}-p_{c\mid10}p_{c\mid01})
	\\
    &=\frac{p_{Y=1}p_{Y=0}}{(p_{Y=1}p_{c\mid11}+p_{Y=0}p_{c\mid10})(p_{Y=1}p_{c\mid01}+p_{Y=0}p_{c\mid00})}\cdot(p_{c\mid11}p_{c\mid00}-p_{c\mid10}p_{c\mid01}).
    \\
    \V\textup{-bias}(C=c,\textsc{or})&=\frac{\mathrm{P}(Y=1\mid X=1,C=c)}{\mathrm{P}(Y=0\mid X=1,C=c)}\cdot\frac{\mathrm{P}(Y=0\mid X=0,C=c)}{\mathrm{P}(Y=1\mid X=0,C=c)}
	\\
    &=\frac{\mathrm{P}(Y=1,C=c\mid X=1)}{\mathrm{P}(Y=0,C=c\mid X=1)}\cdot\frac{\mathrm{P}(Y=0,C=c\mid X=0)}{\mathrm{P}(Y=1,C=c\mid X=0)}\\
    &=\frac{p_{c\mid11}p_{c\mid00}}{p_{c\mid10}p_{c\mid01}}.
\end{align*}
\end{proof}

\begin{proof}[\bf Proof of Corollary \ref{crlry:V1}]
\hfill

Consider the first scenario where $X$ has positive effects on $C$ at both levels of $Y$ and $Y$ has positive effects on $C$ at both levels of $X$. We need to show that \V-bias is negative conditioning on at least one level of $C$. The above-mentioned positive effects mean
$$
0\leq p_{C=1\mid00}<\{p_{C=1\mid10},p_{C=1\mid01}\}<p_{C=1\mid11}\leq 1.
$$
(The curly brackets around the pair $p_{C=1\mid10},p_{C=1\mid01}$ means that both of these probabilities are between the other two probabilities, without any information about their own order.)
In the special case where $p_{C=1\mid00}=0$, we have $g(1)<0$, which means \V-bias conditioning on $C=1$ is negative. We now consider the narrower condition
$$
0<p_{C=1\mid00}<\{p_{C=1\mid10},p_{C=1\mid01}\}<p_{C=1\mid11}\leq 1.
$$
\V-bias is negative conditioning on at least one level of $C$ means that if \V-bias is non-negative given one level of $C$, it must be negative conditioning on the other level, which translates to: if $g(1)\geq 0$ then $g(0)<0$, and if $g(0)\geq 0$ then $g(1)<0$. We only need to show proof for one of these two statements (say the former one); the proof for the other is its mirror image.

First, assume that $g(1)=0$, i.e., $p_{C=1\mid00}p_{C=1\mid11}=p_{C=1\mid10}p_{C=1\mid01}$. Referring to Lemma \ref{lm:5}, of these four conditional probabilities of $C$, we can think of $p_{C=1\mid00}$ as $a'$ and $p_{C=1\mid11}$ as $b'$, and the other two probabilities as $a,b$ in between them. This implies the inequality
$$
p_{C=1\mid00}+p_{C=1\mid11}>p_{C=1\mid10}+p_{C=1\mid01}.
$$
\begin{align*}
    p_{C=0\mid00}p_{C=0\mid11}
    &=(1-p_{C=1\mid00})(1-p_{C=1\mid11})\\
    &=1-(p_{C=1\mid00}+p_{C=1\mid11})+p_{C=1\mid00}p_{C=1\mid11}\\
    &=1-(p_{C=1\mid00}+p_{C=1\mid11})+p_{C=1\mid10}p_{C=1\mid01}~~~\text{(the original assumption)}\\
    &<1-(p_{C=1\mid10}+p_{C=1\mid01})+p_{C=1\mid10}p_{C=1\mid01}~~~\text{(the inequality above)}\\
    &=(1-p_{C=1\mid10})(1-p_{C=1\mid10})\\
    &=p_{C=0\mid10}p_{C=0\mid10},
\end{align*}
which means $g(0)<0$.

Second, assume instead that $g(1)>0$, i.e., $p_{C=1\mid00}p_{C=1\mid11}>p_{C=1\mid10}p_{C=1\mid01}$. Consider $p^*$ such  that $p^*p_{C=1\mid11}=p_{C=1\mid10}p_{C=1\mid01}$.
It follows that $0<p^*<p_{C=1\mid00}$. Now consider
$$
0<p*<\{p_{C=1\mid10},p_{C=1\mid01}\}<p_{C=1\mid11}~~~\text{and}~~~p^*p_{C=1\mid11}=p_{C=1\mid10}p_{C=1\mid01}.
$$
Using similar reasoning based on Lemma \ref{lm:5} as above, we arrive at the inequality
$$
(1-p^*)p_{C=0\mid11}<p_{C=0\mid10}p_{C=0\mid10}.
$$
On the other hand, $p^*<p_{C=1\mid00}$ implies that $p_{C=0\mid00}<(1-p^*)$.
Combining this with the inequality above, we have
$$
p_{C=0\mid00}p_{C=0\mid11}<p_{C=0\mid10}p_{C=0\mid10},
$$
which means $g(0)<0$. This completes the proof for the broad scenario where $X$ has positive effects on $C$ at both levels of $Y$ and $Y$ has positive effects on $C$ at both levels of $X$.

For the scenario where $X$ has negative effects on $C$ at both levels of $Y$ and $Y$ has negative effects on $C$ at both levels of $X$, we need only reverse-code both $X$ and $Y$ to arrive at the former scenario. Reverse-coding both variables does not change their sign of their association, so the result for the sign of \V-bias being negative for at least one level of $C$ also applies in this scenario.

In the third scenario where of $X$ and $Y$, one variable has positive effects while the other has negative effects on $C$, we reverse code the variable that has negative effects on $C$ to arrive at the first scenario. Reverse coding only one variable flips the sign of their association, so in this scenario \V-bias is positive for at least one level of $C$.
\end{proof}

\bigskip

\begin{proof}[\bf Proof of Corollary \ref{crlry:V2}]
\hfill

Qualitative interaction between $X$ and $Y$ on $C$ covers situations where one (or both) of the variables $X$ and $Y$ has the property that its effects on $C$ conditioning on the two levels of the other variable are of opposite signs.
It turns out we need only one of the variables $X$ and $Y$ to have this property, for \V-bias to be positive for one level of $C$ and negative for the other.
Without loss of generality, assume $X$ has this property.
Also without loss of generality, assume $X$ has a positive effect on $C$ when $Y=1$ and a negative effect on $C$ when $Y=0$ (if we switch the signs of these effects, similar reasoning applies).
This means
\begin{align*}
    p_{C=1\mid01}&<p_{C=1\mid11},\\
    p_{C=1\mid10}&<p_{C=1\mid00}.
\end{align*}
Combining these, we have $p_{C=1\mid01}p_{C=1\mid10}<p_{C=1\mid11}p_{C=1\mid00}$, which means $g(1)>0$.
The condition can also be re-expressed as
\begin{align*}
    p_{C=0\mid01}&>p_{C=0\mid11},\\
    p_{C=0\mid10}&>p_{C=0\mid00}.
\end{align*}
Combining these, we have $p_{C=0\mid01}p_{C=0\mid10}>p_{C=0\mid11}p_{C=0\mid00}$, which means $g(0)<0$.
That $g(1)$ and $g(0)$ are of opposite signs means that \V-bias takes on opposite signs conditioning on the two levels of $C$.
\end{proof}

\bigskip

\begin{proof}[\bf Proof of Theorem \ref{thm:Y}]
\hfill

By Lemma \ref{lm:1}, $\Y\textup{-bias}(D=d,\mathrm{cov})$ is equal to
$$\frac{1}{\{\mathrm{P}(D=d)\}^2}
\left\{
\begin{matrix*}[l]
\mathrm{P}(X=1,Y=1,D=d)\mathrm{P}(X=0,Y=0,D=d)-\\
\mathrm{P}(X=1,Y=0,D=d)\mathrm{P}(X=0,Y=1,D=d)
\end{matrix*}
\right\}.
$$
We expand the second term in this product as
\begin{align*}
&\begin{matrix*}[l]
\{\mathrm{P}(X=1,Y=1,C=1,D=d)+\mathrm{P}(X=1,Y=1,C=0,D=d)\}\times\\
\{\mathrm{P}(X=0,Y=0,C=1,D=d)+\mathrm{P}(X=0,Y=0,C=0,D=d)\}-\\
\{\mathrm{P}(X=1,Y=0,C=1,D=d)+\mathrm{P}(X=1,Y=0,C=0,D=d)\}\times\\
\{\mathrm{P}(X=0,Y=1,C=1,D=d)+\mathrm{P}(X=0,Y=1,C=0,D=d)\}
\end{matrix*}
\\
&=p_{X=1}p_{Y=1}p_{X=0}p_{Y=0}
\left\{
\begin{matrix*}[l](p_{C=1\mid11}p_{d\mid1}+p_{C=0\mid11}p_{d\mid0})
(p_{C=1\mid00}p_{d\mid1}+p_{C=0\mid00}p_{d\mid0})-\\
(p_{C=1\mid10}p_{d\mid1}+p_{C=0\mid10}p_{d\mid0})
(p_{C=1\mid01}p_{d\mid1}+p_{C=0\mid01}p_{d\mid0})
\end{matrix*}
\right\}
\\
&=p_{X=1}p_{Y=1}p_{X=0}p_{Y=0}
\left\{
\begin{matrix*}[l]
p_{d\mid1}^2(p_{C=1\mid11}p_{C=1\mid00}-p_{C=1\mid10}p_{C=1\mid01})+\\
p_{d\mid0}^2(p_{C=0\mid11}p_{C=0\mid00}-p_{C=0\mid10}p_{C=0\mid01})+\\
p_{d\mid1}p_{d\mid0}
\begin{pmatrix*}[l]
p_{C=1\mid11}p_{C=0\mid00}+p_{C=0\mid11}p_{C=1\mid00}-\\
p_{C=1\mid10}p_{C=0\mid01}-p_{C=0\mid10}p_{C=1\mid01}
\end{pmatrix*}
\end{matrix*}
\right\}
\\
&=p_{X=1}p_{Y=1}p_{X=0}p_{Y=0}
\left\{
\begin{matrix*}[l]
p^2_{d\mid1}(p_{C=1\mid11}p_{C=1\mid00}-p_{C=1\mid10}p_{C=1\mid01})+\\
p^2_{d\mid0}(p_{C=0\mid11}p_{C=0\mid00}-p_{C=0\mid10}p_{C=0\mid01})+\\
p_{d\mid1}p_{d\mid0}
\begin{pmatrix*}[l]
p_{C=1\mid11}-p_{C=1\mid11}p_{C=1\mid00}+p_{C=0\mid11}-p_{C=0\mid11}p_{C=0\mid00}-\\
p_{C=1\mid10}+p_{C=1\mid10}p_{C=1\mid01}-p_{C=0\mid10}+p_{C=0\mid10}p_{C=0\mid01}
\end{pmatrix*}
\end{matrix*}
\right\}
\\
&=p_{X=1}p_{Y=1}p_{X=0}p_{Y=0}
\left\{
\begin{matrix*}[l]
p^2_{d\mid1}(p_{C=1\mid11}p_{C=1\mid00}-p_{C=1\mid10}p_{C=1\mid01})+\\
p^2_{d\mid0}(p_{C=0\mid11}p_{C=0\mid00}-p_{C=0\mid10}p_{C=0\mid01})+\\
p_{d\mid1}p_{d\mid0}
\begin{pmatrix*}[l]
-p_{C=1\mid11}p_{C=1\mid00}-p_{C=0\mid11}p_{C=0\mid00}+\\
p_{C=1\mid10}p_{C=1\mid01}+p_{C=0\mid10}p_{C=0\mid01}
\end{pmatrix*}
\end{matrix*}
\right\}
\\
&=p_{X=1}p_{X=0}p_{Y=1}p_{Y=0}\cdot
(p_{d\mid1}-p_{d\mid0})\cdot
\left\{
\begin{matrix*}[l]
p_{d\mid1}(p_{C=1\mid11}p_{C=1\mid00}-p_{C=1\mid10}p_{C=1\mid01})-\\
p_{d\mid0}(p_{C=0\mid11}p_{C=0\mid00}-p_{C=0\mid10}p_{C=0\mid01})
\end{matrix*}
\right\}.
\end{align*}
Therefore,
\begin{align*}
    \Y\textup{-bias}(D=d,\mathrm{cov})
    &=\frac{p_{X=1}p_{X=0}p_{Y=1}p_{Y=0}}{\{\mathrm{P}(D=d)\}^2}\cdot
    (p_{d\mid1}-p_{d\mid0})\cdot\left\{
    \begin{matrix*}[l]
        p_{d\mid1}(p_{C=1\mid00}p_{C=1\mid11}-p_{C=1\mid10}p_{C=1\mid01})-\\
        p_{d\mid0}(p_{C=0\mid00}p_{C=0\mid11}-p_{C=0\mid10}p_{C=0\mid01})
    \end{matrix*}
    \right\}.
    \end{align*}
By Lemma \ref{lm:2},
\begin{align*}
\Y\textup{-bias}(D=d,\textsc{rd})
=\frac{\Y\textup{-bias}(D=d,\mathrm{cov})}{\mathrm{var}(X\mid D=d)}
=\frac{\Y\textup{-bias}(D=d,\mathrm{cov})\{\mathrm{P}(D=d)\}^2}{\mathrm{P}(X=1\mid D=d)\mathrm{P}(X=0\mid D=d)\{\mathrm{P}(D=d)\}^2}.
\end{align*}
The denominator in this expression is equal to
\begin{align*}
&p_{X=1}p_{X=0}\cdot\mathrm{P}(D=d\mid X=1)\mathrm{P}(D=d\mid X=0)\\
&=p_{X=1}p_{X=0}\cdot
\left\{
    \begin{matrix*}[l]
		p_{Y=1}(p_{C=1\mid11}p_{d\mid1}+p_{C=0\mid11}p_{d\mid0})+\\
		p_{Y=0}(p_{C=1\mid10}p_{d\mid1}+p_{C=0\mid10}p_{d\mid0})
    \end{matrix*}
    \right\}
    \cdot
    \left\{
    \begin{matrix*}[l]
		p_{Y=1}(p_{C=1\mid01}p_{d\mid1}+p_{C=0\mid01}p_{d\mid0})+\\
		p_{Y=0}(p_{C=1\mid00}p_{d\mid1}+p_{C=0\mid00}p_{d\mid0})
    \end{matrix*}
    \right\}.
\end{align*}
Therefore, 
\begin{align*}
\Y\textup{-bias}(D=d,\textsc{rd})=
    \frac{p_{Y=1}p_{Y=0}\cdot
    (p_{d\mid1}-p_{d\mid0})\cdot\left\{
    \begin{matrix*}[l]
        p_{d\mid1}(p_{C=1\mid00}p_{C=1\mid11}-p_{C=1\mid10}p_{C=1\mid01})-\\
        p_{d\mid0}(p_{C=0\mid00}p_{C=0\mid11}-p_{C=0\mid10}p_{C=0\mid01})
    \end{matrix*}
    \right\}
    }{
    \left\{
    \begin{matrix*}[l]
		p_{Y=1}(p_{C=1\mid11}p_{d\mid1}+p_{C=0\mid11}p_{d\mid0})+\\
		p_{Y=0}(p_{C=1\mid10}p_{d\mid1}+p_{C=0\mid10}p_{d\mid0})
    \end{matrix*}
    \right\}
    \cdot
    \left\{
    \begin{matrix*}[l]
		p_{Y=1}(p_{C=1\mid01}p_{d\mid1}+p_{C=0\mid01}p_{d\mid0})+\\
		p_{Y=0}(p_{C=1\mid00}p_{d\mid1}+p_{C=0\mid00}p_{d\mid0})
    \end{matrix*}
    \right\}
    }.
\end{align*}

\begin{align*}
\Y\textup{-bias}(D=d,\textsc{or})
&=\frac{\mathrm{P}(Y=1\mid X=1,D=d)\mathrm{P}(Y=0\mid X=0,D=d)}{\mathrm{P}(Y=1\mid X=0,D=d)\mathrm{P}(Y=0\mid X=1,D=d)}
\\
&=\frac{\mathrm{P}(D=d\mid X=1,Y=1)\mathrm{P}(D=d\mid X=0,Y=0)}{\mathrm{P}(D=d\mid X=1,Y=0)\mathrm{P}(D=d\mid X=0,Y=1)}
\\
&=\frac{(p_{d\mid1}p_{C=1\mid11}+p_{d\mid0}p_{C=0\mid11})(p_{d\mid1}p_{C=1\mid00}+p_{d\mid0}p_{C=0\mid00})}{(p_{d\mid1}p_{C=1\mid10}+p_{d\mid0}p_{C=0\mid10})(p_{d\mid1}p_{C=1\mid01}+p_{d\mid0}p_{C=0\mid01})}
\\
&=\frac{\{(p_{d\mid1}-p_{d\mid0})p_{C=1\mid11}+p_{d\mid0}\}\{(p_{d\mid1}-p_{d\mid0})p_{C=1\mid00}+p_{d\mid0}\}}{\{(p_{d\mid1}-p_{d\mid0})p_{C=1\mid10}+p_{d\mid0}\}\{(p_{d\mid1}-p_{d\mid0})p_{C=1\mid01}+p_{d\mid0}\}}
\\
&=\frac{(p_{d\mid1}-p_{d\mid0})^2p_{C=1\mid11}p_{C=1\mid00}+(p_{d\mid1}-p_{d\mid0})p_{d\mid0}(p_{C=1\mid11}+p_{C=1\mid00})+p_{d\mid0}^2}
{(p_{d\mid1}-p_{d\mid0})^2p_{C=1\mid10}p_{C=1\mid01}+(p_{d\mid1}-p_{d\mid0})p_{d\mid0}(p_{C=1\mid10}+p_{C=1\mid01})+p_{d\mid0}^2}
\\
&=\frac{(p_{d\mid1}-p_{d\mid0})\{p_{d\mid1}p_{C=1\mid11}p_{C=1\mid00}+p_{d\mid0}(1-p_{C=0\mid11}p_{C=0\mid00})\}+p_{d\mid0}^2}
{(p_{d\mid1}-p_{d\mid0})\{p_{d\mid1}p_{C=1\mid10}p_{C=1\mid01}+p_{d\mid0}(1-p_{C=0\mid10}p_{C=0\mid01})\}+p_{d\mid0}^2}
\\
&=\frac{(p_{d\mid1}-p_{d\mid0})(p_{d\mid1}p_{C=1\mid11}p_{C=1\mid00}-p_{d\mid0}p_{C=0\mid11}p_{C=0\mid00})+p_{d\mid1}p_{d\mid0}}
{(p_{d\mid1}-p_{d\mid0})(p_{d\mid1}p_{C=1\mid10}p_{C=1\mid01}-p_{d\mid0}p_{C=0\mid10}p_{C=0\mid01})+p_{d\mid1}p_{d\mid0}}
\end{align*}
\end{proof}

\begin{proof}[\bf Proof of Corollary \ref{crlry:Y}]

\begin{align*}
    \mathrm{cov}(X,Y\mid D=d)\
    &=\frac{p_{X=1}p_{X=0}p_{Y=1}p_{Y=0}}{\{\mathrm{P}(D=d)\}^2}\times(p_{d\mid1}-p_{d\mid0})\times\{p_{d\mid1}g(1)-p_{d\mid0}g(0)\}\\
    &=\frac{(p_{d\mid1}-p_{d\mid0})}{\{\mathrm{P}(D=d)\}^2}\times
    \begin{bmatrix*}[l]
        p_{d\mid1}\{\mathrm{P}(C=1)\}^2\frac{p_{X=1}p_{X=0}p_{Y=1}p_{Y=0}g(1)}{\{\mathrm{P}(C=1)\}^2}-\\
        p_{d\mid0}\{\mathrm{P}(C=0)\}^2\frac{p_{X=1}p_{X=0}p_{Y=1}p_{Y=0}g(0)}{\{\mathrm{P}(C=0)\}^2}
    \end{bmatrix*}\\
    &=\frac{(p_{d\mid1}-p_{d\mid0})}{\{\mathrm{P}(D=d)\}^2}\times
    \begin{bmatrix*}[l]
    p_{d\mid1}\{\mathrm{P}(C=1)\}^2\mathrm{cov}(X,Y\mid C=1)-\\
    p_{d\mid0}\{\mathrm{P}(C=0)\}^2\mathrm{cov}(X,Y\mid C=0)
    \end{bmatrix*}.
\end{align*}
~
\end{proof}

\bigskip

\begin{proof}[\bf Proof of Theorem \ref{thm:stratum-extension}]
\hfill

First, consider the \rightM~structure. By Lemma \ref{lm:1}, \rightM\textup{-bias}$(C=c,\mathrm{cov})$ is equal to
$$
\frac{1}{\{\mathrm{P}(C=c)\}^2}\cdot
    \left\{
    \begin{matrix*}[l]
    \mathrm{P}(X=1,Y=1,C=c)\mathrm{P}(X=0,Y=0,C=c)-\\
    \mathrm{P}(X=1,Y=0,C=c)\mathrm{P}(X=0,Y=1,C=c)
    \end{matrix*}
    \right\}.
$$
We expand the second term in this product as
\begin{align*}
    &\{\mathrm{P}(X=1,B=1,Y=1,C=c)+\mathrm{P}(X=1,B=0,Y=1,C=c)\}\times\\
    &\{\mathrm{P}(X=0,B=1,Y=0,C=c)+\mathrm{P}(X=0,B=0,Y=0,C=c)\}-\\
    &\{\mathrm{P}(X=1,B=1,Y=0,C=c)+\mathrm{P}(X=1,B=0,Y=0,C=c)\}\times\\
    &\{\mathrm{P}(X=0,B=1,Y=1,C=c)+\mathrm{P}(X=0,B=0,Y=1,C=c)\}
    \\
    =&\{p_{Y=1\mid1}\mathrm{P}(X=1,B=1,C=c)+p_{Y=1\mid0}\mathrm{P}(X=1,B=0,C=c)\}\times\\
    &\{p_{Y=0\mid1}\mathrm{P}(X=0,B=1,C=c)+p_{Y=0\mid0}\mathrm{P}(X=0,B=0,C=c)\}-\\
    &\{p_{Y=0\mid1}\mathrm{P}(X=1,B=1,C=c)+p_{Y=0\mid0}\mathrm{P}(X=1,B=0,C=c)\}\times\\
    &\{p_{Y=1\mid1}\mathrm{P}(X=0,B=1,C=c)+p_{Y=1\mid0}\mathrm{P}(X=0,B=0,C=c)\}
    \\
    =&p_{Y=1\mid1}p_{Y=0\mid0}\mathrm{P}(X=1,B=1,C=c)\mathrm{P}(X=0,B=0,C=c)+\\
    &p_{Y=1\mid0}p_{Y=0\mid1}\mathrm{P}(X=1,B=0,C=c)\mathrm{P}(X=0,B=1,C=c)+\\
    &p_{Y=1\mid0}p_{Y=0\mid1}\mathrm{P}(X=1,B=1,C=c)\mathrm{P}(X=0,B=0,C=c)-\\
    &p_{Y=1\mid1}p_{Y=0\mid0}\mathrm{P}(X=1,B=0,C=c)\mathrm{P}(X=0,B=1,C=c)
    \\
    =&
    \left\{
    \begin{matrix*}[l]
    \mathrm{P}(X=1,B=1,C=c)\mathrm{P}(X=0,B=0,C=c)-\\
    \mathrm{P}(X=1,B=0,C=c)\mathrm{P}(X=0,B=1,C=c)
    \end{matrix*}
    \right\}
    \cdot(p_{Y=1\mid1}p_{Y=0\mid0}-p_{Y=1\mid0}p_{Y=0\mid1})
    \\
    =&
    \left\{
    \begin{matrix*}[l]
    \mathrm{P}(X=1,B=1,C=c)\mathrm{P}(X=0,B=0,C=c)-\\
    \mathrm{P}(X=1,B=0,C=c)\mathrm{P}(X=0,B=1,C=c)
    \end{matrix*}
    \right\}
    (p_{Y=1\mid1}-p_{Y=1\mid0}).
\end{align*}
It follows that
\begin{align*}
    \rightM\textup{-bias}(C=c,\mathrm{cov})
    &=\frac{1}{\{\mathrm{P}(C=c)\}^2}\cdot
    \left\{
    \begin{matrix*}[l]
    \mathrm{P}(X=1,B=1,C=c)\mathrm{P}(X=0,B=0,C=c)-\\
    \mathrm{P}(X=1,B=0,C=c)\mathrm{P}(X=0,B=1,C=c)
    \end{matrix*}
    \right\}
    \cdot\textsc{rd}_\textup{right}
    \\
    &=\V\textup{-bias-em}(C=c,\mathrm{cov})\cdot\textsc{rd}_\textup{right}~~~\text{(by Lemma \ref{lm:1})}.\\
    \rightM\textup{-bias}(C=c,\mathrm{rd})
    &=\frac{\mathrm{cov}(X,Y\mid C=c)}{\mathrm{var}(X\mid C=c)}~~~\text{(by Lemma \ref{lm:2})}
    \\
    &=\frac{\V\textup{-bias-em}(C=c,\mathrm{cov})\cdot\textsc{rd}_\textup{right}}{\mathrm{var}(X\mid C=c)}~~~\text{(based on result immediately above)}
    \\
    &=\V\textup{-bias-em}(C=c,\mathrm{rd})\cdot\textsc{rd}_\textup{right}~~~\text{(by Lemma \ref{lm:2})}.
\end{align*}
Next, consider the \leftM~structure. The proof for
$$\leftM\textup{-bias}(C=c,\mathrm{cov})=\textsc{rd}_\textup{left}\cdot\V\textup{-bias-em}(C=c,\mathrm{cov})$$
is similar to the proof for \rightM-bias$(C=c,\mathrm{cov})$, because the \leftM~structure is a mirror image of the \rightM~structure, and covariance is symmetric.
We derive \leftM-bias as conditional \textsc{rd}:
\begin{align*}
    \leftM\textup{-bias}(C=c,\textsc{rd})
    &=\frac{\leftM\textup{-bias}(C=c,\mathrm{cov})}{\mathrm{var}(X\mid C=c)}~~~\text{(by Lemma \ref{lm:2})}
    \\
    &=\frac{\textsc{rd}_\textup{left}\cdot\V\textup{-bias-em}(C=c,\mathrm{cov})}{\mathrm{var}(X\mid C=c)}
    \\
    &=\textsc{rd}_\textup{left}\cdot\frac{\V\textup{-bias-em}(C=c,\mathrm{cov})}{\mathrm{var}(A\mid C=c)}\cdot\frac{\mathrm{var}(A\mid C=c)}{\mathrm{var}(X\mid C=c)}
    \\
    &=\textsc{rd}_\textup{left}\cdot\V\textup{-bias-em}(C=c,\textsc{rd})\cdot\textsc{vr}(c)~~~\text{(by Lemma \ref{lm:2})}.
\end{align*}

The proof for \M-bias is a trivial extension of the proofs for \rightM- and \leftM-bias.
The proofs for \rightlongM-, \leftlongM- and \longM-bias are almost exactly the same as the proofs for \rightM-, \leftM- and \M-bias, respectively, except replacing $C=c$ with $D=d$.
\end{proof}

\bigskip

\begin{proof}[\bf Proof of Theorem \ref{thm:Vlm}]
\hfill

By Lemma 3, the weights that average \V-bias$(C=1,\textsc{rd})$ and \V-bias$(C=0,\textsc{rd})$ to \V-bias(\textsc{lm}) are $w_{C=1}$ and $w_{C=0}$ with the form
$$
w_{C=c}
=\frac{\mathrm{P}(C=1-c)\mathrm{P}(C=c,X=1)\mathrm{P}(C=c,X=0)}{
	\left\{
    \begin{matrix*}[l]
        \mathrm{P}(C=0)\mathrm{P}(C=1,X=1)\mathrm{P}(C=1,X=0)+\\
        \mathrm{P}(C=1)\mathrm{P}(C=0,X=1)\mathrm{P}(C=0,X=0)
    \end{matrix*}
    \right\}
}.
$$
We can rewrite \V-bias$(C=c,\textsc{rd})$ from Theorem \ref{thm:V} as
$$\V\text{-bias}(C=c,\textsc{rd})    =\frac{p_{X=1}p_{X=0}p_{Y=1}p_{Y=0}g(c)}{\mathrm{P}(C=c,X=1)\mathrm{P}(C=c,X=0)}.$$
Combining these with the weights, we have
$$
\V\text{-bias}(\textsc{lm})=
\frac{p_{X=1}p_{X=0}p_{Y=1}p_{Y=0}\{\mathrm{P}(C=0)g(1)+\mathrm{P}(C=1)g(0)\}}{
	\left\{
    \begin{matrix*}[l]
        \mathrm{P}(C=0)\mathrm{P}(C=1,X=1)\mathrm{P}(C=1,X=0)+\\
        \mathrm{P}(C=1)\mathrm{P}(C=0,X=1)\mathrm{P}(C=0,X=0)
    \end{matrix*}
    \right\}
}.
$$
We tackle the numerator and denominator separately. The numerator includes the term
\begin{align*}
    \mathrm{P}&(C=0)g(1)+\mathrm{P}(C=1)g(0)
    =g(1)-\mathrm{P}(C=1)[g(1)-g(0)]\\
    =&(p_{C=1\mid11}p_{C=1\mid00}-p_{C=1\mid10}p_{C=1\mid01})-
    \begin{pmatrix*}[l]
    p_{X=1}p_{Y=1}p_{C=1\mid11}+\\
    p_{X=1}p_{Y=0}p_{C=1\mid10}+\\
    p_{X=0}p_{Y=1}p_{C=1\mid01}+\\
    p_{X=0}p_{Y=0}p_{C=1\mid00}
    \end{pmatrix*}
    \cdot(p_{C=1\mid11}+p_{C=1\mid00}-p_{C=1\mid10}-p_{C=1\mid01})
    \\
    =&-p^2_{C=1\mid00}(1-p_{X=1})(1-p_{Y=1})-p^2_{C=1\mid11}p_{X=1}p_{Y=1}+\\
    &p^2_{C=1\mid10}p_{X=1}(1-p_{Y=1})+p^2_{C=1\mid01}(1-p_{X=1})p_{Y=1}+\\
    &p_{C=1\mid00}p_{C=1\mid11}(p_{X=1}+p_{Y=1}-2p_{X=1}p_{Y=1})+\\
    &p_{C=1\mid10}p_{C=1\mid01}(-1+p_{X=1}+p_{Y=1}-2p_{X=1}p_{Y=1})+\\
    &p_{C=1\mid00}p_{C=1\mid10}(1-2p_{X=1})(1-p_{Y=1})+p_{C=1\mid00}p_{C=1\mid01}(1-p_{X=1})(1-2p_{Y=1})+\\
    &p_{C=1\mid10}p_{C=1\mid11}p_{X=1}(2p_{Y=1}-1)+p_{C=1\mid01}p_{C=1\mid11}(2p_{X=1}-1)p_{Y=1}
    \\
    =&-\left\{
    \begin{matrix*}[l]
    p_{X=1}(p_{C=1\mid11}-p_{C=1\mid10})+\\
    p_{X=0}(p_{C=1\mid01}-p_{C=1\mid00})
    \end{matrix*}
    \right\}
    \cdot
    \left\{
    \begin{matrix*}[l]
    p_{Y=1}(p_{C=1\mid11}-p_{C=1\mid01})+\\
    p_{Y=0}(p_{C=1\mid10}-p_{C=1\mid00})
    \end{matrix*}
    \right\}.
\end{align*}
By Lemma 4, the denominator is equal to
\begin{align*}
   &\mathrm{P}(X=0)\mathrm{P}(X=1,C=1)\mathrm{P}(X=1,C=0)+
    \mathrm{P}(X=1)\mathrm{P}(X=0,C=1)\mathrm{P}(X=0,C=0)
    \\
    &=p^2_{X=1}p_{X=0}\mathrm{P}(C=1\mid X=1)\mathrm{P}(C=0\mid X=1)+
    p_{X=1}p^2_{X=0}\mathrm{P}(C=1\mid X=0)\mathrm{P}(C=0\mid X=0)
    \\
    &=p_{X=1}p_{X=0}\left\{
    \begin{matrix*}[l]
        p_{X=1}(p_{Y=1}p_{C=1\mid11}+p_{Y=0}p_{C=1\mid10})(p_{Y=1}p_{C=0\mid11}+p_{Y=0}p_{C=0\mid10})+\\
        p_{X=0}(p_{Y=1}p_{C=1\mid01}+p_{Y=0}p_{C=1\mid00})(p_{Y=1}p_{C=0\mid01}+p_{Y=0}p_{C=0\mid00})
    \end{matrix*}
    \right\}.
\end{align*}
Putting these results together, we have the result in Theorem \ref{thm:Vlm}.
\end{proof}

\begin{proof}[Proof of Corollary \ref{crlry:V3}]

Consider the scenario where $X$ has positive effects on $C$ at both levels of $Y$ and $Y$ has positive effects on $C$ at both levels of $X$. This means
\begin{align*}
    p_{C=1\mid11}-p_{C=1\mid10}>0,\\
    p_{C=1\mid01}-p_{C=1\mid00}>0,\\
    p_{C=1\mid11}-p_{C=1\mid01}>0,\\
    p_{C=1\mid10}-p_{C=1\mid00}>0.
\end{align*}
It follows that
\begin{align*}
    h(c)=-
    \left\{
    \begin{matrix*}[l]
		p_{X=1}(p_{C=1\mid11}-p_{C=1\mid10})+\\
		p_{X=0}(p_{C=1\mid01}-p_{C=1\mid00})
    \end{matrix*}
    \right\}
    \times
    \left\{
    \begin{matrix*}[l]
		p_{Y=1}(p_{C=1\mid11}-p_{C=1\mid01})+\\
		p_{Y=0}(p_{C=1\mid10}-p_{C=1\mid00})
    \end{matrix*}
    \right\}
    <0,
\end{align*}
therefore \V-bias(\textsc{lm}) is negative.

The proofs for the other scenarios are similar.
\end{proof}

\bigskip

\begin{proof}[\bf Proof of Theorem \ref{thm:general}]
\hfill

We first prove Theorem \ref{thm:general} for structures \V, \Y, \leftM, \leftlongM, and then extend to \rightM, \M, \rightlongM, \longM. To be concise, we will stop at the general form of $\phi(\textup{structure})$, except for $\phi(\V)$, which has been derived in the proof of Theorem \ref{thm:Vlm}. We include the derivation of all other $\phi(\textup{structure})$ detailed expressions in a stand-alone section following this proof.\\

\noindent\underline{\V~structure}: Based on the proof of Theorem \ref{thm:Vlm},
\begin{align*}
	\phi(\V)&=\mathrm{P}(C=0)\mathrm{P}(C=1,X=1)\mathrm{P}(C=1,X=0)+
    \mathrm{P}(C=1)\mathrm{P}(C=0,X=1)\mathrm{P}(C=0,X=0)
    \\
    &=p_{X=1}p_{X=0}\cdot
	\left\{
    \begin{matrix*}[l]
        p_{X=1}(p_{Y=1}p_{C=1\mid11}+p_{Y=0}p_{C=1\mid10})(p_{Y=1}p_{C=0\mid11}+p_{Y=0}p_{C=0\mid10})+\\
        p_{X=0}(p_{Y=1}p_{C=1\mid01}+p_{Y=0}p_{C=1\mid00})(p_{Y=1}p_{C=0\mid01}+p_{Y=0}p_{C=0\mid00})
    \end{matrix*}
    \right\}.
\end{align*}
As $\textsc{var}_\text{left}=p_{X=1}p_{X=0}$, $\textsc{var}_\text{right}=p_{Y=1}p_{Y=0}$, the result from Theorem \ref{thm:Vlm} translates to
$$
\V\text{-bias}(\textsc{lm})=h(c)\cdot\textsc{var}_\text{left}\cdot\textsc{var}_\text{right}\cdot 1/\phi(\V).
$$

\noindent\underline{\Y~structure}: By Lemma \ref{lm:3} and the definition of $\phi$ given in the Appendix, the weights that average \Y-bias$(D=1,\textsc{rd})$ and \Y-bias$(D=0,\textsc{rd})$ to \Y-bias(\textsc{lm}) take the form
\begin{align*}
w_{D=d}
&=\mathrm{P}(D=1-d)\mathrm{P}(D=d,X=1)\mathrm{P}(D=d,X=0)/\phi(\Y).
\end{align*}
The proof of Theorem \ref{thm:Y} shows that the denominator in the \Y-bias$(D=d,\textsc{rd})$ formula in Theorem \ref{thm:Y} is equal to $\mathrm{P}(D=d\mid X=1)\mathrm{P}(D=d\mid X=0)$.
We thus rewrite the \Y-bias$(D=d,\textsc{rd})$ formula as
\begin{align*}
\Y\text{-bias}(D=d,\textsc{rd})    
&=\frac{p_{X=1}p_{X=0}p_{Y=1}p_{Y=0}\cdot(p_{d\mid1}-p_{d\mid0})\cdot\{p_{d\mid1}g(1)-p_{d\mid0}g(0)\}}{\mathrm{P}(D=d,X=1)\mathrm{P}(D=d,X=0)}\\
&=\frac{\textsc{var}_\textup{left}\cdot\textsc{var}_\textup{right}\cdot(p_{d\mid1}-p_{d\mid0})\cdot\{p_{d\mid1}g(1)-p_{d\mid0}g(0)\}}{\mathrm{P}(D=d,X=1)\mathrm{P}(D=d,X=0)}.
\end{align*}
Combining these with the weights, we have
\begin{align*}
\Y\text{-bias}(\textsc{lm})&=
\frac{\textsc{var}_\textup{left}\cdot\textsc{var}_\textup{right}\cdot(p_{D=1\mid1}-p_{D=1\mid0})}{\phi(\Y)}\times\\
&~~~~~[
\mathrm{P}(D=0)\{p_{D=1\mid1}g(1)-p_{D=1\mid0}g(0)\}+
\mathrm{P}(D=1)\{p_{D=0\mid0}g(0)-p_{D=0\mid1}g(1)\}
]
\end{align*}
where
\begin{align*}
    &\mathrm{P}(D=0)\{p_{D=1\mid1}g(1)-p_{D=1\mid0}g(0)\}+\mathrm{P}(D=1)\{p_{D=0\mid0}g(0)-p_{D=0\mid1}g(1)\}=\\
    &=\{p_{D=1\mid1}g(1)-p_{D=1\mid0}g(0)\}+\mathrm{P}(D=1)\{g(0)-g(1)\}\\
    &=p_{D=1\mid1}g(1)-p_{D=1\mid0}g(0)+(\mathrm{P}(C=1)p_{D=1\mid1}+\mathrm{P}(C=0)p_{D=1\mid0})\{g(0)-g(1)\}\\
    &=(p_{D=1\mid1}-p_{D=1\mid0})\{\mathrm{P}(C=0)g(1)+\mathrm{P}(C=1)g(0)\}\\
    &=(p_{D=1\mid1}-p_{D=1\mid0})h(c).
\end{align*}
Therefore,
\begin{align*}
\Y\text{-bias}(\textsc{lm})
&=\frac{\textsc{var}_\textup{left}\cdot\textsc{var}_\textup{right}\cdot(p_{D=1\mid1}-p_{D=1\mid0})^2\cdot h(c)}{\phi(\Y)}
=h(c)\cdot\textsc{rd}_\textup{child}^2\cdot\textsc{var}_\textup{left}\cdot\textsc{var}_\textup{right}\cdot 1/\phi(\Y).
\end{align*}

\noindent\underline{\leftM~structure}: By Lemma \ref{lm:3}, the weights that average \leftM-bias$(C=1,\textsc{rd})$ and \leftM-bias$(C=0,\textsc{rd})$ to \leftM-bias(\textsc{lm}) take the form
\begin{align*}
w_{C=c}
=\frac{\mathrm{P}(C=1-c)\mathrm{P}(C=c,X=1)\mathrm{P}(C=c,X=0)}{\phi(\leftM)}.
\end{align*}
We rewrite \leftM-bias$(C=c,\textsc{rd})$ from Theorem \ref{thm:stratum-extension},
\begin{align*}
    \leftM\text{-bias}(C=c,\textsc{rd})
    &=\textsc{rd}_\text{left}\cdot\V\text{-bias-em}(C=c,\textsc{rd})\cdot\textsc{vr}(c)\\
    &=\textsc{rd}_\text{left}\cdot\frac{p_{A=1}p_{A=0}p_{Y=1}p_{Y=0}g(c)}{\mathrm{P}(C=c,A=1)\mathrm{P}(C=c,A=0)}\cdot\frac{\mathrm{P}(A=1\mid C=c)\mathrm{P}(A=0\mid C=c)}{\mathrm{P}(X=1\mid C=c)\mathrm{P}(X=0\mid C=c)}\\
    &=\textsc{rd}_\text{left}\cdot\frac{\textsc{var}_\text{left}\cdot\textsc{var}_\text{right}\cdot g(c)}{\mathrm{P}(C=c,X=1)\mathrm{P}(C=c,X=0)}.
\end{align*}
Combining these with the weights, we have
\begin{align*}
    \leftM\text{-bias}(\textsc{lm})
    &=\frac{\textsc{rd}_\text{left}\cdot\textsc{var}_\text{left}\cdot\textsc{var}_\text{right}\cdot\{\mathrm{P}(C=0)g(1)+\mathrm{P}(C=1)g(0)\}}{\phi(\leftM)}\\
    &=h(c)\cdot\textsc{rd}_\text{left}\cdot\textsc{var}_\text{left}\cdot\textsc{var}_\text{right}\cdot 1/\phi(\leftM).
\end{align*}

\noindent\underline{\leftlongM~structure}: By Lemma \ref{lm:3}, the weights that average \leftlongM-bias$(D=1,\textsc{rd})$ and \leftlongM-bias$(D=0,\textsc{rd})$ to \leftlongM-bias(\textsc{lm}) take the form
$$
w_{D=d}=\mathrm{P}(D=1-d)\mathrm{P}(D=d,X=1)\mathrm{P}(D=d,X=0)/\phi(\leftlongM).
$$
We rewrite \leftlongM-bias$(D=d,\textsc{rd})$ from Theorem \ref{thm:stratum-extension},
\begin{align*}
    \leftlongM\text{-bias}(D=d,\textsc{rd})
    &=\textsc{rd}_\text{left}\cdot\Y\text{-bias-em}(D=d,\textsc{rd})\cdot\textsc{vr}(d)\\
    &=\textsc{rd}_\text{left}\cdot\frac{p_{A=1}p_{A=0}p_{Y=1}p_{Y=0}(p_{d\mid1}-p_{d\mid0})\{p_{d\mid1}g(1)-p_{d\mid0}g(0)\}}{\mathrm{P}(D=d,A=1)\mathrm{P}(D=d,A=0)}\times\\
    &~~~~\frac{\mathrm{P}(A=1\mid D=d)\mathrm{P}(A=0\mid D=d)}{\mathrm{P}(X=1\mid D=d)\mathrm{P}(X=0\mid D=d)}\\
    &=\frac{\textsc{rd}_\text{left}\cdot\textsc{var}_\text{left}\cdot\textsc{var}_\text{right}\cdot(p_{d\mid1}-p_{d\mid0})\cdot\{p_{d\mid1}g(1)-p_{d\mid0}g(0)\}}{\mathrm{P}(D=d,A=1)\mathrm{P}(D=d,A=0)},
\end{align*}
where the last step is obtained by reasoning in a similar manner as in the proof for the \leftM~structure. Combining these with the weights, we have
\begin{align*}
    \leftlongM\text{-bias}(\textsc{lm})
    &=h(c)\cdot\textsc{rd}_\text{left}\cdot\textsc{var}_\text{left}\cdot\textsc{var}_\text{right}\cdot\textsc{rd}^2_\text{child}\cdot 1/\phi(\leftlongM).
\end{align*}

\noindent\underline{\rightM, \M, \rightlongM~and \longM~structures}: The proof for these structures is a simple extension of the results for the \V, \leftM, \Y~and \leftlongM~structures, respectively. We show it for the \rightM~structure as an example:
Theorem \ref{thm:stratum-extension} shows that conditioning on a level of $C$, \rightM-bias is equivalent to the embedded \V-bias times $\textsc{rd}_\text{right}$. The weights that average the $C$-stratum-specific \rightM-bias to \rightM-bias(\textsc{lm}) are the same weights used for \V-bias-em(\textsc{lm}), as they involve the same variables $X$ and $C$. Also, by definition, $\phi(\rightM)$ is the same as $\phi(\V)$ from the embedded \V~structure. That means
$$
\rightM\text{-bias}(\textsc{lm})=\V\text{-bias-em}(\textsc{lm})\cdot\textsc{rd}_\text{right}=h(c)\cdot\textsc{var}_\text{left}\cdot\textsc{var}_\text{right}\cdot\textsc{rd}_\text{right}\cdot\frac{1}{\phi(\rightM)}.
$$
Similar reasoning gives
\begin{align*}
\M\text{-bias}(\textsc{lm})
&=\leftM\text{-bias-em}(\textsc{lm})\cdot\textsc{rd}_\text{right}
=h(c)\cdot\textsc{var}_\text{left}\cdot\textsc{var}_\text{right}\cdot\textsc{rd}_\text{left}\cdot\textsc{rd}_\text{right}\cdot 1/\phi(\M),\\
\rightlongM\text{-bias}(\textsc{lm})
&=\Y\text{-bias-em}(\textsc{lm})\cdot\textsc{rd}_\text{right}
=h(c)\cdot\textsc{var}_\text{left}\cdot\textsc{var}_\text{right}\cdot\textsc{rd}_\text{right}\cdot\textsc{rd}^2_\text{child}\cdot 1/\phi(\rightlongM),\\
    \longM\text{-bias}(\textsc{lm})
    &=\leftlongM\text{-bias-em}(\textsc{lm})\cdot\textsc{rd}_\text{right}
    =h(c)\cdot\textsc{var}_\text{left}\cdot\textsc{var}_\text{right}\cdot\textsc{rd}_\text{left}\cdot\textsc{rd}_\text{right}\cdot\textsc{rd}^2_\text{child}\cdot 1/\phi(\longM),
\end{align*}
where \leftM-bias-em(\textsc{lm}) refers to the bias of the \leftM~structure embedded in the \M~structure, and \leftlongM-bias-em(\textsc{lm}) refers to the bias of the \leftlongM~structure embedded in the \M~structure.
\end{proof}

\bigskip

\begin{proof}[\bf Detailed expressions for $\phi(\textup{structure})$ and their derivation]
\hfill

We first list these expressions before showing their derivation.

\begin{align*}
\phi(\V)&=p_{X=1}p_{X=0}\left\{
    \begin{matrix*}[l]
        p_{X=1}(p_{Y=1}p_{C=1\mid11}+p_{Y=0}p_{C=1\mid10})(p_{Y=1}p_{C=0\mid11}+p_{Y=0}p_{C=0\mid10})+\\
        p_{X=0}(p_{Y=1}p_{C=1\mid01}+p_{Y=0}p_{C=1\mid00})(p_{Y=1}p_{C=0\mid01}+p_{Y=0}p_{C=0\mid00})
    \end{matrix*}
    \right\},
\\
\\
\phi(\Y)&=p_{X=1}p_{X=0}\times
    \left\{
    \begin{matrix}
		p_{X=1}
		\times
		\begin{pmatrix*}[l]
			p_{Y=1}p_{C=1\mid11}p_{D=1\mid1}+\\
			p_{Y=1}p_{C=0\mid11}p_{D=1\mid0}+\\
			p_{Y=0}p_{C=1\mid10}p_{D=1\mid1}+\\
			p_{Y=0}p_{C=0\mid10}p_{D=1\mid0}
		\end{pmatrix*}
		\times
		\begin{pmatrix*}[l]
			p_{Y=1}p_{C=1\mid11}p_{D=0\mid1}+\\
			p_{Y=1}p_{C=0\mid11}p_{D=0\mid0}+\\
			p_{Y=0}p_{C=1\mid10}p_{D=0\mid1}+\\
			p_{Y=0}p_{C=0\mid10}p_{D=0\mid0}
		\end{pmatrix*}
		+\\
		p_{X=0}
		\times
		\begin{pmatrix*}[l]
			p_{Y=1}p_{C=1\mid01}p_{D=1\mid1}+\\
			p_{Y=1}p_{C=0\mid01}p_{D=1\mid0}+\\
			p_{Y=0}p_{C=1\mid00}p_{D=1\mid1}+\\
			p_{Y=0}p_{C=0\mid00}p_{D=1\mid0}
		\end{pmatrix*}
		\times
		\begin{pmatrix*}[l]
			p_{Y=1}p_{C=1\mid01}p_{D=0\mid1}+\\
			p_{Y=1}p_{C=0\mid01}p_{D=0\mid0}+\\
			p_{Y=0}p_{C=1\mid00}p_{D=0\mid1}+\\
			p_{Y=0}p_{C=0\mid00}p_{D=0\mid0}
		\end{pmatrix*}
    \end{matrix}
    \right\},
\\
\\
   \phi(\leftM)&=
    \mathrm{P}(X=1)\mathrm{P}(X=0)\mathrm{P}(C=1)\mathrm{P}(C=0)-\\
    &~~~~~~~~~~~~~~p_{A=1}^2p_{A=0}^2(p_{X=1\mid1}-p_{X=1\mid0})^2
    \{\mathrm{P}(C=1\mid A=1)-\mathrm{P}(C=1\mid A=0)\}^2
    \\
    &=\begin{pmatrix*}[l]
		p_{A=1}p_{X=1\mid1}+\\
		p_{A=0}p_{X=1\mid0}
    \end{pmatrix*}
    \times
    \begin{pmatrix*}[l]
		p_{A=1}p_{X=0\mid1}+\\
		p_{A=0}p_{X=0\mid0}
    \end{pmatrix*}
    \times
    \begin{pmatrix*}[l]
		p_{A=1}p_{Y=1}p_{C=1\mid11}+\\
		p_{A=1}p_{Y=0}p_{C=1\mid10}+\\
		p_{A=0}p_{Y=1}p_{C=1\mid01}+\\
		p_{A=0}p_{Y=0}p_{C=1\mid00}
    \end{pmatrix*}
    \times
    \begin{pmatrix*}[l]
		p_{A=1}p_{Y=1}p_{C=0\mid11}+\\
		p_{A=1}p_{Y=0}p_{C=0\mid10}+\\
		p_{A=0}p_{Y=1}p_{C=0\mid01}+\\
		p_{A=0}p_{Y=0}p_{C=0\mid00}
    \end{pmatrix*}
    -
    \\
    &~~~~~~p_{A=1}^2p_{A=0}^2(p_{X=1\mid1}-p_{X=1\mid0})^2
    \{p_{Y=1}(p_{C=1\mid11}-p_{C=1\mid01})+p_{Y=0}(p_{C=1\mid10}-p_{C=1\mid00})\}^2,
\\
\\
    \phi(\leftlongM)&=
    \mathrm{P}(X=1)\mathrm{P}(X=0)\mathrm{P}(D=1)\mathrm{P}(D=0)-\\
    &~~~~~p_{A=1}^2p_{A=0}^2(p_{X=1\mid1}-p_{X=1\mid0})^2
    \{\mathrm{P}(C=1\mid A=1)-\mathrm{P}(C=1\mid A=0)\}^2(p_{D=1\mid1}-p_{D=1\mid0})^2
    \\
    &=\begin{pmatrix*}[l]
		p_{A=1}p_{X=1\mid1}+\\
		p_{A=0}p_{X=1\mid0}
    \end{pmatrix*}
    \times
    \begin{pmatrix*}[l]
		p_{A=1}p_{X=0\mid1}+\\
		p_{A=0}p_{X=0\mid0}
    \end{pmatrix*}
    \times
    \\
    &~~~~~\left\{
    p_{D=1\mid1}
    \begin{pmatrix*}[l]
		p_{A=1}p_{Y=1}p_{C=1\mid11}+\\
		p_{A=1}p_{Y=0}p_{C=1\mid10}+\\
		p_{A=0}p_{Y=1}p_{C=1\mid01}+\\
		p_{A=0}p_{Y=0}p_{C=1\mid00}
    \end{pmatrix*}
    +
    p_{D=1\mid0}
    \begin{pmatrix*}[l]
		p_{A=1}p_{Y=1}p_{C=0\mid11}+\\
		p_{A=1}p_{Y=0}p_{C=0\mid10}+\\
		p_{A=0}p_{Y=1}p_{C=0\mid01}+\\
		p_{A=0}p_{Y=0}p_{C=0\mid00}
    \end{pmatrix*}
    \right\}
    \times
    \\
    &~~~~~\left\{
    p_{D=0\mid1}
    \begin{pmatrix*}[l]
		p_{A=1}p_{Y=1}p_{C=1\mid11}+\\
		p_{A=1}p_{Y=0}p_{C=1\mid10}+\\
		p_{A=0}p_{Y=1}p_{C=1\mid01}+\\
		p_{A=0}p_{Y=0}p_{C=1\mid00}
    \end{pmatrix*}
    +
    p_{D=0\mid0}
    \begin{pmatrix*}[l]
		p_{A=1}p_{Y=1}p_{C=0\mid11}+\\
		p_{A=1}p_{Y=0}p_{C=0\mid10}+\\
		p_{A=0}p_{Y=1}p_{C=0\mid01}+\\
		p_{A=0}p_{Y=0}p_{C=0\mid00}
    \end{pmatrix*}
    \right\}
    -
    \\
    &~~p_{A=1}^2p_{A=0}^2(p_{X=1\mid1}-p_{X=1\mid0})^2
    \left\{
    \begin{matrix*}[l]
		p_{Y=1}(p_{C=1\mid11}-p_{C=1\mid01})+\\
		p_{Y=0}(p_{C=1\mid10}-p_{C=1\mid00})
    \end{matrix*}
    \right\}^2
    (p_{D=1\mid1}-p_{D=1\mid0})^2.
\end{align*}
The expressions for $\phi(\rightM),\phi(\rightlongM),\phi(\M)$ and $\phi(\longM)$ are the same as $\phi(\V),\phi(\Y),\phi(\leftM)$ and $\phi(\leftlongM)$, respectively, except $Y$ is replaced by $B$.

\bigskip

Of the four distinct expressions above, $\phi(\V)$ has essentially been derived in the proof of Theorem \ref{thm:Vlm}. We now derive the other three expressions.

\bigskip

Derivation of $\phi(\Y)$  is simple. 
With the \Y~structure, by Lemma \ref{lm:4},
\begin{align*}
	\phi(\Y)
	&=\mathrm{P}(D=0)\mathrm{P}(D=1,X=1)\mathrm{P}(D=1,X=0)+\mathrm{P}(D=1)\mathrm{P}(D=0,X=1)\mathrm{P}(D=0,X=0)
    \\
    &=
    \mathrm{P}(X=0)\mathrm{P}(X=1,D=1)\mathrm{P}(X=1,D=0)+\mathrm{P}(X=1)\mathrm{P}(X=0,D=1)\mathrm{P}(X=0,D=0)
    ~~~\text{(by Lemma \ref{lm:4})}\\
    &=
    p^2_{X=1}p_{X=0}\mathrm{P}(D=1\mid X=1)\mathrm{P}(D=0\mid X=1)+p_{X=1}p^2_{X=0}\mathrm{P}(D=1\mid X=0)\mathrm{P}(D=0\mid X=0)
    \\
    &=p_{X=1}p_{X=0}\times
    \left\{
    \begin{matrix}
		p_{X=1}
		\times
		\begin{pmatrix*}[l]
			p_{Y=1}p_{C=1\mid11}p_{D=1\mid1}+\\
			p_{Y=1}p_{C=0\mid11}p_{D=1\mid0}+\\
			p_{Y=0}p_{C=1\mid10}p_{D=1\mid1}+\\
			p_{Y=0}p_{C=0\mid10}p_{D=1\mid0}
		\end{pmatrix*}
		\times
		\begin{pmatrix*}[l]
			p_{Y=1}p_{C=1\mid11}p_{D=0\mid1}+\\
			p_{Y=1}p_{C=0\mid11}p_{D=0\mid0}+\\
			p_{Y=0}p_{C=1\mid10}p_{D=0\mid1}+\\
			p_{Y=0}p_{C=0\mid10}p_{D=0\mid0}
		\end{pmatrix*}
		+\\
		p_{X=0}
		\times
		\begin{pmatrix*}[l]
			p_{Y=1}p_{C=1\mid01}p_{D=1\mid1}+\\
			p_{Y=1}p_{C=0\mid01}p_{D=1\mid0}+\\
			p_{Y=0}p_{C=1\mid00}p_{D=1\mid1}+\\
			p_{Y=0}p_{C=0\mid00}p_{D=1\mid0}
		\end{pmatrix*}
		\times
		\begin{pmatrix*}[l]
			p_{Y=1}p_{C=1\mid01}p_{D=0\mid1}+\\
			p_{Y=1}p_{C=0\mid01}p_{D=0\mid0}+\\
			p_{Y=0}p_{C=1\mid00}p_{D=0\mid1}+\\
			p_{Y=0}p_{C=0\mid00}p_{D=0\mid0}
		\end{pmatrix*}
    \end{matrix}
    \right\}.
\end{align*}

\bigskip

Derivation of $\phi(\leftM)$ is more complicated.
With the \leftM~structure,
\begin{align*}
    \phi(\leftM)
    &=\mathrm{P}(C=0)\mathrm{P}(C=1,X=1)\mathrm{P}(C=1,X=0)+\mathrm{P}(C=1)\mathrm{P}(C=0,X=1)\mathrm{P}(C=0,X=0)
    \\
    &=\mathrm{P}(C=0)\{p_{A=1}p_{X=1\mid1}\mathrm{P}(C=1\mid A=1)+p_{A=0}p_{X=1\mid0}\mathrm{P}(C=1\mid A=0)\}\times\\
    &~~~~~~~~~~~~~~~~~~~~~~~\{p_{A=1}p_{X=0\mid1}\mathrm{P}(C=1\mid A=1)+p_{A=0}p_{X=0\mid0}\mathrm{P}(C=1\mid A=0)\}+\\
    &~~~~~\mathrm{P}(C=1)\{p_{A=1}p_{X=1\mid1}\mathrm{P}(C=0\mid A=1)+p_{A=0}p_{X=1\mid0}\mathrm{P}(C=0\mid A=0)\}\times\\
    &~~~~~~~~~~~~~~~~~~~~~~~\{p_{A=1}p_{X=0\mid1}\mathrm{P}(C=0\mid A=1)+p_{A=0}p_{X=0\mid0}\mathrm{P}(C=0\mid A=0)\}.
\end{align*}
To simplify notation, we will abbreviate any probabilities not already abbreviated, e.g., $\mathrm{P}(C=c)$ is abbreviated as $P_{C=c}$, and $\mathrm{P}(C=c\mid A=a)$ is abbreviated as $P_{C=c\mid A=a}$. We use the upper case $P$ to differentiate this notation from the lower case $p$ used only to abbreviate marginal probabilities of an exogenous variable or conditional probabilities of an endogenous variable conditioning on all its parents. We continue working with the expression above.

\begin{align*}
	\phi(\leftM)    
    &=P_{C=0}\{p_{A=1}p_{X=1\mid1}P_{C=1\mid A=1}+p_{A=0}p_{X=1\mid0}P_{C=1\mid A=0}\}\cdot
    \{p_{A=1}p_{X=0\mid1}P_{C=1\mid A=1}+p_{A=0}p_{X=0\mid0}P_{C=1\mid A=0}\}+
    \\
    &~~~~~P_{C=1}\{p_{A=1}p_{X=1\mid1}P_{C=0\mid A=1}+p_{A=0}p_{X=1\mid0}P_{C=0\mid A=0}\}\cdot
    \{p_{A=1}p_{X=0\mid1}P_{C=0\mid A=1}+p_{A=0}p_{X=0\mid0}P_{C=0\mid A=0}\}
    \\
    &=(1-P_{C=1})\left\{
        \begin{matrix*}[l]
            p^2_{A=1}p_{X=1\mid1}p_{X=0\mid1}P^2_{C=1\mid A=1}+p^2_{X=0}p_{X=1\mid0}p_{X=0\mid0}P^2_{C=1\mid A=0}+\\
            p_{A=1}p_{A=0}(p_{X=1\mid1}p_{X=0\mid0}+p_{X=1\mid0}p_{X=0\mid1})P_{C=1\mid A=1}P_{C=1\mid A=0}
        \end{matrix*}
    \right\}+\\
    &~~~~~P_{C=1}\left\{
        \begin{matrix*}[l]
            p^2_{A=1}p_{X=1\mid1}p_{X=0\mid1}(1-2P_{C=1\mid A=1}+P^2_{C=1\mid A=1})+\\
            p^2_{A=0}p_{X=1\mid0}p_{X=0\mid0}(1-2P_{C=1\mid A=0}+P^2_{C=1\mid A=0})+\\
            p_{A=1}p_{A=0}(p_{X=1\mid1}p_{X=0\mid0}+p_{X=1\mid0}p_{X=0\mid1})
            (1-P_{C=1\mid A=1}-P_{C=1\mid A=0}+P_{C=1\mid A=1}P_{C=1\mid A=0})
        \end{matrix*}
    \right\}
    \\
    &=\left\{
        \begin{matrix*}[l]
            p^2_{A=1}p_{X=1\mid1}p_{X=0\mid1}P^2_{C=1\mid A=1}+
            p^2_{A=0}p_{X=1\mid0}p_{X=0\mid0}P^2_{C=1\mid A=0}+\\
            p_{A=1}p_{A=0}(p_{X=1\mid1}p_{X=0\mid0}+p_{X=1\mid0}p_{X=0\mid1})P_{C=1\mid A=1}P_{C=1\mid A=0}
        \end{matrix*}
    \right\}+\\
    &~~~~~P_{C=1}\left\{
        \begin{matrix*}[l]
            p^2_{A=1}p_{X=1\mid1}p_{X=0\mid1}+p^2_{A=0}p_{X=1\mid0}p_{X=0\mid0}+\\
            p_{A=1}p_{A=0}(p_{X=1\mid1}p_{X=0\mid0}+p_{X=1\mid0}p_{X=0\mid1})
        \end{matrix*}
    \right\}-\\
    &~~~~~P_{C=1}\left\{
        \begin{matrix*}[l]
            2p^2_{A=1}p_{X=1\mid1}p_{X=0\mid1}P^2_{C=1\mid A=1}+2p^2_{A=0}p_{X=1\mid0}p_{X=0\mid0}P^2_{C=1\mid A=0}+\\
            p_{A=1}p_{A=0}(p_{X=1\mid1}p_{X=0\mid0}+p_{X=1\mid0}p_{X=0\mid1})(P_{C=1\mid A=1}+P_{C=1\mid A=0})
        \end{matrix*}
    \right\}
    \displaybreak\\
    &=\left\{
        \begin{matrix*}[l]
            P^2_{C=1\mid A=1}p^2_{A=1}p_{X=1\mid1}p_{X=0\mid1}+
            P^2_{C=1\mid A=0}p^2_{A=0}p_{X=1\mid0}p_{X=0\mid0}+\\
            P_{C=1\mid A=1}P_{C=1\mid A=0}p_{A=1}p_{A=0}(p_{X=1\mid1}p_{X=0\mid0}+p_{X=1\mid0}p_{X=0\mid1})
        \end{matrix*}
    \right\}+\\
    &~~~~~P_{C=1}(p_{A=1}p_{X=0\mid1}+p_{A=0}p_{X=0\mid0})(p_{A=1}p_{X=1\mid1}+p_{A=0}p_{X=1\mid0})-\\
    &~~~~~P_{C=1}
        \begin{bmatrix*}[l]
            p_{A=1}P_{C=1\mid A=1}\{2p_{A=1}p_{X=1\mid1}p_{X=0\mid1}+p_{A=0}(p_{X=1\mid1}p_{X=0\mid0}+p_{X=1\mid0}p_{X=0\mid1})\}+\\
            p_{A=0}P_{C=1\mid A=0}\{2p_{A=0}p_{X=1\mid0}p_{X=0\mid0}+p_{A=1}(p_{X=1\mid1}p_{X=0\mid0}+p_{X=1\mid0}p_{X=0\mid1})\}
        \end{bmatrix*}
    \\
    &=\left\{
        \begin{matrix*}[l]
            P^2_{C=1\mid A=1}p^2_{A=1}p_{X=1\mid1}p_{X=0\mid1}+
            P^2_{C=1\mid A=0}p^2_{A=0}p_{X=1\mid0}p_{X=0\mid0}+\\
            P_{C=1\mid A=1}P_{C=1\mid A=0}p_{A=1}p_{A=0}(p_{X=1\mid1}p_{X=0\mid0}+p_{X=1\mid0}p_{X=0\mid1})
        \end{matrix*}
    \right\}+\\
    &~~~~~P_{C=1}P_{X=1}P_{X=0}-\\
    &~~~~~P_{C=1}
        \begin{bmatrix*}[l]
            p_{A=1}P_{C=1\mid A=1}\{p_{X=1\mid1}(p_{A=1}p_{X=0\mid1}+p_{A=0}p_{X=0\mid0})+p_{X=0\mid1}(p_{A=1}p_{X=1\mid1}+p_{A=0}p_{X=1\mid0})\}+\\
            p_{A=0}P_{C=1\mid A=0}\{p_{X=0\mid0}(p_{A=0}p_{X=1\mid0}+p_{A=1}p_{X=1\mid1})+p_{X=1\mid0}(p_{A=0}p_{X=0\mid0}+p_{A=1}p_{X=0\mid1})\}
        \end{bmatrix*}
    \\
    &=\left\{
        \begin{matrix*}[l]
            P^2_{C=1\mid A=1}p^2_{A=1}p_{X=1\mid1}p_{X=0\mid1}+
            P^2_{C=1\mid A=0}p^2_{A=0}p_{X=1\mid0}p_{X=0\mid0}+\\
            P_{C=1\mid A=1}P_{C=1\mid A=0}p_{A=1}p_{A=0}(p_{X=1\mid1}p_{X=0\mid0}+p_{X=1\mid0}p_{X=0\mid1})
        \end{matrix*}
    \right\}+\\
    &~~~~~P_{C=1}P_{X=1}P_{X=0}-\\
    &~~~~~
        \begin{pmatrix*}[l]
            p_{A=1}P_{C=1\mid A=1}+\\
            p_{A=0}P_{C=1\mid A=0}
        \end{pmatrix*}
        \left\{
        \begin{matrix*}[l]
            p_{A=1}P_{C=1\mid A=1}(p_{X=1\mid1}P_{X=0}+p_{X=0\mid1}P_{X=1})+\\
            p_{A=0}P_{C=1\mid A=0}(p_{X=0\mid0}P_{X=1}+p_{X=1\mid0}P_{X=0})
        \end{matrix*}
        \right\}\\
    &=\left\{
        \begin{matrix*}[l]
            P^2_{C=1\mid A=1}p^2_{A=1}p_{X=1\mid1}p_{X=0\mid1}+
            P^2_{C=1\mid A=0}p^2_{A=0}p_{X=1\mid0}p_{X=0\mid0}+\\
            P_{C=1\mid A=1}P_{C=1\mid A=0}p_{A=1}p_{A=0}(p_{X=1\mid1}p_{X=0\mid0}+p_{X=1\mid0}p_{X=0\mid1})
        \end{matrix*}
    \right\}+\\
    &~~~~~P_{C=1}P_{X=1}P_{X=0}-\\
    &~~~~~
        \begin{bmatrix*}[l]
            P^2_{C=1\mid A=1}p^2_{A=1}(p_{X=1\mid1}P_{X=0}+p_{X=0\mid1}P_{X=1})+
            P^2_{C=1\mid A=0}p^2_{A=0}(p_{X=1\mid0}P_{X=0}+p_{X=0\mid0}P_{X=1})+\\
            P_{C=1\mid A=1}P_{C=1\mid A=0}p_{A=1}p_{A=0}\{(p_{X=1\mid1}+p_{X=1\mid0})P_{X=0}+(1-p_X=1\mid1-p_{X=1\mid0})P_{X=1}\}
        \end{bmatrix*}
    \\
    &=P_{C=1}P_{X=1}P_{X=0}+\\
    &~~~~~
        \begin{bmatrix*}[l]
            P^2_{C=1\mid A=1}p^2_{A=1}(p_{X=1\mid1}p_{X=0\mid1}-p_{X=1\mid1}P_{X=0}-p_{X=0\mid1}P_{X=1})+\\
            P^2_{C=1\mid A=0}p^2_{A=0}(p_{X=1\mid0}p_{X=0\mid0}-p_{X=1\mid0}P_{X=0}-p_{X=0\mid0}P_{X=1})+\\
            P_{C=1\mid A=1}P_{C=1\mid A=0}p_{A=1}p_{A=0}\times\\
            \{p_{X=1\mid1}p_{X=0\mid0}+p_{X=1\mid0}p_{X=0\mid1}-(p_{X=1\mid1}+p_{X=1\mid0})P_{X=0}-(2-p_{X=1\mid1}-p_{X=1\mid0})P_{X=1}\}
        \end{bmatrix*}
    \\
    &=P_{C=1}P_{X=1}P_{X=0}-\\
    &~~~~~
        \begin{bmatrix*}[l]
            P^2_{C=1\mid A=1}p^2_{A=1}\{p^2_{A=0}(p_{X=1\mid1}-p_{X=1\mid0})^2+P_{X=1}P_{X=0}\}+\\
            P^2_{C=1\mid A=0}p^2_{A=0}\{p^2_{A=1}(p_{X=1\mid1}-p_{X=1\mid0})^2+P_{X=1}P_{X=0}\}+\\
            P_{C=1\mid A=1}P_{C=1\mid A=0}p_{A=1}p_{A=0}\{-2p_{A=1}p_{A=0}(p_{X=1\mid1}-p_{X=1\mid0})^2+2P_{X=1}P_{X=0}\}
        \end{bmatrix*}
    \\
    &=P_{C=1}P_{X=1}P_{X=0}-\\
    &~~~~~P_{X=1}P_{X=0}(P^2_{C=1\mid A=1}p^2_{A=1}+P^2_{C=1\mid A=0}p^2_{A=0}+2P_{C=1\mid A=1}P_{C=1\mid A=0}p_{A=1}+p_{A=0})-\\
    &~~~~~p^2_{A=1}p^2_{A=0}(p_{X=1\mid1}-p_{X=1\mid0})^2(P^2_{C=1\mid A=1}+P^2_{C=1\mid A=0}-2P_{C=1\mid A=1}P_{C=1\mid A=0})\\
    &=P_{C=1}P_{X=1}P_{X=0}-\\
    &~~~~~P_{X=1}P_{X=0}(P_{C=1\mid A=1}p_{A=1}+P_{C=1\mid A=0}p_{A=0})^2-\\
    &~~~~~p^2_{A=1}p^2_{A=0}(p_{X=1\mid1}-p_{X=1\mid0})^2(P_{C=1\mid A=1}-P_{C=1\mid A=0})^2\\
    &=P_{C=1}P_{X=1}P_{X=0}-P_{X=1}P_{X=0}P_{C=1}^2-\\
    &~~~~~p^2_{A=1}p^2_{A=0}(p_{X=1\mid1}-p_{X=1\mid0})^2(P_{C=1\mid A=1}-P_{C=1\mid A=0})^2\\
    &=P_{X=1}P_{X=0}P_{C=1}P_{C=0}-p^2_{A=1}p^2_{A=0}(p_{X=1\mid1}-p_{X=1\mid0})^2(P_{C=1\mid A=1}-P_{C=1\mid A=0})^2\\
    &=\begin{pmatrix*}[l]
		p_{A=1}p_{X=1\mid1}+\\
		p_{A=0}p_{X=1\mid0}
    \end{pmatrix*}
    \times
    \begin{pmatrix*}[l]
		p_{A=1}p_{X=0\mid1}+\\
		p_{A=0}p_{X=0\mid0}
    \end{pmatrix*}
    \times
    \begin{pmatrix*}[l]
		p_{A=1}p_{Y=1}p_{C=1\mid11}+\\
		p_{A=1}p_{Y=0}p_{C=1\mid10}+\\
		p_{A=0}p_{Y=1}p_{C=1\mid01}+\\
		p_{A=0}p_{Y=0}p_{C=1\mid00}
    \end{pmatrix*}
    \times
    \begin{pmatrix*}[l]
		p_{A=1}p_{Y=1}p_{C=0\mid11}+\\
		p_{A=1}p_{Y=0}p_{C=0\mid10}+\\
		p_{A=0}p_{Y=1}p_{C=0\mid01}+\\
		p_{A=0}p_{Y=0}p_{C=0\mid00}
    \end{pmatrix*}
    -
    \\
    &~~~~~~p_{A=1}^2p_{A=0}^2(p_{X=1\mid1}-p_{X=1\mid0})^2
    \{p_{Y=1}(p_{C=1\mid11}-p_{C=1\mid01})+p_{Y=0}(p_{C=1\mid10}-p_{C=1\mid00})\}^2.
\end{align*}

\bigskip

To derive $\phi(\leftlongM)$, we build on the derivation of $\phi(\leftM)$. Note that the derivation of $\phi(\leftM)$ above, up to the step before the last step, has shown that
\begin{align*}
    \phi(\leftM)
    &=\mathrm{P}(C=0)\mathrm{P}(C=1,X=1)\mathrm{P}(C=1,X=0)+\mathrm{P}(C=1)\mathrm{P}(C=0,X=1)\mathrm{P}(C=1,X=0)
    \\
    &=P_{X=1}P_{X=0}P_{C=1}P_{C=0}-p^2_{A=1}p^2_{A=0}(p_{X=1\mid1}-p_{X=1\mid0})^2(P_{C=1\mid A=1}-P_{C=1\mid A=0})^2.
\end{align*}
With the \leftlongM~structure, similar reasoning (replacing $C$ with $D$) shows that
\begin{align*}
    \phi(\leftlongM)
    &=\mathrm{P}(D=0)\mathrm{P}(D=1,X=1)\mathrm{P}(D=1,X=0)+\mathrm{P}(D=1)\mathrm{P}(D=0,X=1)\mathrm{P}(D=1,X=0)
    \\
    &=P_{X=1}P_{X=0}P_{D=1}P_{D=0}-p^2_{A=1}p^2_{A=0}(p_{X=1\mid1}-p_{X=1\mid0})^2(P_{D=1\mid A=1}-P_{D=1\mid A=0})^2,
\end{align*}
which can be expanded as,
\begin{align*}
    \phi(\leftlongM)
    &=P_{X=1}P_{X=0}P_{D=1}P_{D=0}-p^2_{A=1}p^2_{A=0}(p_{X=1\mid1}-p_{X=1\mid0})^2\times\\
    &~~~~~(p_{D=1\mid1}P_{C=1\mid A=1}+p_{D=1\mid0}P_{C=0\mid A=1}-p_{D=1\mid1}P_{C=1\mid A=0}-p_{D=1\mid0}P_{C=0\mid A=0})^2\\
    &=P_{X=1}P_{X=0}P_{D=1}P_{D=0}-
    p^2_{A=1}p^2_{A=0}(p_{X=1\mid1}-p_{X=1\mid0})^2(P_{C=1\mid A=1}-P_{C=1\mid A=0})^2(p_{D=1\mid1}-p_{D=1\mid0})^2
    \\
    &=\begin{pmatrix*}[l]
		p_{A=1}p_{X=1\mid1}+\\
		p_{A=0}p_{X=1\mid0}
    \end{pmatrix*}
    \times
    \begin{pmatrix*}[l]
		p_{A=1}p_{X=0\mid1}+\\
		p_{A=0}p_{X=0\mid0}
    \end{pmatrix*}
    \times
    \\
    &~~~~~\left\{
    p_{D=1\mid1}
    \begin{pmatrix*}[l]
		p_{A=1}p_{Y=1}p_{C=1\mid11}+\\
		p_{A=1}p_{Y=0}p_{C=1\mid10}+\\
		p_{A=0}p_{Y=1}p_{C=1\mid01}+\\
		p_{A=0}p_{Y=0}p_{C=1\mid00}
    \end{pmatrix*}
    +
    p_{D=1\mid0}
    \begin{pmatrix*}[l]
		p_{A=1}p_{Y=1}p_{C=0\mid11}+\\
		p_{A=1}p_{Y=0}p_{C=0\mid10}+\\
		p_{A=0}p_{Y=1}p_{C=0\mid01}+\\
		p_{A=0}p_{Y=0}p_{C=0\mid00}
    \end{pmatrix*}
    \right\}
    \times
    \\
    &~~~~~\left\{
    p_{D=0\mid1}
    \begin{pmatrix*}[l]
		p_{A=1}p_{Y=1}p_{C=1\mid11}+\\
		p_{A=1}p_{Y=0}p_{C=1\mid10}+\\
		p_{A=0}p_{Y=1}p_{C=1\mid01}+\\
		p_{A=0}p_{Y=0}p_{C=1\mid00}
    \end{pmatrix*}
    +
    p_{D=0\mid0}
    \begin{pmatrix*}[l]
		p_{A=1}p_{Y=1}p_{C=0\mid11}+\\
		p_{A=1}p_{Y=0}p_{C=0\mid10}+\\
		p_{A=0}p_{Y=1}p_{C=0\mid01}+\\
		p_{A=0}p_{Y=0}p_{C=0\mid00}
    \end{pmatrix*}
    \right\}
    -
    \\
    &~~~~~p_{A=1}^2p_{A=0}^2(p_{X=1\mid1}-p_{X=1\mid0})^2
    \left\{
    \begin{matrix*}[l]
		p_{Y=1}(p_{C=1\mid11}-p_{C=1\mid01})+\\
		p_{Y=0}(p_{C=1\mid10}-p_{C=1\mid00})
    \end{matrix*}
    \right\}^2
    (p_{D=1\mid1}-p_{D=1\mid0})^2.
\end{align*}
\end{proof}

\end{document}